\documentclass{CSML}

\def\dOi{12(4:13)2016}
\lmcsheading%
{\dOi}
{1--52}
{}
{}
{May\phantom.~14, 2014}
{Dec.~31, 2016}
{}

\ACMCCS{[{\bf Theory of computation}]: Logic---Automated reasoning}

\usepackage[T1]{fontenc}
\usepackage{relsize}
\usepackage{amsmath}
\usepackage{amssymb}
\usepackage{arydshln}
\usepackage{booktabs}
\usepackage{bussproofs}
\usepackage{calc}
\usepackage{cite}
\usepackage{stmaryrd}
\usepackage{subfigure}
\usepackage[usenames]{color}
\usepackage{units}
\usepackage{url}
\usepackage[all]{xy}

\usepackage{hyperref}

\newcommand\Paper{article}
\newcommand\betweenaxs{1.5\smallskipamount}
\newcommand\betweentf{1\smallskipamount}
\newcommand\lvminusbaselineskip{0pt}

\theoremstyle{plain}

\newenvironment{acknowledgements}{\section*{Acknowledgement}}{}

\newcommand\omicron{\ensuremath{o}}

\newcommand\betwtypax{0pt}

\renewcommand\phi{\varphi}

\newcommand\VDOTS{\hbox{}\smash{\lower.5ex\hbox{\kern1.5em\vdots}}}
\newcommand\www{\ensuremath{\mathit{b}}}

\newcommand\LQ{``}\newcommand\RQ{''}

\newcommand\jjj{\ensuremath\vartheta}

\newcommand\betweenfigs{\vskip4.5\smallskipamount} 

\newcommand\afterDot{\;}
\newcommand\zooh{\vphantom{$A^{B^{C}}$}}

\EnableBpAbbreviations

\DeclareSymbolFont{letters}{OML}{txmi}{m}{it} 

\makeatletter

\newcommand\krn{\kern.125ex}

\def\pmbx#1{\setbox0=\hbox{#1}%
\kern+.00625em\raise.015em\box0\kern-\wd0\kern-.00625em}

\newbox\hackBoxB

\newcommand\tsfont[1]{\textsf{\upshape#1}}
\let\tsfontsuffix=\tsfont
\newcommand\mono[1]{\ensuremath{\widetilde{#1}}}
\newcommand\monobf[1]{\setbox\hackBoxB=\hbox{\mono#1}%
  \lower.05ex\hbox{\pmbx{\hbox{\mono{\color{white}{#1}}}}}\kern-\wd\hackBoxB\usebox\hackBoxB}

\newcommand\monom{\lower.8ex\hbox{\mono{\hphantom{m}}}}
\newcommand\erased{\mbox{\tsfont{\relsize{-1}\krn e\krn}}}
\newcommand\args{\mbox{\tsfont{\relsize{-1}\krn a\krn}}}
\newcommand\filtertxt{\textsuperscript{\vthinspace\itshape\filt}}
\newcommand\filter{\ensuremath{^\filt}}
\newcommand\guards{\mbox{\tsfont{\relsize{-1}\krn g\krn}}}
\newcommand\tags{\mbox{\tsfont{\relsize{-1}\krn t\krn}}}
\newcommand\native{\mbox{\tsfont{\relsize{-1}\krn n\krn }}}
\newcommand\query{\text{\relsize{-1}\tsfontsuffix{?}\@}}
\newcommand\qquery{\text{\relsize{-1}\tsfontsuffix{?\kern-.1ex?}\@}} 
\newcommand\at{\text{\relsize{-1}\tsfontsuffix{\kern-.1ex@\kern-.1ex}\@}} 

\newcommand\hAPP{\const{hAPP}}
\newcommand\hBOOL{\const{hBOOL}}
\newcommand\ti{\const{t}}
\def\is{\const{g}}

\newcommand{\typ}[1]{^{\kern.15ex\smash{#1}}} 

\def\a{\alpha}
\def\b{\beta}
\def\t{\sigma}
\def\u{\tau}
\def\bool{\omicron}

\def\vvthinspace{\kern+0.0416667em}
\def\vthinspace{\kern+0.083333em}
\def\negvthinspace{\kern-0.083333em}
\def\negvvthinspace{\kern-0.0416667em}
\newcommand\dash{\vvthinspace---\vvthinspace}

\def\vvthinspace{\kern+0.04em}

\newcommand\eq{\approx}
\def\const#1{\mbox{\relsize{-1}\textsf{#1}}}  

\newcommand\PHI{\ensuremath{\mathrm{\Phi}}}

\DeclareFontFamily{OT1}{pzc}{}
\DeclareFontShape{OT1}{pzc}{m}{it}{<-> s * [1.10] pzcmi7t}{}
\DeclareMathAlphabet{\mathcal}{OT1}{pzc}{m}{it}

\newcommand\shortsect{Section}
\newcommand\longsect{Section}

\def\LAN#1\RAN{\mskip.5mu\langle#1\rangle\mskip-.5mu}
\def\LANX#1\RAN{\mskip.5mu\langle#1\rangle\mskip-.5mu}
\def\LANY#1\RAN{\mskip.5mu\langle#1\rangle\mskip-.5mu}

\def\llan{\kern-.1ex\langle\kern-.55ex\langle} 
\def\rran{\rangle\kern-.55ex\rangle\kern-.1ex}

\newcommand\lenfilt[2]{\left|{#2}_{#1}\right|}

\newcommand\EQ{\hfill${} \,::=\, {}$}
\newcommand\OR{\hfill$|\;\;$}
\newcommand\syntaxarrayspec{@{}p{1em}@{}p{2em}@{}p{9.5em}@{\quad\enskip}l@{}}

\newcommand\tinycomma{\mskip+1mu;\mskip+1mu}

\newcommand\encode[2]{\smash{\displaystyle\llan{#1}\rran_{\smash{#2}\vphantom{tg}}}}

\newcommand\typex[1]{\llan{#1}\rran}
\newcommand\erasedx[1]{\encode{#1}{\erased}}
\newcommand\argsx[1]{\encode{#1}{\args}}
\newcommand\argsfx[1]{\encode{#1}{\args\filter}}

\newcommand\guardsx[1]{\encode{#1}{\guards}}
\newcommand\tagsx[1]{\encode{#1}{\tags}}

\newcommand\mguardsqx[1]{\encode{#1}{\mono\guards\query}}

\newcommand\mtagsqx[1]{\encode{#1}{\mono\tags\query}}
\newcommand\mtagsqqx[1]{\encode{#1}{\mono\tags\qquery}}

\newcommand\guardsqx[1]{\encode{#1}{\guards\query}}

\newcommand\guardsax[1]{\encode{#1}{\guards\at}}
\newcommand\tagsqx[1]{\encode{#1}{\tags\query}}
\newcommand\tagsqqx[1]{\encode{#1}{\tags\qquery}}
\newcommand\tagsax[1]{\encode{#1}{\tags\at}}

\newcommand\CT[2]{\lfloor#2\rfloor_{#1}}

\renewcommand\to{\shortrightarrow}
\renewcommand{\leftrightarrow}{\leftarrow\kern-1.75ex\rightarrow} 

\newcommand\filt{\tsfont{\relsize{-1}\textsl{x}}}
\renewcommand\inf{\tsfont{\relsize{-1}inf}}
\newcommand\ninf{\tsfont{\relsize{-1}ninf}}
\newcommand\phan{\tsfont{\relsize{-1}phn}}
\newcommand\ctor{\tsfont{\relsize{-1}ctr}}

\newcommand\gfilt[3]{{#2}_{#1}(#3)}
\newcommand\INF[2]{\gfilt{#1}{\text{\inf}}{#2}}
\newcommand\NINF[2]{\gfilt{#1}{\text{\ninf}}{#2}}
\newcommand\PHAN[2]{\gfilt{#1}{\phan}{#2}}

\newcommand\FILT[2]{\gfilt{#1}{\filt}{#2}}

\newcommand\SIGMA{\mathrm\Sigma}

\newcommand\MONO[1]{#1\mathrel\rhd}
\newcommand\NONMONO[1]{#1\not\mathrel\rhd}

\newcommand\NV{\mathrm{NV}}
\newcommand\UV{\mathrm{UV}}

\newcommand\Cover[1]{\mathrm{Cover\vvthinspace}_{#1}}
\newcommand\Inf{\mathrm{Inf}}
\newcommand\InfX{\mathrm{Inf}^*}
\newcommand\Ctor{\mathrm{Ctor}}

\newcommand\Pp{\mathcal{P}\mskip1.5mu'}

\newcommand\xx{\ensuremath{\kern+.2ex{\tsfont{\kern-.2ex{\slshape\relsize{-1} x}\kern.2ex}}}}
\newcommand\PHIii[1]{\encode{\PHI}{#1}}

\newcommand\PROP[1]{\hbox to 5.1em{\hfill\sc #1}\rm}

\renewcommand\labelitemi{\lower.1375ex\hbox{\large\textbullet}}

\newcommand{\betweenitems}{\vskip2\smallskipamount}

\def\win#1#2#3{\!{\bfseries#1#2#3}\!}

\def\midrule{\noalign{\ifnum0=`}\fi
  \@aboverulesep=\aboverulesep
  \global\@belowrulesep=\belowrulesep
  \global\@thisruleclass=\@ne
  \@ifnextchar[{\@BTrule}{\@BTrule[\lightrulewidth]}}
\def\bottomrule{\noalign{\ifnum0=`}\fi
  \@aboverulesep=\aboverulesep
  \global\@belowrulesep=\belowbottomsep
  \global\@thisruleclass=\@ne
  \@ifnextchar[{\@BTrule}{\@BTrule[\lightrulewidth]}}

\newcommand\ourfrac[2]{\frac{\,#1\,\strut}{\,#2\,\strut}}
\newcommand\btwhyps{\kern1.0em}
\newcommand\btwrules{\kern1.5em}

\newcommand\SEM[1]{\llbracket#1\rrbracket}

\newcommand\dom{\mathbb{D}}

\renewcommand\AX{\textrm{Ax}}

\def\tabularStuff{@{\kern.1ex}ll@{\kern1.5em}rrr@{\kern1em}rrrr@{\kern1em}rrrr@{\kern1em}r@{\kern.1ex}}

\newcommand\ismonkey{\is{}_{{\vthinspace\mathit{monkey}}}}

\let\oldS=\S
\renewcommand\S{\oldS\vthinspace}

\newcommand\QED{\qed}

\newcommand\IUS{\negvthinspace\raise.15ex\hbox{$\scriptsize\_$}\negvthinspace}

\newcommand\cmono{
\ensuremath{\small\infty}} 
\newcommand\cmonox[1]{\encode{#1}{\cmono}}
\newcommand\xmono{\mono{\xx}} 

\newcommand\leftOut[1]{}

\newcommand\Ax{\mathit{Ax}}
\newcommand\D{{D}} 

\newcommand\epssym{\varepsilon}
\newcommand\eps[1]{\epssym(#1)}

\newcommand\TVars{\mathrm{TVars}}
\newcommand\FTVars{\mathrm{FTVars}}
\newcommand\Vars{\mathrm{Vars}}
\newcommand\FVars{\mathrm{FVars}}

\newcommand\TypesOf{\Type}
\newcommand\GInst{\textrm{GInst}}
\newcommand\mgi{\textrm{mgi}}
\newcommand\Type{\mathit{Type}}
\newcommand\GType{\mathit{GType}}
\newcommand\ra{\rightarrow}
\newcommand\AAA{\mathcal A}
\newcommand\VV{\mathcal V}
\newcommand\Vun{\VV^{\,\forall}}
\newcommand\Vex{\VV^{\,\exists}}
\newcommand\Vt{\VV_{\textrm{typed}}}
\newcommand\Vtun{\Vt^{\,\forall}}
\newcommand\Vtex{\Vt^{\,\exists}}
\newcommand\MM{\mathcal M}
\newcommand\NN{\MM'}
\newcommand\KK{\mathcal K}
\newcommand\FF{\mathcal F}
\newcommand\PP{\mathcal P}
\newcommand\ov{\bar}
\newcommand\ff{\ensuremath{f}}
\newcommand\pp{\ensuremath{p}}

\let\QED=\relax

\renewcommand\EQ{\hfill$\;\,{::=}\,\;$}
\renewcommand\OR{\hfill$|\;\;\;$}
\renewcommand\syntaxarrayspec{@{}p{1em}@{}p{2.25em}@{}p{9.5em}@{\quad\enskip}l@{}}

\begin{document}

\title[Encoding Monomorphic and Polymorphic Types]{Encoding Monomorphic and Polymorphic Types}

\author[J.~C.~Blanchette]{Jasmin Christian Blanchette\rsuper a}
\address{{\lsuper a}Inria \& LORIA, Nancy, France; Max-Planck-Institut f\"ur Informatik, Saarbr\"ucken, Germany}
\email{jasmin.blanchette@\{inria.fr,mpi-inf.mpg.de\}}

\author[S.~B\"ohme]{Sascha B\"ohme\rsuper b}
\address{{\lsuper b}Fakult\"at f\"ur Informatik, Technische Universit\"at M\"unchen, Germany}
\email{boehmes@in.tum.de}

\author[A.~Popescu]{Andrei Popescu\rsuper c}
\address{{\lsuper c}Department of Computer Science, School of Science and Technology, Middlesex University, UK}
\email{a.popescu@mdx.ac.uk}

\author[N.~Smallbone]{Nicholas Smallbone\rsuper d}
\address{{\lsuper d}Dept.\ of CSE,\ Chalmers University of Technology, Gothenburg, Sweden}
\email{nicsma@chalmers.se}

\bibliographystyle{spmpsci}

\maketitle

\begin{abstract}
Many automatic theorem provers are restricted to untyped logics,
and existing translations from typed logics are bulky or unsound.
Recent research proposes monotonicity as a means to remove some clutter
when translating monomorphic to untyped first-order logic.
Here we pursue this approach systematically, analysing formally a variety of encodings
that further improve on efficiency while retaining soundness and completeness.
We extend the approach to rank-1 polymorphism and
present alternative schemes that lighten the translation of
polymorphic symbols based on the novel notion of ``cover''.
The new encodings are implemented in Isabelle\slash HOL\@ as part of the
Sledgehammer tool.
We include informal proofs of soundness and correctness, and have formalised
the monomorphic part of this work in Isabelle\slash HOL\@.
Our evaluation finds the new encodings vastly superior to previous
schemes.
\end{abstract}

\setcounter{footnote}{0}

\section{Introduction}
\label{sec:introduction}

Specification languages, proof assistants, and other theorem proving
applications are typically based on polymorphism formalisms,
but state-of-the-art automatic provers support only untyped or monomorphic logics.
The existing sound (proof-reflecting) and complete (proof-preserving) translation schemes for
polymorphic types, whether they revolve around functions (tags) or predicates (guards),
produce clutter that severely hampers the proof search \cite{meng-paulson-2008-trans},
and lighter approaches based on type arguments are unsound
\cite{meng-paulson-2008-trans,stickel-1986}. As a result, application authors face a difficult
choice between soundness and efficiency when interfacing with automatic provers.

The fourth author, together with Claessen and Lilliestr\"om
\cite{claessen-et-al-2011}, 
designed a pair of
sound, complete, and efficient
translations from monomorphic many-typed to untyped first-order logic with equality.
The key insight is that \emph{monotonic} types\dash types whose domain can 
be extended with new elements while preserving satisfiability\dash can be
merged. The remaining types can be made monotonic by introducing
suitable protectors. 

\begin{exa}[Monkey Village]\afterDot
\label{ex:monkey-village}%
Imagine a village of monkeys \cite{claessen-et-al-2011}
where each monkey owns at least two bananas.
The predicate $\const{owns} : \mathit{monkey} \times \mathit{banana} \to \bool$
(where $\bool$ denotes truth values)
associates monkeys with bananas, and the functions $\const{b}_1$, $\const{b}_2 :
\mathit{monkey} \to \mathit{banana}$ witness the existence of each monkey's
minimum supply of bananas:\strut
\begin{quotex}
$\forall M \mathbin: \mathit{monkey}.\;\, \const{owns}(M, \const{b}_1(M)) \mathrel{\land} \const{owns}(M, \const{b}_2(M))$ \\
$\forall M \mathbin: {\mathit{monkey}}.\;\, \const{b}_1(M) \not\eq \const{b}_2(M)$ \\
$\forall M_{1}, {M_2} \mathbin: {\mathit{monkey}},\; B \mathbin: {\textit{banana}}.\;\,
\const{owns}(M_{1}, B) \mathrel{\land} \const{owns}(M_2, B)
\rightarrow M_{1\!} \eq M_2$
\end{quotex}
The axioms are clearly satisfiable.
\end{exa}

In the monkey village of Example~\ref{ex:monkey-village},
the type \textit{banana} is monotonic,
because any model with $b$ bananas can be extended
to a model with $b' > b$ bananas. In contrast,
\textit{monkey} is nonmonotonic,
because there can live at most $\lfloor b/2\rfloor$ monkeys in a village with a
finite supply of $b$~bananas.
Syntactically, the monotonicity of \textit{banana} is inferable from the
absence of a positive equality $B \eq t$
or $t \eq B$, where $B$ is a variable of type \textit{banana} and $t$ is
arbitrary; such a literal would be needed to make the type nonmonotonic.

The example can be encoded as follows, using
the predicate $\ismonkey$ to guard against ill-typed instantiations
of $M$, $M_1$, and $M_2$:
\begin{quotex}
$\exists M.\;\, \ismonkey(M)$ \\[\betweenaxs]
$\forall M.\;\, \ismonkey(M) \rightarrow \const{owns}(M, \const{b}_1(M)) \mathrel{\land} \const{owns}(M, \const{b}_2(M))$ \\
$\forall M.\;\, \ismonkey(M) \rightarrow \const{b}_1(M) \not\eq \const{b}_2(M)$ \\
$\forall M_{1}, {M_2}, B.\;
\ismonkey(M_{1\!}) \mathrel\land
\ismonkey(M_2) \mathrel\land
\const{owns}(M_{1}, B) \mathrel{\land} \const{owns}(M_2, B) \rightarrow M_{1\!} \mathbin\eq M_{2\!}$
\end{quotex}
The first axiom states the existence of a monkey;
this is necessary for completeness, since model carriers are required to be nonempty.
Thanks to monotonicity, it is sound to omit all type information regarding
bananas. The intuition behind this is that the $\ismonkey$ predicate makes the
problem fully monotonic, and for such problems it it possible to synchronise the cardinalities of the
different types and to merge the types, yielding an equisatisfiable untyped (or singly typed) problem.
For example, a model $\MM$ of the typed problem with $m$ monkeys and $b$ bananas will give rise to
a model $\NN$ of the untyped problem with $b$ ``bananamonkeys,'' among which $m$ are monkeys
according to the interpretation of $\ismonkey$; conversely, from a model of the untyped problem
with $b$ ``bananamonkeys'' including $m$ values for which $\ismonkey$ is true (with $m > 0$ thanks to
the first axiom), it is easy to construct a model of the typed problem with $b$ bananas and $m$
monkeys.

Monotonicity is not decidable, but it can often be inferred using suitable calculi.
In this \Paper, we exploit this idea systematically, analysing a variety of encodings
based on monotonicity:\
some are minor adaptations of existing ones, while others are novel encodings
that further improve on the size of the translated formulae.

In addition, we
generalise the monotonicity approach to a rank-1 polymorphic logic, as
embodied by the
typed first-order form TFF1
of the TPTP (Thousand of Problems for Theorem Provers)
\cite{blanchette-paskevich-2013}, a de facto standard implemented by a number
of reasoning tools.
Unfortunately, the presence of a single equality literal
$\mathit{X} \eq t$ or $t \eq \mathit{X}$,
where $\mathit{X}$ is a polymorphic variable of type $\a$, will lead the analysis
to classify all types as possibly nonmonotonic and force the use of protectors
everywhere, as in the traditional encodings.
A typical example is the list axiom
$ \forall \mathit{X} \mathbin: {\a},\; \mathit{Xs} \mathbin: \mathit{list}(\a).\;\,
  \const{hd}(\const{cons}(\mathit{X}\!, \mathit{Xs})) \eq \mathit{X}$.
We solve this issue through a
novel scheme that reduces the clutter associated
with nonmonotonic types,
based on the observation that protectors are required
only when translating 
the particular formulae that prevent a type from being inferred monotonic.
This contribution improves the monomorphic case as well:\ for the
monkey village example, our scheme detects that the first two axioms are
harmless and translates them without the
$\ismonkey$ guards. (In fact, by appealing to a more general
notion of monotonicity, it is possible to eliminate all type information in the
monkey village problem in a sound fashion.)

Encoding types in an untyped logic is an old problem,
and several solutions have been proposed in the literature.
We start by reviewing four main
traditional approaches (\shortsect~\ref{sec:traditional-type-encodings}), which
prepare the ground for the more advanced encodings presented in this \Paper.
Next, we present improvements of the traditional encodings that
aim at reducing the clutter associated with polymorphic symbols, based on the
novel notion of ``cover''
(\shortsect{}~\ref{sec:alternative-cover-based-encoding-of-polymorphism}).
Then we move our attention to monotonicity-based encodings, which try to
minimise the number of added tags or guards.
We first present known and novel monotonicity-based schemes that handle only ground types
(\shortsect{}~\ref{sec:monotonicity-based-type-encodings-the-monomorphic-case});
these are interesting in their own right and serve as stepping stones for the full-blown polymorphic encodings
(\shortsect{}~\ref{sec:complete-monotonicity-based-encoding-of-polymorphism}).
Proofs of correctness accompany the descriptions of the new encodings. The
proofs explicitly relate models of unencoded and encoded problems.

Figure~\ref{fig:main-encodings} presents a brief overview of the main encodings.
The traditional encodings are identified by single letters (\erased{} for full type erasure, \args{} for type arguments,
\tags{} for type tags, \guards{} for type guards). The nontraditional encodings append a suffix to the letter:\
\at{} (= cover-based), \query{}~(=~monotonicity-based, lightweight),
or \qquery{} (= monotonicity-based, feather\-weight).
The decoration \smash{\monom{}} identifies
the monomorphic version of an encoding. Among the nontraditional schemes,
\mono\tags\query{} and \mono\guards\query{} are due to Claessen et al.\
\cite{claessen-et-al-2011}; the other encodings are novel.

\begin{figure}[t!]
\centerline{\newcommand\HX[1]{\hbox to 6.25em{#1\hfill}}%
\begin{tabular}{@{\kern.1ex}l@{\kern1.5em}l@{\kern1.25em}l@{\kern1em}l@{\kern1.25em}l@{\kern.1ex}}
& \HX{Traditional} & \HX{Cover-based} & \multicolumn{2}{@{}l@{}}{Monotonicity-based} \\[-.3ex] 
& \HX{(Polymorphic)} & \HX{(Polymorphic)} & \HX{Monomorphic} & \HX{Polymorphic} \\
\midrule
Full type erasure
  & \erased{} (\S\ref{ssec:full-type-erasure}) \\
Type arguments
  & \args{} (\S\ref{ssec:type-arguments})
    &
      & 
        \\
Type tags
  & \tags{} (\S\ref{ssec:type-tags})
    & \tags\at{} (\S\ref{ssec:cover-based-type-tags})
      & \mono\tags\query{}, \mono\tags\qquery{} (\S\ref{ssec:monotonicity-based-type-tags-monomorphic})
        & \tags\query{}, \tags\qquery{} (\S\ref{ssec:monotonicity-based-type-tags-polymorphic})
          \\
Type guards
  & \guards{} (\S\ref{ssec:type-guards})
    & \guards\at{} (\S\ref{ssec:cover-based-type-guards})
      & \mono\guards\query{}, \mono\guards\qquery{} (\S\ref{ssec:monotonicity-based-type-guards-monomorphic})
        & \HX{\guards\query{}, \guards\qquery{} (\S\ref{ssec:monotonicity-based-type-guards-polymorphic})}
\end{tabular}}%

\caption{\vthinspace The main encodings}%
\label{fig:main-encodings}%
\end{figure}

A formalisation \cite{blanchette-frocos2013} of the monomorphic part of our results has been developed
in the proof assistant
Isabelle\slash HOL \cite{nipkow-et-al-2002,nipkow-klein-2014}.
The encodings have been implemented in
Sledgehammer
\cite{meng-paulson-2008-trans,blanchette-et-al-2013-smt},
which provides a bridge between Isabelle\slash HOL
and automatic theorem provers (\shortsect~\ref{sec:implementation}).
They were evaluated with Alt-Ergo, E, SPASS, Vampire, and Z3 on a
benchmark suite consisting of proof goals from existing Isabelle
formalisations (\shortsect~\ref{sec:evaluation}).
Our comparisons include the traditional encodings as well as the provers' native
support for monomorphic types where it is available. Related work is considered
at the end (\shortsect~\ref{sec:related-work}).

We presented an earlier version of this work at the TACAS 2013 conference
\cite{blanchette-et-al-2013-types}. The current text extends the conference paper
with detailed proofs and a discussion of implementational issues
(\shortsect~\ref{sec:implementation}). It also corrects the
side conditions of two encoding definitions, which resulted in an unexpected
incompleteness. 

\begin{conv} \label{con-general}
Given a name, such as $t$, that ranges over a certain domain, such as the set of terms,
its overlined version $\ov{t}$ ranges over lists (tuples) of items in this domain, and $\left|\smash{\ov{t}}\right|$ denotes the
length of $\ov{t}$. We also write $\left|\smash{A}\right|$ for the cardinal of a set $A$.
A set is \emph{countable} if it is finite or countably infinite.
Given an element $u$, $(u)^n$ or $u^n$ denotes the list consisting of $n$ occurrences of $u$.
Given a nonempty set $A$, we write $\eps{A}$ for some arbitrary but fixed element of $A$. We use
$\eps{A}$ in definitions where the choice of the element does not matter, as long as it belongs to $A$.
\end{conv}

\section{Background: Logics}
\label{sec:background-logics}

This \Paper{} involves three versions of classical first-order logic with
equality:\ polymorphic, monomorphic, and untyped. They correspond to the TPTP
syntaxes TFF1 \cite{blanchette-paskevich-2013}, TFF0 \cite{sutcliffe-et-al-2012-tff},
and FOF \cite{sutcliffe-tptp}, respectively, excluding interpreted arithmetic.

\subsection{Polymorphic First-Order Logic}
\label{ssec:polymorphic-first-order-logic}

The source logic is a rank-1 polymorphic logic as specified 
by TFF1
\cite{blanchette-paskevich-2013}.

We fix $\mathcal A$, a countably infinite set of \emph{type variables} with typical
element $\a$, and $\mathcal V$, a countably infinite set of
\emph{term variables} with typical element $\mathit{X}$.

\begin{defi}[Syntax]\afterDot
A \emph{polymorphic signature} is a triple $\SIGMA = (\mathcal K, \mathcal F\!, \mathcal P)$,
where $\mathcal K$ is a countable 
set of type constructors $k$ with arities,
$\mathcal F$ is a countable
set of function symbols~$\ff$ with arities,
and $\mathcal P$ is a
countable
set of predicate symbols $\pp$ with arities.

For type constructors $k \in \KK$, the arity is a natural number $n$. This
association is written $k :: n$. Types, forming the set $\Type_\KK$, are then
defined inductively starting with type variables and applying type constructors
according to their arities:

\medskip

$\begin{array}{\syntaxarrayspec}
\multicolumn{4}{@{}l@{}}{\strut\hbox{\em Types:}} \\
\noindent\hfill$\t$ & \EQ & $k(\bar\t)$\quad where $k :: |\bar\t|$ & \text{constructor type} \\
& \OR & $\a$ & \text{type variable}
\end{array}$

\medskip

For function symbols $\ff \in \FF$, the arity is a
triple $(\ov{\alpha},\ov{\sigma},\sigma)$, where $\ov{\alpha}$ is
a list of distinct type variables, $\ov{\sigma}$ is a list of types, and $\sigma$ is a type such that all the variables appearing
in $\ov{\sigma}$ and $\sigma$ are among the ones in $\ov{\alpha}$. We write this association as
$\ff \, : \,\forall\bar\a.\; \ov{\sigma} \to \t$ or
$\ff \, : \,\forall\bar\a.\; \sigma_1 \times \cdots \times \sigma_n \to \t$.
Finally, for predicate symbols $\pp \in \PP$, the arity is a pair $(\ov{\alpha},\ov{\sigma})$, where $\ov{\alpha}$ is
a list of distinct type variables and $\ov{\sigma}$ is a list of types. We write this association as
$\pp \, : \,\forall\bar\a.\; \ov{\sigma} \to \bool$ or
$\pp \, : \,\forall\bar\a.\; \sigma_1 \times \cdots \times \sigma_n \to \bool$.

Overloading of different arities for the same symbol is excluded by definition.
The symbols $\forall$, $\times$, $\to$, and $\bool$ (omicron) are not type constructors but syntax.
The arity declarations for function and predicate symbols can be seen as instances of the general syntax
$s : \,\forall\bar\a.\; \bar\t \to \varsigma$,
where $s \in \mathcal{F} \mathrel{\uplus} \mathcal{P}$ and $\varsigma$ is either a type or $\bool$.
An application of~s will require $\left|\smash{\bar\a}\right|$ type arguments (written in angle brackets) and
$\left|\smash{\bar\t}\right|$ term arguments (written in parentheses).
In examples, we may 
omit type arguments from $\ov{\alpha}$
that are irrelevant or clear from the context.

A {\em typed} ({\em term\/}) {\em variable} is a pair of a variable and a type, written $\mathit{X}^\sigma$.
We fix $\Vt$, the set of type variables. The terms and formulae are defined below.

\medskip
$\begin{array}{\syntaxarrayspec}
\multicolumn{4}{@{}l@{}}{\strut\hbox{\em Terms:}} \\
\hfill$t$ & \EQ & $f\LAN\bar\t\RAN(\bar t\,)$ & \text{function term} \\
  & \OR & $\mathit{X}^\sigma$ & \text{typed variable}
\end{array}$

\smallskip

$\begin{array}{\syntaxarrayspec}
\multicolumn{4}{@{}l@{}}{\strut\hbox{\em Formulae:}} \\
\hfill$\phi$
        & \EQ & $p\LAN\bar\t\RAN(\bar t\,) \;\;|\;\; \lnot\,p\LAN\bar\t\RAN(\bar t\,)$ & \text{predicate literal} \\
        & \OR & $t_{1\!} \eq t_2 \;\;|\;\; t_{1\!} \not\eq t_2$ & \text{equality literal} \\
        & \OR & $\phi_{1\!} \mathrel\land \phi_2 \;\;|\;\;
                 \phi_{1\!} \mathrel\lor \phi_2$ & \text{binary connective} \\
        & \OR & $\forall \mathit{X} \mathbin: \t.\;\, \phi \;\;|\;\;
                 \exists \mathit{X} \mathbin: \t.\;\, \phi$ & \text{term quantification} \\
        & \OR & $\forall \a.\; \phi$ & \text{type quantification}
\end{array}$

\medskip

All type quantification is universal.
We also impose the following syntactic restriction,
not captured in the formation rules:\ type quantifiers
may only occur at the top of the formula, i.e.\ {not} underneath a term quantifier or a connective,
and formulae are in negation normal form (NNF), which is most suitable for
defining and reasoning about translations. We sometimes
use implication~$\phi_1 \mathrel\land \cdots \mathrel\land \phi_m
\rightarrow \psi_1 \mathrel\lor \cdots \mathrel\lor \psi_n$ as an abbreviation for
$\lnot\,\phi_1 \mathrel\lor \cdots \mathrel\lor \lnot\,\phi_m
\mathrel\lor \psi_1 \mathrel\lor \cdots \mathrel\lor \psi_n$ to enhance readability.
In examples, we freely nest quantifiers and connectives.

\begin{defi}[Free and Fresh Variables, Groundness, and Sentences]\afterDot
\label{def-vars}
$\TVars(\sigma)$, $\TVars(t)$, and $\FTVars(\phi)$ denote the sets of
type variables occurring in the type $\sigma$,
occurring in the term $t$, and occurring freely in the formula $\phi$,
respectively.
Similarly,
$\Vars(t)$ and $\FVars(\phi)$ denote the sets of typed term variables occurring in
the term $t$ and occurring freely in the formula $\phi$, respectively. For example,
$\TVars(\mathit{X}^\sigma) = \TVars(\sigma)$,
$\Vars(\mathit{X}^\sigma) = \{\mathit{X}^\sigma\}$,
$\FVars(\forall \mathit{X} \mathbin: \sigma.\;\phi) = \FVars(\exists \mathit{X} \mathbin: \sigma.\;\phi) = \FVars(\phi) - \{\mathit{X}^\sigma\}$.
A type $\sigma$ is {\em ground} if $\TVars(\sigma) = \emptyset$. We let $\GType_\Sigma$ denote the set of ground types.
A term $t$ is {\em ground} if $\TVars(t) = \emptyset$ and $\Vars(t) = \emptyset$.
A formula $\phi$ is a {\em sentence} if $\TVars(\phi) = \emptyset$ and
$\FVars(\phi) = \emptyset$.
\end{defi}

\begin{conv} \label{con-formulae}
We assume that the set of variables $\VV$ is partitioned in two infinite sets:\
$\Vun$, of \emph{universal} variables, and
$\Vex$, of \emph{existential} variables. The former
are the only ones allowed to be universally quantified, and the latter are the only ones allowed to be existentially quantified.
We let $\Vtun$ for the set of typed variable $\mathit{X}^\sigma$ with $\mathit{X} \in \Vun$; and similarly for $\Vtex$.
We also assume that each type or term
variable is bound only once in a formula, and that if the typed variable $\mathit{X}^\sigma$
appears in the scope of a quantification $\forall \mathit{X} \mathbin:\sigma'$ or $\exists \mathit{X} \mathbin:\sigma'$, then $\sigma = \sigma'$.
The last assumption allows us to omit the superscript $\sigma$ from $\mathit{X}^\sigma$ in examples.
\end{conv}


\end{defi}


The typing rules and semantics of the logic are modelled after those of TFF1.
Briefly, the type arguments completely
determine the types of the term arguments and, for functions, of the result.
Polymorphic symbols are interpreted as families of functions or predicates
indexed by domains corresponding to ground types. All types are inhabited (nonempty).

\begin{defi}[Type Substitution]\afterDot
A \emph{type substitution} $\rho$ is a function that maps every type variable $\alpha$ to a type, written $\alpha\vthinspace\rho$.
It is extended from type variables to types and type tuples in the standard way:\
$k(\ov{t})\vthinspace\rho = k (\ov{t}\vthinspace\rho)$ and $(t_1,\ldots,t_n)\vthinspace\rho = (t_1\vthinspace\rho,\ldots,t_n\vthinspace\rho)$.
\end{defi}

\begin{defi}[Typing Rules]\afterDot
%
%
A judgement
$t : \t$ expresses that the term $t$
is {\em well typed\/} and has type $\t$.  
A judgement 
$\phi : \bool$ expresses that the formula
$\phi$ is {\em well typed\/}.  
The \emph{typing rules} of polymorphic first-order logic are given below:\strut

\vskip\abovedisplayskip

\centerline{$\displaystyle \ourfrac{}{ \mathit{X}^\sigma : \sigma}$%
\btwrules
$\displaystyle
\ourfrac{f : \forall\bar\a.\; \bar\t \to \t
\btwrules
 t_{\negvthinspace j} : \t_{\negvthinspace j}\,\rho \text{\enskip for all~} j
}{
f\LAN\bar\a\vthinspace\rho\RAN(\bar t\,) : \t\vthinspace\rho}
$}

\medskip

\centerline{$\displaystyle
\ourfrac{p : \forall\bar\a.\; \bar\t \to \bool
\btwrules
 t_{\negvthinspace j} : \t_{\negvthinspace j}\,\rho \text{\enskip for all~} j
}{
p\LAN\bar\a\vthinspace\rho\RAN(\bar t\,) : \bool}
$
\btwrules
$\displaystyle
\ourfrac{p : \forall\bar\a.\; \bar\t \to \bool
\btwrules
 t_{\negvthinspace j} : \t_{\negvthinspace j}\,\rho \text{\enskip for all~} j
}{
\lnot\,p\LAN\bar\a\vthinspace\rho\RAN(\bar t\,) : \bool}
$}

\medskip

\centerline{%
$\displaystyle
\ourfrac{ t_1 : \t \btwhyps  t_2 : \t}{
 t_1 \eq t_2 : \bool}$%
\btwrules
$\displaystyle
\ourfrac{ t_1 : \t \btwhyps  t_2 : \t}{
 t_1 \not\eq t_2 : \bool}$%
\btwrules
$\displaystyle
\ourfrac{ \phi_1 : \bool \btwhyps  \phi_2 : \bool}{
 \phi_1 \mathrel\land \phi_2 : \bool}$%
}

\medskip

\centerline{%
$\displaystyle
\ourfrac{ \phi_1 : \bool \btwhyps  \phi_2 : \bool}{
 \phi_1 \mathrel\lor \phi_2 : \bool}$%
\btwrules
$\displaystyle
\ourfrac{ \phi : \bool}{
 \forall \mathit{X}\! \mathbin{:} \t .\; \phi \,:\, \bool}%
\btwrules
\ourfrac{ \phi : \bool}{
 \exists \mathit{X}\! \mathbin{:} \t .\; \phi \,:\, \bool}%
\btwrules
\ourfrac{ \phi : \bool}{
 \forall \a .\; \phi \,:\, \bool}$}

\vskip\belowdisplayskip
\end{defi}

\medskip

\begin{defi}[Problem]\label{def-problem}
A \emph{problem} is a set of (well-typed) sentences.
\end{defi}

\begin{lem}
For all terms $t$, there exists at most one type $\sigma$ such that $t : \sigma$.
\end{lem}
\begin{proof}
Immediate by induction on $t$.\QED
\end{proof}

\begin{conv}
We write $t\typ\sigma$ to indicate that the term $t$
is well typed and its (unique) type is $\sigma$. This convention is consistent with
the notation $\mathit{X}^\sigma$ for term variables of type $\sigma$. 
\end{conv}

\begin{defi}[Semantics]\afterDot
Let $\SIGMA = (\mathcal K, \mathcal F\!, \mathcal P)$ be a
polymorphic signature. A \emph{structure} $\MM$ for $\SIGMA$ is a tuple of families
$\bigl(\dom,(k^\MM)_{k \in \KK},(\ff^{\,\MM})_{\ff \in \FF},(\pp^\MM)_{\pp \in \PP}\bigr)$, where
\begin{itemize}
\item $\dom$ is a nonempty collection of nonempty sets called the \emph{domains};
\item if $k :: n$, then $k^\MM : \dom^n \ra \dom$. Given a \emph{type variable valuation} $\theta : \mathcal A \to \dom$,
this induces an \emph{interpretation} of types $\SEM{\phantom{i}}^\MM_\theta$ defined by the equations
$\SEM{k(\bar\t)}^\MM_\theta \,=\, k^\MM(\SEM{\bar\t}^\MM_\theta)$
and
$\SEM{\a}^\MM_\theta \,=\, \theta(\a)$;%
\hfill\hbox{}
\item if $\ff \, : \,\forall \alpha_1,\ldots,\alpha_m.\; \t_1 \times \cdots \times \t_{\!n} \to \t$,
then $\ff^{\,\MM} : \prod_{\ov{D} \in \dom^m}
  \SEM{\sigma_1}^\MM_{\theta_{\ov{D}}} \times \cdots \times \SEM{\sigma_n}^\MM_{\theta_{\ov{D}}} \ra \SEM{\sigma}^\MM_{\theta_{\ov{D}}}$,
where $\theta_{\ov{D}}$ is a type variable valuation mapping each $\alpha_i$ to $D_i$;
\item if $\pp \, : \,\forall \alpha_1,\ldots,\alpha_m.\; \t_1 \times \cdots \times \t_{\!n} \to \bool$,
then $\pp^\MM \subseteq \prod_{\ov{D} \in \dom^m}
  \SEM{\sigma_1}^\MM_{\theta_{\ov{D}}} \times \cdots \times \SEM{\sigma_n}^\MM_{\theta_{\ov{D}}}$,
where $\theta_{\ov{D}}$ is as above.
\end{itemize}
%
As expected from the arity of $\ff$, the interpretation $\ff^\MM$ is a function first taking $m$ type arguments and then $n$
data (element) arguments. Similarly, the interpretation $\pp^\MM$ is a predicate respecting the arity of $\pp$.

Given a type variable valuation $\theta$ and a compatible \emph{term variable
valuation} $\xi : {\mathcal V} \to \prod_{\sigma \in \Type} \SEM{\sigma}^\MM_\theta$, the \emph{interpretation} of terms and
formulae by the structure $\MM$ is as follows:
\newcommand\stuff{^\MM_{\theta,\xi}}
\newcommand\stuffT{^\MM_{\theta}}
\begin{align*}
\SEM{f\LAN\bar\t\RAN(\bar t\,)}\stuff
    & \,=\, f^{\vthinspace \mathcal M}(\SEM{\bar\t}\stuffT)\, (\SEM{\bar t\,}\stuff) &
\SEM{\mathit{X}^\sigma}\stuff
    & \,=\, \xi(\mathit{X})(\sigma) \\[\betweentf]
\SEM{p\LAN\bar\t\RAN(\bar t\,)}\stuff
    & \,=\, p^{\vthinspace \mathcal M}(\SEM{\bar\t}\stuffT)\, (\SEM{\bar t\,}\stuff) &
\SEM{t_1 \eq t_2}\stuff
    & \,=\, (\SEM{t_1}\stuff = \SEM{t_2}\stuff) \\
\SEM{\lnot\,p\LAN\bar\t\RAN(\bar t\,)}\stuff
    & \,=\, \lnot\,p^{\vthinspace \mathcal M}(\SEM{\bar\t}\stuffT)\, (\SEM{\bar t\,}\stuff) &
\SEM{t_1 \not\eq t_2}\stuff
    & \,=\, (\SEM{t_1}\stuff \not= \SEM{t_2}\stuff) \\
\SEM{\phi_1 \mathrel\land \phi_2}\stuff
    & \,=\, \SEM{\phi_1}\stuff \mathrel\land \SEM{\phi_2}\stuff &
\SEM{\forall \mathit{X} \mathbin: \t.\; \phi}\stuff
    & \,=\, \forall a \in \SEM{\t}\stuffT.\; \SEM{\phi}^\MM_{\theta,\xi\smash{[\mathit{X}\mapsto a]}} \\
\SEM{\phi_1 \mathrel\lor \phi_2}\stuff
    & \,=\, \SEM{\phi_1}\stuff \mathrel\lor \SEM{\phi_2}\stuff &
\SEM{\exists \mathit{X} \mathbin: \t.\; \phi}\stuff
    & \,=\, \exists a \in \SEM{\t}\stuffT.\; \SEM{\phi}^\MM_{\theta,\xi\smash{[\mathit{X}\mapsto a]}} \\
\SEM{\forall\a.\; \phi}\stuff
    & \,=\, \forall D \in \dom.\; \SEM{\phi}^\MM_{\theta\smash{[\a\mapsto D],\xi}}
\end{align*}
We omit irrelevant subscripts to $\SEM{\phantom{i}}$, writing $\SEM{\t}^\MM$ if
$\t$ is ground and $\SEM{\phi}^\MM$ if $\phi$ is a sentence.

A structure $\mathcal M$ is a \emph{model} of a problem $\PHI$ if
$\SEM{\phi}^{\mathcal M}$ is true for every $\phi \in \PHI$. A problem that has
a model is \emph{satisfiable}.
\end{defi}

\begin{exa}[Algebraic Lists]\afterDot
\label{ex:algebraic-lists}%
The following axioms induce a minimalistic first-order theory of algebraic
lists that will serve as our main running example:
\begin{quotex}
$\forall \a.\;
 \forall \mathit{X} \mathbin: \a,\; \mathit{Xs} \mathbin: \mathit{list}(\a).\;\,
  \const{nil} \not\eq \const{cons}(\mathit{X}\!, \mathit{Xs})$ \\
$\forall \a.\;
 \forall \mathit{Xs} \mathbin: \mathit{list}(\a).\;\,
  \mathit{Xs} \eq \const{nil} \mathrel{\lor}
  (\exists \mathit{Y} \mathbin: \a,\; \mathit{Ys} \mathbin: \mathit{list}(\a).\;\,
  \mathit{Xs} \eq \const{cons}(\mathit{Y}\!, \mathit{Ys})
  )$ \\
$\forall \a.\;
 \forall \mathit{X} \mathbin: {\a},\; \mathit{Xs} \mathbin: \mathit{list}(\a).\;\,
  \const{hd}(\const{cons}(\mathit{X}\!, \mathit{Xs})) \eq \mathit{X} \mathrel\land
  \const{tl}(\const{cons}(\mathit{X}\!, \mathit{Xs})) \eq \mathit{Xs}$
\end{quotex}
We conjecture that \const{cons} is injective.
The conjecture's negation can be expressed employing an unknown but fixed Skolem type $\www$:
\begin{quotex}
$\exists \mathit{X}\!, \mathit{Y} \mathbin: \www,\; \mathit{Xs}, \mathit{Ys} \mathbin: \mathit{list}(\www).\;\,
\const{cons}(\mathit{X}\!, \mathit{Xs}) \eq \const{cons}(\mathit{Y}\!, \mathit{Ys})
\mathrel\land (\mathit{X} \not\eq \mathit{Y} \mathrel{\lor} \mathit{Xs} \not\eq \mathit{Ys})$
\end{quotex}
Because the \const{hd} and \const{tl} equations force injectivity of
\const{cons} in both arguments, the problem consisting of the three axioms and the negated
conjecture is unsatisfiable. The conjecture is a consequence of the axioms.
\end{exa}

We are interested in encoding polymorphic problems $\PHI$ in a manner that
preserves and reflects their satisfiability.  It will be technically convenient to
assume that their signatures have at least one nullary type constructor, so that the set of ground types is
nonempty. It is obvious that this assumption is harmless:\ if it is not satisfied, we simply extend
the signature with a distinguished nullary type constructor $\iota :: 0$. Since $\iota$ does not appear in the formulae of $\PHI$
and since in models the set of domains is assumed to be nonempty, this signature extension does not
affect its satisfiability:\ given a model of $\PHI$ in the original signature, we obtain one in the extended signature
by interpreting $\iota$ as an arbitrary domain; given a model in the extended signature, we obtain one in the
original signature by omitting the interpretation of~$\iota$.

\begin{conv}
\label{con-iota}
For all polymorphic signatures $\Sigma = (\KK,\FF,\PP)$ that we consider, we assume that $\KK$ contains
at least one nullary type constructor.
\end{conv}

The following lemma shows that, in structures, the collection of domains can be
regarded as a copy of the ground types.

\begin{lem} \afterDot
\label{lem-syntactic-domains}
If a polymorphic $\Sigma$-problem $\PHI$ has a model, it also has a model $\MM =
\bigl(\dom,\allowbreak
(k^\MM)_{k \in \KK},\allowbreak\_,\_\bigr)$ such that the following conditions are met\/{\rm:}
\begin{enumerate}
%
%
\item each $k^\MM$ is injective, and \smash{$k^\MM(\ov{D}) \not= {k'}^\MM(\ov{E})$} whenever $k \not= k'${\rm;}
\item $\dom = \{\SEM{\tau}^\MM \mid \tau \in \GType_\Sigma\}${\rm;}
\item the type interpretation function\/ $\SEM{\phantom{i}}^\MM$ is a bijection between\/ $\GType$ and\/ $\dom${\rm;}
\item $\dom$ is countable\/{\rm;}
\item $\dom$ is disjoint from each $D \in \dom$, and any distinct $D_1,D_2 \in \dom$ are disjoint.
\end{enumerate}
\end{lem}
\begin{proof}
%
%
Assume $\PHI$ has a model $\MM$.
%
%
To prove (1), we construct a model $\NN$ from
$\MM$ by tagging the domains with types, i.e.\ by defining $\dom' = \{ D \times \{\sigma\} \mid D \in \dom' \;\wedge\; \sigma \in \GType_\Sigma\}$
and maintaining types across the application of type constructors:
$k^{\NN}(D_1 \times \{\sigma_1\},\ldots,D_n \times \{\sigma_n\}) = k^{\MM}(D_1,\ldots,D_n) \times \{k(\sigma_1,\ldots,\sigma_n)\}$.
The interpretations of the function and predicate symbols are adjusted accordingly.
It is easy to prove that $\MM$ and $\NN$ satisfy the same polymorphic formulae; in particular, $\NN$ is a model of $\PHI$.
For (2), 
we define a model $\NN' = (\dom'',(k^{\NN'})_{k \in \KK},\_,\_)$ from $\NN$
by taking $\dom'' \subseteq \dom'$ to be the image of $\SEM{\phantom{i}}^{\NN} : \GType_\Sigma \ra \dom'$
and by taking $k^{\NN'}$ and $s^{\NN'}$ to be the restrictions of $k^{\NN}$ and $s^{\NN}$. Thanks to
Convention \ref{con-iota}, $\dom''$ is nonempty, and moreover
$k^{\NN'}$ is well defined on $\dom$.
Again, it is easy to prove that $\NN'$ satisfies all the polymorphic formulae that
$\NN$ satisfies. In particular, $\NN'$ is a model of $\PHI$.
Points (3) and (4) follow from (1) and (2).
Finally, we prove (5) by replacing the domains with disjoint copies of them.
\QED
\end{proof}

\subsection{Monomorphic First-Order Logic}
\label{ssec:monomorphic-first-order-logic}

Monomorphic first-order logic, more commonly known as many-sorted first-order logic and corresponding
to TPTP TFF0 \cite{sutcliffe-et-al-2012-tff},
has signatures $\Sigma = (\Type,\FF,\PP)$, where $\Type$ is a countable set of types (or sorts)
ranged over by $\sigma$,
$\FF$ is a countable set of function symbols $f : \ov{\sigma} \ra \sigma$ with arities, and
$\PP$ is a countable set of predicate symbols $p : \ov{\sigma} \ra \bool$ with arities.
$\Sigma$-structures
$\MM = \bigl((\D_\sigma)_{\sigma \in \Type},(\ff^{\,\MM})_{\ff \in \FF},(\pp^\MM)_{\pp \in \PP}\bigr)$ interpret the types as sets and the function and predicate symbols
as functions and predicates of the suitable arities.
Given a model $\MM$ and a valuation $\xi : \VV \ra \smash{\prod_{\sigma \in \Type} \D_\sigma}$, the interpretations of terms and formulae,
$\SEM{t}^\MM_{\xi}$ and $\SEM{\phi}^\MM_{\xi}$,
are defined as expected.

Monomorphic first-order logic can be viewed as a special case of polymorphic
first-order logic, with a polymorphic signature considered monomorphic
when all its type constructors are nullary and the arities of its function and predicate symbols
contain no type variables.

%

\begin{exa}\afterDot
\label{ex:algebraic-lists-monomorphised}%
A monomorphised version of the algebraic list problem of Example~\ref{ex:algebraic-lists},
with $\a$ instantiated by $\www$, follows:
\begin{quotex}
$\forall \mathit{X} \mathbin: \www,\; \mathit{Xs} \mathbin: \mathit{list\IUS}\www.\;\,
  \const{nil}_{\www} \not\eq \const{cons}_{\vvthinspace\www}(\mathit{X}\!, \mathit{Xs})$ \\
$\forall \mathit{Xs} \mathbin: \mathit{list\IUS}\www.\;\,
  \mathit{Xs} \eq \const{nil}_{\www} \mathrel{\lor}
  (\exists \mathit{Y} \mathbin: \www,\; \mathit{Ys} \mathbin: \mathit{list\IUS}\www.\;\,
  \mathit{Xs} \eq \const{cons}_{\vvthinspace\www}(\mathit{Y}\!, \mathit{Ys})
  )$ \\
$\forall \mathit{X} \mathbin: {\www},\; \mathit{Xs} \mathbin: \mathit{list\IUS}\www.\;\,
  \const{hd}_{\www}(\const{cons}_{\vvthinspace\www}(\mathit{X}\!, \mathit{Xs})) \eq \mathit{X} \mathrel\land
  \const{tl}_{\www}(\const{cons}_{\vvthinspace\www}(\mathit{X}\!, \mathit{Xs})) \eq \mathit{Xs}$ \\
$\exists \mathit{X}\!, \mathit{Y} \mathbin: \www,\; \mathit{Xs}, \mathit{Ys} \mathbin: \mathit{list\IUS}\www.\;\,
\const{cons}_{\vvthinspace\www}(\mathit{X}\!, \mathit{Xs}) \eq \const{cons}_{\vvthinspace\www}(\mathit{Y}\!, \mathit{Ys})
\mathrel\land (\mathit{X} \not\eq \mathit{Y} \mathrel{\lor} \mathit{Xs} \not\eq \mathit{Ys})$
\end{quotex}
Like the original polymorphic problem, it is unsatisfiable.
\end{exa}

\subsection{Untyped First-Order Logic}
\label{ssec:untyped-first-order-logic}

The final target logic for all our encodings, untyped first-order logic,
coincides with the TPTP first-order form FOF \cite{sutcliffe-tptp}.
This is the logic traditionally implemented in
automatic theorem provers and finite model finders.
An \emph{untyped signature} is a pair $\SIGMA =
(\mathcal F\!, \mathcal P)$, where $\mathcal F$ and $\mathcal P$ are countable sets of
function and predicate symbols with arities, where the arities are natural numbers.
The notation $s^n$ indicates that the symbol $s$ has arity $n$.
The untyped syntax is identical to that of the monomorphic
logic, except that variable terms do not contain types and quantification is written $\forall
\mathit{X}.\; \phi$ and $\exists \mathit{X}.\; \phi$.
%
The structures for $\SIGMA =
(\mathcal F\!, \mathcal P)$ are triples
$\MM = \smash{\bigl(\D,(\ff^{\,\MM})_{\ff \in \FF},(\pp^\MM)_{\pp \in \PP}\bigr)}$, where
$\D$ is the domain and $\ff^{\,\MM}$ and $\pp^{\,\MM}$ are
$n$-ary functions and predicates on $\dom$, with $n$ being the symbol's arity.

\subsection{Type Encodings}
\label{ssec:type-enc}

The type encodings discussed in this article are given by functions that
take problems $\PHI$ in a logic ${\mathcal L}$ to problems $\PHI'$ in a logic
${\mathcal L}'$, where ${\mathcal L}$ and ${\mathcal L}'$ are among the three
logics introduced above.

\begin{conv} \label{con-encodings}
Each of the considered encodings will be specified by the following data:
\begin{itemize}
\item a function that maps ${\mathcal L}$-signatures $\Sigma$ to ${\mathcal L}'$-signatures $\Sigma'$;
\item for all ${\mathcal L}$-signatures $\Sigma$ and problems $\PHI$ over $\Sigma$:
  \begin{itemize}
  \item a (possibly empty) set $\Ax_{\PHI}$ of sentences over $\Sigma$, the {\em axioms}, added by the translation;
  \item a function $\encode{\phantom{i}}{}$ between formulae over $\Sigma$ and formulae over $\Sigma'$,
the {\em formula translation}.
  \end{itemize}
\end{itemize}
The formula translation is typically based on a {\em term translation} $\encode{\phantom{i}}{}$.
The encoding $\encode{\PHI}{}$ of a problem $\PHI$ is given by the union between the axioms and the componentwise translations:
$\encode{\PHI}{} = \Ax_{\,\PHI} \,\cup\, \{\encode{\phi}{} \!\mid \phi \in \PHI\}$.
\end{conv}

Central to this \Paper{} are the notions of soundness and completeness of an encoding:

\begin{defi}[Correctness]\afterDot
An encoding as above is \emph{sound} for a class of problems $\mathcal C$ if satisfiability of $\PHI \in \mathcal C$
implies satisfiability of $\encode{\PHI}{}$; it is \emph{complete} for $\mathcal C$ if, given $\PHI \in \mathcal C$, satisfiability
of $\encode{\PHI}{}$ implies satisfiability of~$\PHI$; it is \emph{correct} for $\mathcal C$ if
it is both sound and complete (i.e.\ $\PHI$ and $\encode{\PHI}{}$ are equisatisfiable).
In case $\mathcal C$ is the class of all problems, we omit it and simply call the encoding sound, complete, or correct.
\end{defi}

\section{Traditional Type Encodings}
\label{sec:traditional-type-encodings}

There are four main traditional approaches to encoding polymorphic types:\ full
type erasure, type arguments, type tags, and type guards
\cite{enderton-1972,stickel-1986,wick-mccune-1989,meng-paulson-2008-trans}.
Before introducing them, we first establish some conventions that will be useful
throughout the article.

\begin{conv} \label{con-distinguished}
We will often need to extend signatures $\Sigma$ with one or more of the
following distinguished symbols. Whenever we employ them, we assume they are
fresh with respect to $\Sigma$ and have the indicated arities:
\begin{itemize}
\item a nullary type constructor $\jjj$; 
\item a function symbol $\ti : \forall\a.\; \a\to\a$;
\item a predicate symbol $\is : \forall\a.\; \a\to\bool$.
\end{itemize}
Terms of type $\jjj$ will be used to represent the types of $\Sigma$,
$\ti$ will be used to tag terms with type information,
and $\is$ will be used to guard formulae with type information.
\end{conv}

\begin{conv} \label{con-tvarAsVar}
Since the sets of type and term variables, $\AAA$ and $\VV$, are countably infinite, we can fix a function
from $\AAA$ to $\VV$, $\alpha \mapsto \VV(\alpha)$, such that
\begin{enumerate}
\item it is injective, i.e.\ $\VV(\alpha_1) = \VV(\alpha_2)$ implies $\alpha_1 = \alpha_2$;
\item it allows for an infinite supply of term variables that do not correspond to type variables,
i.e.\ the set of term variables not having the form $\VV(\alpha)$ is infinite;
\item each $\VV(\alpha)$ is a universal variable.
\end{enumerate}
This function can be used to encode types as terms. Thanks to (2),
we can safely assume that the source problems $\PHI$ do
not contain variables of the form $\VV(\alpha)$.
\end{conv}

\subsection{Full Type Erasure}
\label{ssec:full-type-erasure}

The easiest way to translate a typed problem into an untyped logic is to erase
all its type information, which means omitting
all type arguments, type quantifiers, and
types in term quantifiers.
We call this encoding \erased{}.

\begin{defi}[Full Erasure \erased{}]\afterDot%
\label{def-full-erasure-sv}%
The \emph{full type erasure} encoding \erased{} translates a polymorphic problem
over $\SIGMA = (\KK, \FF, \PP)$ into an untyped problem over $\SIGMA' = (\FF', \PP')$, where the symbols in
$\SIGMA'$ have the same term arities as in $\SIGMA$
(but without type arguments).  Thus, if
$\ff \, : \,\forall\ov{\alpha}.\; \sigma_1 \times \ldots \times \sigma_{\!n} \to \sigma$ and
$\pp \, : \,\forall\ov{\alpha}.\; \sigma_1 \times \ldots \times \sigma_{\!n} \to \bool$
are in $\FF$ and $\PP$, respectively, then $\ff$ and $\pp$ have arities $n$ in $\FF'$ and $\PP'$, respectively.
The encoding adds no axioms, and the term and formula translations
$\erasedx{\phantom{i}}$ are defined as follows:
\begin{align*}
\erasedx{f\LAN\bar\t\RAN(\bar t\,)} & \,=\, f(\erasedx{\bar t\,}) &
\erasedx{\mathit{X}^\sigma} &\,=\, \mathit{X}
\\[\betweentf]
\erasedx{p\LAN\bar\t\RAN(\bar t\,)} & \,=\, p(\erasedx{\bar t\,}) &
  \erasedx{\forall \mathit{X} \mathbin: \t.\;\, \phi} & \,=\, \forall \mathit{X}.\; \erasedx{\phi}
\\
\erasedx{\lnot\,p\LAN\bar\t\RAN(\bar t\,)} & \,=\, \lnot\,p(\erasedx{\bar t\,})
& \erasedx{\exists \mathit{X} \mathbin: \t.\;\, \phi} & \,=\, \exists \mathit{X}.\; \erasedx{\phi} \\
\erasedx{\forall \a.\; \phi} & \,=\, \erasedx{\phi}
\end{align*}
Here and elsewhere, we omit the trivial cases where the function is simply
applied to its subterms or subformulae, as in $\vthinspace\erasedx{\phi_1
\mathrel\land \phi_2} \vthinspace=\vthinspace \erasedx{\phi_1} \mathrel\land \erasedx{\phi_2}$.
Recall that, according to Section~\ref{ssec:type-enc},
the \erased{} translation of a problem $\PHI$ is simply the componentwise translation
of its formulae:\ $\erasedx{\PHI} \,=\,
\{\erasedx{\phi} \mid \phi\in\PHI\}$.
\end{defi}

%

%


\begin{exa}%
\label{ex:monkey-village-erased}%
Encoded using \erased{}, the monkey village axioms of
Example~\ref{ex:monkey-village} become
\begin{quotex}
$\forall M.\;\, \const{owns}(M, \const{b}_1(M)) \mathrel{\land} \const{owns}(M, \const{b}_2(M))$ \\
$\forall M.\;\, \const{b}_1(M) \not\eq \const{b}_2(M)$ \\
$\forall M_{1}, {M_2}, B.\;\,
\const{owns}(M_{1}, B) \mathrel{\land} \const{owns}(M_2, B)
\rightarrow M_{1\!} \eq M_2$
\end{quotex}
Like the original axioms, the encoded axioms are satisfiable:\ the requirement
that each monkey possesses two bananas of its own can be met by taking an
infinite domain (since $2k = k$ for any infinite cardinal $k$).
\end{exa}

However, full type erasure is generally
unsound in the presence of equality
because equality can be used to encode cardinality constraints on domains.
For example,
the axiom $\forall \mathit{U}
\mathbin: \mathit{unit}.\; \mathit{U} \eq \const{unity}$ forces the domain of
$\mathit{unit}$ to have only one element. Its erasure, $\forall \mathit{U}.\; \mathit{U} \eq
\const{unity}$, effectively restricts \emph{all} types to one element%
; a contradiction is derivable from any disequality $t \not\eq u$ or any pair of
clauses $\pp(\bar{t})$ and~$\lnot\,\pp(\bar u)$.
An expedient proposed by Meng and Paulson \cite[\S2.8]{meng-paulson-2008-trans},
which they implemented in Sledgehammer,
is to filter out all axioms of the form
$\forall \mathit{X} \mathbin: \t.\;\, \mathit{X} \eq \const{a}_{1\!} \mathrel{\lor} \cdots
\mathrel{\lor} \mathit{X} \eq \const{a}_n$, but this makes the translation incomplete
and generally does not suffice to prevent unsound cardinality reasoning.

An additional issue
with full type erasure is that it
confuses distinct monomorphic instances of polymorphic symbols. The formula
$\const{q}\LAN a\RAN(\const{f}\LANX a\RAN) \mathrel\land
 \lnot\,\const{q}\LAN b\RAN({\const{f}\LANX b\RAN})$
is satisfiable, but its type erasure
$\const q(\const{f}) \mathrel\land \lnot\, \const q(\const{f})$
is unsatisfiable.
A more intuitive example might be 
$N \not\eq 0 \rightarrow N > 0$, which we would expect to
hold for the natural number versions of $0$ and~$>$ but not for
integers or real numbers.


Nonetheless, full type erasure is complete, and this property will be useful
later.

\begin{thm}[Completeness of \erased]\afterDot
\label{thm:completeness-of-erased}%
Full type erasure 
is complete.
\end{thm}
\begin{proof}
From a model $\NN = \smash{\bigl(\D,(\ff^{\,\NN})_{\ff \in \FF'},(\pp^{\NN})_{\pp \in \PP'}\bigr)}$ of $\erasedx{\PHI}$, we construct a
structure $\MM = \smash{\bigl(\dom,(k^\MM)_{k \in \KK},(\ff^{\,\MM})_{\ff \in \FF},(\pp^\MM)_{\pp \in \PP}\bigr)}$ for the signature of
$\PHI$ by taking the same domain for all types and interpreting all instances
of each polymorphic symbol in the same way as~$\mathcal M$:
\begin{itemize}
\item $\dom = \{\D\}$, and $k^\MM$ maps everything to $\D$;
\item if $s \in \FF \mathrel{\uplus} \PP$, then $s^{\,\MM}(\ov{D})(\ov{d}) = s^{\,\NN}(\ov{d})$.
\end{itemize}
Given $\theta : \mathcal A \ra \dom$ and $\xi : {\mathcal V} \ra \prod_{\sigma \in \Type_\Sigma} \SEM{\sigma}^\MM_\theta$,
we define $\xi' : {\mathcal V} \ra D$ by $\xi'(\mathit{X}) = \xi(\mathit{X})(\alpha)$, where $\alpha$ is any type variable. (The choice of $\alpha$
is irrelevant because $\theta$ maps all type variables to $\D$,
the only element of $\dom$.)
The next facts follow by structural induction on $t$ and $\phi$ (for arbitrary $\theta$ and $\xi$):
\begin{itemize}
\item $\SEM{t}^\MM_{\theta,\xi} = \SEM{\erasedx{t}}^{\NN}_{\xi'}$;
\item $\SEM{\phi}^\MM_{\theta,\xi} = \SEM{\erasedx{\phi}}^{\NN}_{\xi'}$.
\end{itemize}
In particular, for sentences, we have $\SEM{\phi}^\MM = \SEM{\erasedx{\phi}}^{\NN}$;
and since $\NN$ is a model of $\erasedx{\PHI}$, it follows that $\MM$ is a model of $\PHI$.
\QED
\end{proof}

By way of composition, the
\erased{} encoding lies at the heart of all the
encodings presented in this \Paper. Given $n$~encodings $\xx_1, \ldots, \xx_n$,
we write
$\encode{\phantom{i}}{\xx_1\tinycomma\ldots\tinycomma\xx_n}$ for the composition
$\encode{\phantom{i}}{\xx_{n}} \circ \cdots \circ
\encode{\phantom{i}}{\xx_{1}}$.
Typically, $n$ will be $2$ or $3$ and $\xx_n$ will be $\erased$.
Moreover, $\xx_i$ will be correct and 
will transform the problem so that it belongs to a fragment for which
$\xx_{i+1}$ is also correct.
This will ensure that the whole composition is correct.
Finally, because $\xx_2,\ldots,\xx_n$ will always be fixed
for a given $\xx_1$,
we will call the entire composition
$\encode{\phantom{i}}{\xx_1\tinycomma\ldots\tinycomma\xx_n}$
the ``$\xx_1$ encoding''.

\subsection{Type Arguments}
\label{ssec:type-arguments}

A natural way to prevent the (unsoundness-causing) confusion arising with full
type erasure
is to encode types as terms in the untyped logic.
Instances of polymorphic symbols can be distinguished using explicit type
arguments, encoded as terms:
$n$-ary type constructors $k$ become $n$-ary function
symbols $\const k$,
and type variables~$\a$ become term variables~$A$.
A polymorphic symbol with $m$ type
arguments is passed $m$ additional term arguments. The example
given in the previous subsection
is translated to $\const q(\const{a}, \ff(\const{a})) \mathrel\land \allowbreak
\lnot\,\const q(\const{b}, \ff(\const{b}))$%
, and a fully polymorphic
instance $\const{f}\LANX\a\RAN$ would be encoded as $\const{f}(A)$ (with $A$ a term variable).
We call this encoding $\args$.
%
%

We now proceed with first encoding the types in isolation and then the typed terms.

\begin{defi}[Term Encoding of Types]\afterDot \label{def-term-enc-types}
Let $\mathcal K$ be a finite set of $n$-ary type constructors.
The \emph{term encoding} of a polymorphic type over $\mathcal K$ is a term
over the signature $(\{\jjj\}, \KK', \emptyset)$, where $\KK'$ contains a function symbol $\const k :
\jjj^{\,n} \to \jjj$ for each $k \in \KK$ with $k :: n$. The
encoding is specified by the following equations:
\begin{align*}
\typex{k(\bar\t)} & \,=\, \const k(\typex{\bar\t}) &
  \typex{\a} & \,=\, \mathcal{V}(\a)
\end{align*}
\end{defi}

\begin{defi}[Traditional Arguments \args{}]\afterDot
We first define the encoding function $\encode{\phantom{i}}{\args}$, which translates
polymorphic problems over $\SIGMA = (\KK, \FF, \PP)$ to polymorphic problems over
$(\KK \mathrel{\uplus} \{\jjj\},\allowbreak \FF \mathrel{\uplus} \KK',\allowbreak \PP)$,
where $\jjj$ and $\KK'$ are as in Definition~\ref{def-term-enc-types}.
It adds no axioms, and its term and formula translations
are defined as follows:
\begin{align*}
\argsx{f\LAN\bar\t\RAN(\bar t\,)} & \,=\,  f\LAN\bar\t\RAN(\typex{\bar\t}, \argsx{\bar t\,}) \\[\betweentf]
\argsx{p\LAN\bar\t\RAN(\bar t\,)} & \,=\,  p\LAN\bar\t\RAN(\typex{\bar\t}, \argsx{\bar t\,}) &
  \argsx{\forall \a.\; \phi} & \,=\, \forall \a.\;\forall\typex{\a} \mathbin: \jjj.\; \argsx{\phi} \\
\argsx{\lnot\,p\LAN\bar\t\RAN(\bar t\,)} & \,=\, \lnot\, p\LAN\bar\t\RAN(\typex{\bar\t}, \argsx{\bar t\,})
\end{align*}
(Again, we omit the trivial cases, e.g.\ $\argsx{\forall \mathit{X} \mathbin: \t.\;\, \phi} =
\forall \mathit{X} \mathbin: \t.\;\, \argsx{\phi}$.)
The \emph{traditional type arguments} encoding \args{}
is defined as the composition $\encode{\phantom{i}}{\args\tinycomma\erased}$.
It follows from the definition that
\args{} translates a polymorphic problem
over $\SIGMA = (\mathcal{K}, \mathcal F\!, \mathcal{P})$
into an untyped problem over $\SIGMA' = (\mathcal{F}' \mathrel{\uplus} \KK', \PP')$, where
the symbols in $\FF', \PP'$ are the same as those in
$\FF, \PP$; and for each symbol $s :
\forall\bar\a.\allowbreak\; \bar\t \to \varsigma
\in \FF \mathrel{\uplus} \PP$, the
arity of $s$ in $\SIGMA'$ is $\left|\smash{\bar\a}\right| + \left|\smash{\bar\t}\right|$.
\end{defi}

\begin{exa}%
\label{ex:algebraic-lists-args}%
The \args{} encoding translates the algebraic list
problem of Example~\ref{ex:algebraic-lists} into the following
untyped problem:
\begin{quotex}
$\forall A, \mathit{X}\!, \mathit{Xs}.\;\,
  \const{nil}(A) \not\eq \const{cons}(A, \mathit{X}\!, \mathit{Xs})$ \\
$\forall A, \mathit{Xs}.\;\,
  \mathit{Xs} \eq \const{nil}(A) \mathrel{\lor}
  (\exists \mathit{Y}\!, \mathit{Ys}.\;\,
  \mathit{Xs} \eq \const{cons}(A, \mathit{Y}\!, \mathit{Ys})
  )$ \\
$\forall A, \mathit{X}\!, \mathit{Xs}.\;\,
  \const{hd}(A, \const{cons}(A, \mathit{X}\!, \mathit{Xs})) \eq \mathit{X} \mathrel\land \const{tl}(A, \const{cons}(A, \mathit{X}\!, \mathit{Xs})) \eq \mathit{Xs}$ \\
$\exists \mathit{X}\!, \mathit{Y}\!, \mathit{Xs}, \mathit{Ys}.\;\,
\const{cons}(\const{\www}, \mathit{X}\!, \mathit{Xs}) \eq \const{cons}(\const{\www}, \mathit{Y}\!, \mathit{Ys})
\mathrel\land (\mathit{X} \not\eq \mathit{Y} \mathrel{\lor} \mathit{Xs} \not\eq \mathit{Ys})$
\end{quotex}
\end{exa}

\noindent The $\args$ encoding coincides with $\erased$ for monomorphic problems and
suffers from the same unsoundness with respect to equality and cardinality
constraints. Nonetheless, $\args$ will form the basis of all the sound
polymorphic encodings in a slightly generalised version, called \args\filter{},
for suitable instances of $\xx$.
%


\begin{defi}[Type Argument Filter]\afterDot
\label{def:type-argument-filter}%
Given a signature $\SIGMA = (\mathcal{K}, \mathcal F\!, \mathcal{P})$,
a \emph{type argument filter} $\filt$
maps any $s : \forall\a_1,\ldots,\a_m.\allowbreak\; \bar\t \to \varsigma$
to a subset $\tsfont{\textsl{x}}_{s} = \{i_1,\ldots,i_{m'}\} \subseteq \{1,\ldots,m\}$
of its type argument indices.
Given a list $\bar z$ of length $m$, $\FILT{s}{\bar z}$ denotes the sublist $z_{i_1}, \ldots,z_{i_{\smash{m'}}}$,
where $i_1 < \cdots < i_{\smash{m'}}$.
Filters are implicitly extended to $\{1\}$ 
for the distinguished symbols
$\ti, \is \notin \mathcal{F} \mathrel{\uplus} \mathcal{P}$ introduced in Convention \ref{con-distinguished}.
\end{defi}

\begin{defi}[Generic Arguments \args\negvthinspace\filtertxt]\afterDot
\label{def:generic-arguments}%
Given a type argument filter $\filt$,
we first define the encoding $\encode{\phantom{i}}{\args\filter}$ that translates polymorphic problems over
$\SIGMA = (\KK, \FF, \PP)$ to polymorphic problems over
$(\KK \mathrel{\uplus} \{\jjj\}, \FF \mathrel{\uplus} \KK', \PP)$,
where $\jjj$ and $\KK'$ are as in Definition~\ref{def-term-enc-types}.
It adds no axioms, and its term and formula translations are defined as follows:
\begin{align*}
\argsfx{ f\LAN\bar\t\RAN(\bar t\,)} & \,=\,  f\LAN\bar\t\RAN(\typex{\FILT{f}{\bar\t}}, \argsfx{\bar t\,}) \\[\betweentf]
\argsfx{ p\LAN\bar\t\RAN(\bar t\,)} & \,=\,  p\LAN\bar\t\RAN(\typex{\FILT{p}{\bar\t}}, \argsfx{\bar t\,}) &
  \argsfx{\forall \a.\; \phi} & \,=\, \forall \a.\;\forall\typex{\a} \mathbin: \jjj.\; \argsfx{\phi} \\
\argsfx{\lnot\,p\LAN\bar\t\RAN(\bar t\,)} & \,=\, \lnot\,p\LAN\bar\t\RAN(\typex{\FILT{p}{\bar\t}}, \argsfx{\bar t\,})
\end{align*}
The \emph{generic type arguments} encoding \args\filter{}
is the composition $\encode{\phantom{i}}{\args\filter\tinycomma\erased}$.
It translates a polymorphic problem
over $\SIGMA = (\mathcal{K}, \mathcal F\!, \mathcal{P})$
into an untyped problem over $\SIGMA' = (\mathcal{F}' \mathrel{\uplus} \mathcal{K}, \Pp)$, where
the symbols in $\mathcal{F}', \Pp$ are the same as those in
$\mathcal F\!, \mathcal{P}$; and for each symbol $s :
\forall\bar\a.\allowbreak\; \bar\t \to \varsigma
\in \mathcal{F} \mathrel{\uplus} \mathcal{P}$, the
arity of $s$ in $\SIGMA'$ is $\lenfilt{ s}{\filt} + \left|\smash{\bar\t}\right|$.
\end{defi}

The \erased{} and \args{} encodings correspond to the special cases of
\args$^\filt$ where $\filt$ returns none or all of the type arguments,
respectively.

\begin{thm}[Completeness of \args\negvthinspace\filtertxt]\afterDot
\label{thm:completeness-of-argsx}%
The type arguments encoding \args\filter{} is complete.
\end{thm}
\begin{proof}
%
%
Recall that \args\filter{} is defined as the composition of
$\argsfx{\phantom{i}}$ and $\erasedx{\phantom{i}}$. Since
$\erasedx{\phantom{i}}$ is complete by Theorem~\ref{thm:completeness-of-erased},
it suffices to show that $\argsfx{\phantom{i}}$ is complete.
Let
$\NN = \bigl(\dom',\allowbreak(k^{\NN})_{k \in \KK' \mathrel{\uplus} \{\jjj\}},\allowbreak(\ff^{\,\NN})_{\ff \in \FF \mathrel{\uplus} \KK'},(\pp^{\NN})_{\pp \in \PP}\bigr)$
be a model of $\argsfx{\PHI}$.
We will construct a
structure $\MM = \bigl(\dom,\allowbreak (k^\MM)_{k \in \KK},\allowbreak(\ff^{\,\MM})_{\ff \in \FF},\allowbreak(\pp^\MM)_{\pp \in \PP}\bigr)$ for the signature of
$\PHI$ by taking the same domains as $\MM'$,  interpreting the type constructors other than $\jjj$ in the same way,
and interpreting the function and predicate symbols $s$ as in $\NN$, but supplying, for the extra arguments,
a suitable tuple from $\SEM{\jjj}^{\NN}$
that reflects the type arguments.

To this end, we first apply Lemma \ref{lem-syntactic-domains} to obtain
that every element of $\dom'$ is uniquely represented as $\SEM{\tau}^{\NN}$
with $\tau \in \GType_{\Sigma'}$.
We define
\begin{itemize}
\item $\dom = \dom'$ and $k^\MM = k^{\NN}$;
\item if $s \in \FF \mathrel{\uplus} \PP$, then
$$s^{\,\MM}(\SEM{\ov{\tau}}^{\NN})(\ov{d}) =
\begin{cases}
 s^{\,\NN}
  (\SEM{\ov{\tau}}^{\NN})
  (\SEM{\typex{\tsfont{\textsl{x}}_{s}(\ov{\tau})}}^{\NN},
   \ov{d}) & \mbox{if $\tsfont{\textsl{x}}_{s}(\ov{\tau}) \in \GType_{\Sigma}^n$}
\\
 \mbox{anything} & \mbox{otherwise} 
\end{cases}
$$
where $n$ is the length of $\tsfont{\textsl{x}}_{s}(\ov{\tau})$.
\end{itemize}
Note how the tuple $\SEM{\typex{\tsfont{\textsl{x}}_{s}(\ov{\tau})}}^{\NN}$ reflects the $\tsfont{\textsl{x}}_{s}$-selection from the
type arguments $\SEM{\ov{\tau}}^{\NN}$.

Given $\theta : \mathcal A \ra \dom$ and $\xi : {\mathcal V} \ra \prod_{\sigma \in \Type_{\Sigma}} \SEM{\sigma}^\MM_\theta$,
we define $\xi' : {\mathcal V} \ra \prod_{\sigma' \in \Type_{\Sigma'}} \SEM{\sigma'}^{\NN}_\theta$ by
$$\xi'(X)(\sigma') =
\begin{cases}
 \xi'(X)(\sigma') & \mbox{if $\sigma' \in \Type_{\Sigma}$}
\\
 \SEM{\typex{\tau}}^{\NN}
 & \mbox{if $\sigma' = \jjj$, $X = \VV(\alpha)$, and $\theta(\alpha) = \SEM{\tau}^{\NN}$ for $\tau \in \GType_\Sigma$}
\\
 \mbox{anything}
 & \mbox{otherwise}
\end{cases}
$$
The next facts follow by structural induction on $t$ and $\phi$ (for arbitrary $\theta$ and $\xi$):
\begin{itemize}
\item $\SEM{t}^\MM_{\theta,\xi} = \SEM{\argsfx{t}}^{\NN}_{\theta,\xi'}$;
\item $\SEM{\argsfx{\phi}}^{\NN}_{\theta,\xi'}$ implies $\SEM{\phi}^\MM_{\theta,\xi}$.
\end{itemize}
(The reason why for formula interpretation we have only implication and not equality is the $\forall\alpha$ case:\
when encoding a universally quantified formula
$\forall \a.\; \phi$, the result
$\forall \a.\> \forall\typex{\a} \mathbin: \jjj.\; \argsfx{\phi}$
introduces quantification over two variables, $\a$ and $\typex{\a}$, whose interpretations need not be synchronised.
As a result,
$\forall \a.\>\forall\typex{\a} \mathbin: \jjj.\; \argsfx{\phi}$ could be stronger than $\forall \a.\,\phi$.)
In particular, for a sentence $\phi$, $\SEM{\argsfx{\phi}}^{\NN}$ implies $\SEM{\phi}^\MM$,
and hence $\MM$ is a model of $\PHI$ because $\NN$ is a model of $\argsfx{\PHI}$.
\QED
%
\end{proof}


\subsection{Type Tags}
\label{ssec:type-tags}

An intuitive approach to encode type information soundly
(and completely)
is to wrap each term and subterm with its type using type tags.
For polymorphic type systems, this scheme relies on a distinguished binary
function $\ti(\typex{\t}, t)$ that ``annotates'' each term~$t$ with its
type~$\t$ encoded as a term $\typex{\t}$. The tags make most type arguments
superfluous. We call this
scheme $\tags$, after the tag function of the same name. It
is defined as a two-stage process:\ the first stage adds tags $\ti\LAN\t\RAN(t)$
while preserving the polymorphism; the second stage encodes $\ti$'s type
argument as well as any phantom type arguments.

\begin{defi}[Phantom Type Argument]\afterDot
\label{def:phantom-type-argument}%
Let $s : \forall\a_1,\ldots,\a_m.\allowbreak\; \bar\t
\to \varsigma \in \mathcal{F} \mathrel{\uplus} \mathcal{P}$.
The $i$\vthinspace th type argument is a \emph{phantom} if $\a_i$ does not occur
in $\bar\t$ or $\varsigma$. Given a list $\bar z \equiv z_1,\ldots,z_m$,
$\PHAN{s}{\bar z}$ denotes the sublist $z_{i_1}, \ldots,
z_{i_{\smash{m'}}}$ corresponding to the positions in $\ov{\alpha}$ of the phantom type arguments.
\end{defi}

\begin{defi}[Traditional Tags \tags{}]\afterDot
We first define the encoding $\tagsx{\phantom{i}}$ that translates polymorphic problems over
$\SIGMA = (\KK, \FF, \PP)$ to polymorphic problems over $(\KK, \FF \mathrel{\uplus} \{\ti : \forall \alpha.\;\alpha \ra \alpha\}, \PP)$.
It adds no axioms, and its term and formula translations are defined as follows:
\begin{align*}
\tagsx{f\LAN\t\RAN(\bar t\,)} & \,=\, \CT{}{f\LAN\t\RAN(\tagsx{\bar t\,})} &
\tagsx{\mathit{X}^\sigma} & \,=\, \CT{}{\tagsx{\mathit{X}^\sigma}} &
  \text{with}\;\,\CT{}{t\typ\t} & \,=\, \ti\LAN\t\RAN(t)
\end{align*}
The \emph{traditional type tags} encoding \tags{}
is the composition $\encode{\phantom{i}}{\tags\tinycomma\args^\phan\tinycomma\erased}$.
It translates a polymorphic problem over $\SIGMA$ into an untyped problem
over $\SIGMA' = (\mathcal{F}' \mathrel{\uplus} \mathcal{K}' \mathrel{\uplus} \{\ti^2\}, \PP')$, where
$\KK',\FF',\PP'$ are as for \args$^\phan$ (i.e.\ $\args\filter$ with $\filt = \phan$).
\end{defi}

\begin{exa}%
\label{ex:algebraic-lists-tags}%
The \tags{} encoding translates the algebraic list problem of
Example~\ref{ex:algebraic-lists} into
the following equisatisfiable untyped problem:%
\begin{quotex}
$\forall A, \mathit{X}\!, \mathit{Xs}.\;\,
  \ti(\const{list}(A), \const{nil}) \not\eq \ti(\const{list}(A),
    \const{cons}(\ti(A, \mathit{X}), \ti(\const{list}(A), \mathit{Xs})))$ \\
$\forall A, \mathit{Xs}.\;\,
  \ti(\const{list}(A), \mathit{Xs}) \eq \ti(\const{list}(A), \const{nil}) \mathrel{\lor} {}$ \\
$\phantom{\forall A, \mathit{Xs}.\;\,}
  (\exists \mathit{Y}\!, \mathit{Ys}.\;\, \ti(\const{list}(A), \mathit{Xs}) \eq \ti(\const{list}(A), \const{cons}(\ti(A, \mathit{Y}), \ti(\const{list}(A), \mathit{Ys}))))$ \\
$\forall A, \mathit{X}\!, \mathit{Xs}.\;\,
  \ti(A, \const{hd}(\ti(\const{list}(A), \const{cons}(\ti(A, \mathit{X}), \ti(\const{list}(A), \mathit{Xs}))))) \eq \ti(A, \mathit{X})
  \mathrel\land {}$ \\
$\phantom{\forall A, \mathit{X}\!, \mathit{Xs}.\;\,}
 \ti(\const{list}(A), \const{tl}(\ti(\const{list}(A), \const{cons}(\ti(A, \mathit{X}), \ti(\const{list}(A), \mathit{Xs}))))) \eq \ti(\const{list}(A), \mathit{Xs})$ \\
$\exists \mathit{X}\!, \mathit{Y}\!, \mathit{Xs}, \mathit{Ys}.\;\,
\ti(\const{list}(\const{\www}), \const{cons}(\ti(\const{\www}, \mathit{X}), \ti(\const{list}(\const{\www}), \mathit{Xs}))) \eq {}$ \\
$\phantom{\exists \mathit{X}\!, \mathit{Y}\!, \mathit{Xs}, \mathit{Ys}.\;\,}
  \quad \ti(\const{list}(\const{\www}), \const{cons}(\ti(\const{\www}, \mathit{Y}), \ti(\const{list}(\const{\www}), \mathit{Ys}))) \mathrel\land {}$ \\
$\phantom{\exists \mathit{X}\!, \mathit{Y}\!, \mathit{Xs}, \mathit{Ys}.\;\,}
 (\ti(\const{\www}, \mathit{X}) \not\eq \ti(\const{\www}, \mathit{Y}) \mathrel{\lor} \ti(\const{list}(\const{\www}), \mathit{Xs}) \not\eq \ti(\const{list}(\const{\www}), \mathit{Ys}))$
\end{quotex}
Since there are no phantoms in this example, \args$^\phan$ adds no extra arguments.
All type information is carried
by the $\ti$ function's first argument.
\end{exa}

\begin{exa}\label{ex:linorder}
Consider the following formula, with $\const{linorder} : \forall\a.\>\bool$
and $\const{less\_eq} : \forall\a.\; \a \times \a \to \bool$:
\begin{quotex}
$\forall \a.\> \forall \mathit{X} : \a, \mathit{Y} : \a.\;\, \const{linorder}\LAN\a\RAN \rightarrow \const{less\_eq}(\mathit{X}, Y) \mathrel\lor \const{less\_eq}(Y, \mathit{X})$
\end{quotex}
The $\a$ variable in $\const{linorder}$'s arity declaration is a
phantom. The \tags{} encoding preserves it as an explicit term
argument:
\begin{quotex}
$\forall A, \mathit{X}\!, \mathit{Y}.\;\, \const{linorder}(A) \rightarrow \const{less\_eq}(\ti(A, \mathit{X}),\, \ti(A, Y)) \mathrel\lor \const{less\_eq}(\ti(A, Y),\, \ti(A, \mathit{X}))$
\end{quotex}
As the formula suggests, phantom type arguments can be used to encode predicates on types,
mimicking type classes \cite{wenzel-1997}.
\end{exa}

We postpone the proof of the following theorem to
Section~\ref{ssec:cover-based-type-tags}, in the context of our improved encodings:

\begin{sloppy}
\begin{thm}[Correctness of \tags]\afterDot
\label{thm:correctness-of-tags}%
The traditional type tags encoding \hbox{\rm\tags} is \hbox{correct}.
\end{thm}
\end{sloppy}

\subsection{Type Guards}
\label{ssec:type-guards}

Type tags heavily burden the terms. An alternative is to
introduce type guards, which are predicates that restrict the range
of variables. They take the form of a
distinguished predicate $\is(\typex{\t}, t)$ that checks whether $t$
has type~$\t$. The terms are smaller than with tags, but the formulae
contain more disjuncts.

With the type tags encoding, only phantom type arguments need to be encoded;
here, we must encode any type arguments that cannot be read off the types of the
term arguments. Thus, the type argument is encoded for
$\const{nil}\LAN\a\RAN$ (which has no term arguments) but omitted for
$\const{cons}\LAN\a\RAN(\mathit{X}\!,\mathit{Xs})$,
$\const{hd}\LAN\a\RAN(\mathit{Xs})$, and
$\const{tl}\LAN\a\RAN(\mathit{Xs})$.

\begin{defi}[Inferable Type Argument]\afterDot
\label{def:inferable-type-argument}%
Let $s : \forall\a_1,\ldots,\a_m.\allowbreak\;
\bar\t \to \varsigma \in \mathcal{F} \mathrel{\uplus} \mathcal{P}$.
A type argument $\alpha_j$ is \emph{inferable} if it occurs in some of the
term arguments' types, i.e.\ if there exists
an index $i$ 
such that
$\alpha_j$ occurs in $\sigma_i$.
Given a list $\bar z \equiv z_1,\ldots,z_m$,
let $\INF{s}{\bar z}$ denote the sublist
$z_{i_1}, \ldots, z_{i_{\smash{m'}}}$ corresponding to the positions in $\ov{\alpha}$ of the inferable type
arguments, and let $\NINF{s}{\bar z}$ denote the sublist for noninferable
type arguments.
\end{defi}

Observe that a phantom type argument is in particular noninferable, and a
noninferable nonphantom type argument is one that appears in $\varsigma$ but not
in
$\ov{\sigma}$.

\begin{sloppy}
\begin{defi}[Traditional Guards \guards{}]\afterDot
We first define the encoding $\guardsx{\phantom{i}}$, which translates
a polymorphic problem over $\SIGMA = (\KK,\FF,\PP)$
into an untyped problem over $(\KK,\FF,\PP \mathrel{\uplus} \{\is : \forall \alpha.\;\alpha \ra \bool\})$.
Its term and formula translations are defined as follows:
\begin{align*}
\guardsx{\forall \mathit{X} \mathbin: \t.\;\, \phi} & \,=\, \forall \mathit{X} \mathbin: \t.\;\, \is\LAN\t\RAN(\mathit{X}) \rightarrow \guardsx{\phi} &\enskip 
\guardsx{\exists \mathit{X} \mathbin: \t.\;\, \phi} & \,=\, \exists \mathit{X} \mathbin: \t.\;\, \is\LAN\t\RAN(\mathit{X}) \mathrel\land \guardsx{\phi}
\end{align*}
The encoding also adds the following \emph{typing axioms}:
\[\!\begin{aligned}[t]
& \textstyle \forall \bar\a.\; \mathit{\bar X} \mathbin: \bar\t.\;\,
\bigl(\bigwedge\nolimits_{\smash{j}}\; \is\LAN\t_{\!j}\RAN(X_{\negvvthinspace j})\bigr) \rightarrow
\is\LAN\t\RAN(f\LAN\bar\a\RAN(\mathit{\bar X}))
  &\enskip& \text{for~}\ff : \forall\bar\a.\; \bar\t \to \t \in \mathcal{F} \\[\betwtypax]
& \forall \a.\; \exists \mathit{X} \mathbin: \a.\;\, \is\LAN\a\RAN(\mathit{X})
\end{aligned}\]
(Following Convention \ref{con-encodings},
the translation of a problem is
given by \vthinspace$\guardsx{\PHI} \vthinspace=\vthinspace \AX \mathrel\cup  \{\guardsx{\phi} \mid \phi\in\PHI\}$,
where \AX{} are the typing axioms.)
The \emph{traditional type guards} encoding~\guards{}
is defined as the composition $\encode{\phantom{i}}{\guards\tinycomma\args^\ninf\tinycomma\erased}$.
It translates a polymorphic problem over~$\SIGMA$
into an untyped problem over $\SIGMA' = (\FF' \mathrel{\uplus} \KK',
\PP' \mathrel{\uplus} \{\is^2\})$, where
$\KK',\FF',\PP'$ are as for \args$^\ninf$.
\end{defi}
\end{sloppy}

The last typing axiom in the above definition
witnesses inhabitation of every type. It is necessary for
completeness, in case some of the types do not appear in the result
type of any function symbol.


\begin{exa}%
\label{ex:algebraic-lists-guards}%
The \guards{} encoding translates the algebraic list problem of
Example~\ref{ex:algebraic-lists} into the following:
\begin{quotex}
$\forall A.\;\, \is(\const{list}(A), \const{nil}(A))$ \\
$\forall A, \mathit{X}\!, \mathit{Xs}.\;\, \is(A, \mathit{X}) \mathrel\land \is(\const{list}(A), \mathit{Xs}) \rightarrow \is(\const{list}(A), \const{cons}(\mathit{X}\!, \mathit{Xs}))$ \\
$\forall A, \mathit{Xs}.\;\, \is(\const{list}(A), \mathit{Xs}) \rightarrow \is(A, \const{hd}(\mathit{Xs}))$ \\
$\forall A, \mathit{Xs}.\;\, \is(\const{list}(A), \mathit{Xs}) \rightarrow \is(\const{list}(A), \const{tl}(\mathit{Xs}))$ \\
$\forall A.\; \exists \mathit{X}.\;\, \is(A, \mathit{X})$ \\[\betweenaxs]
$\forall A, \mathit{X}\!, \mathit{Xs}.\;\,
  \is(A, \mathit{X}) \mathrel\land \is(\const{list}(A), \mathit{Xs}) \rightarrow
  \const{nil}(A) \not\eq \const{cons}(\mathit{X}\!, \mathit{Xs})$ \\
$\forall A, \mathit{Xs}.\;\,
  \is(\const{list}(A), \mathit{Xs}) \rightarrow {}$ \\
$\phantom{\forall A, \mathit{Xs}.\;\,}
  \mathit{Xs} \eq \const{nil}(A) \mathrel{\lor}
  (\exists \mathit{Y}\!, \mathit{Ys}.\;\,
  \is(A, \mathit{Y}) \mathrel\land \is(\const{list}(A), \mathit{Ys}) \mathrel\land
  \mathit{Xs} \eq \const{cons}(\mathit{Y}\!, \mathit{Ys})
  )$ \\
$\forall A, \mathit{X}\!, \mathit{Xs}.\;\,
  \is(A, \mathit{X}) \mathrel\land \is(\const{list}(A), \mathit{Xs}) \rightarrow
  \const{hd}(\const{cons}(\mathit{X}\!, \mathit{Xs})) \eq \mathit{X} \mathrel\land \const{tl}(\const{cons}(\mathit{X}\!, \mathit{Xs})) \eq \mathit{Xs}$ \\
$\exists \mathit{X}\!, \mathit{Y}\!, \mathit{Xs}, \mathit{Ys}.\;\,
\is(\const{\www}, \mathit{X}) \mathrel\land \is(\const{\www}, \mathit{Y}) \mathrel\land \is(\const{list}(\const{\www}), \mathit{Xs}) \mathrel\land
\is(\const{list}(\const{\www}), \mathit{Ys}) \mathrel\land {}$ \\
$\phantom{\exists \mathit{X}\!, \mathit{Y}\!, \mathit{Xs}, \mathit{Ys}.\;\,}
\const{cons}(\mathit{X}\!, \mathit{Xs}) \eq \const{cons}(\mathit{Y}\!, \mathit{Ys})
\mathrel\land (\mathit{X} \not\eq \mathit{Y} \mathrel{\lor} \mathit{Xs} \not\eq \mathit{Ys})$
\end{quotex}
In this and later examples, we push guards inside past quantifiers and group
them in a conjunction to enhance readability.

The typing axioms must in general be guarded. Without the guards, any
\const{cons}, \const{hd}, or \const{tl} term could be typed with any type,
defeating the purpose of the guard predicates.
\end{exa}

\begin{exa}
Consider the following formula, where
$\const{inl} : \forall \a,\, \b.\; \a \to \mathit{sum}(\a, \b)$ and
$\const{inr} : \forall \a,\, \b.\allowbreak\; \b \to \mathit{sum}(\a, \b)$:
\begin{quotex}
$\forall \a, \b.\; \forall \mathit{X} : \a,\, \mathit{Y} : \b.\;\, \const{inl}\LAN \a,\b\RAN(\mathit{X}) \not\eq \const{inr}\LAN \a,\b\RAN(\mathit{Y})$
\end{quotex}
The $\b$ variable in \const{inl}'s arity declaration and the $\a$ in
\const{inr}'s are noninferable. The \guards{} encoding preserves them as explicit term
arguments:
\begin{quotex}
$\forall A, B, \mathit{X}\!, \mathit{Y}\!.\;\,
  \is(A, \mathit{X}) \mathrel\land \is(B, \mathit{Y}) \rightarrow
  \const{inl}(B, \mathit{X}) \not\eq \const{inr}(A, \mathit{Y})$
\end{quotex}
\end{exa}

\begin{thm}[Correctness of \guards]\afterDot
\label{thm:correctness-of-guards}%
The traditional type guards encoding\/ \hbox{\rm\guards} is correct.
\end{thm}
\begin{proof}
This will be a consequence of Theorem \ref{thm:correctness-of-guards-at} for our parameterised
cover-based encoding \guards\at{}, since \guards{} is a particular case of
\guards\at{}. 
\QED
\end{proof}




\begin{rem}
The above encodings, as well as those discussed in the next sections, all lead
to an untyped problem. An increasing number of automatic provers support
monomorphic types, and it may seem desirable to exploit such support when it is
available. With such provers, we can replace the $\erased$ encoding with a
variant that enforces a basic type discipline by distinguishing two
types, $\jjj$ (for encoded types) and $\iota$ (for encoded terms). An incomplete
(non-proof-preserving) alternative is to perform heuristic monomorphisation
(Section~\ref{ssec:heuristic-monomorphisation}). Hybrid schemes that exploit
monomorphic types, including interpreted types, are also possible and have been
studied by other researchers (Section~\ref{sec:related-work}).
\end{rem}

\section{Cover-Based Encodings of Polymorphism}
\label{sec:alternative-cover-based-encoding-of-polymorphism}


%


Type tags and guards considerably increase the size of the problems passed to
the automatic provers, with a dramatic impact on their performance. A lot of the
type information generated by the traditional encodings \tags{} and \guards{} is
redundant. For example, \tags{} translates
$\const{cons}\LAN\a\RAN(\mathit{X}\!, \mathit{Xs})$ to
$\ti(\const{list}(A), \const{cons}(\ti(A, \mathit{X}), \ti(\const{list}(A), \mathit{Xs})))$,
but intuitively only one of the three tags is necessary to specify the right type
for the expression if we know the arity of \const{cons}. The cover-based encodings
capitalise on this, by supplying only a minimum of protectors and adding typing
axioms that effectively compute the type of function symbols from a selection
of their term arguments' types---the ``cover''.

%
Let us first rigorously define this notion of term arguments
``covering'' type arguments.

\begin{defi}[Cover]\afterDot
\label{def-cover}
Let $s : \forall\bar\a.\; \bar\t \to \varsigma \in \mathcal{F} \mathrel{\uplus} \mathcal{P}$.
A (\emph{type argument}) \emph{cover} $C \subseteq \{1,\ldots,\left|\smash{\bar\t}\right|\}$ for $s$ is a
set of term argument indices such that
any inferable type argument can be inferred from a term argument whose index belongs to $C$,
i.e.\ for all $j$, 
if $\alpha_j$ appears in $\ov{\sigma}$, it also appears
in some $\sigma_i$ such that $i \in C$.
A cover $C$ of $s$
is \emph{minimal} if no proper subset of $C$ is a cover for~$s$; it
is \emph{maximal} if $C = \{1,\ldots,\left|\smash{\bar\t}\right|\}$.
We let $\Cover{s}$ denote an arbitrary but fixed 
cover for $s$.
\end{defi}

In practice, we would normally take a minimal cover for $\Cover{s}$ to reduce
clutter.
Accordingly, $\{1\}$ and $\{2\}$ are minimal covers for $\const{cons} :
\forall\a.\; \a \times \mathit{list}(\a) \to \mathit{list}(\a)$,
whereas $\{1, 2\}$ is a maximal cover.

\begin{conv}\label{cover-cons}
As canonical cover, we arbitrarily choose
$\Cover{\const{cons}} = \{1\}$.
\end{conv}

The cover-based encoding \guards\at{} introduced below
is a generalisation of the traditional encoding \guards{}.
The two encodings coincide if
$\Cover{s}$ is chosen to be maximal for all symbols~$s$.
In contrast, the cover-based encoding \tags\at{} is not exactly
a generalisation of \tags{}, although they share many ideas. For this reason,
we momentarily depart from our general policy of considering tags before guards
so that we can present the easier case first.

Intuitively, \guards\at{} and \tags\at{} ensure
that each term argument position that is part of its enclosing function or
predicate symbol's cover has a unique type associated with it, from
which the omitted type arguments can be inferred. Thus, \tags\at{}
translates $\const{cons}\LAN\a\RAN(\mathit{X}\!, \mathit{Xs})$ to $\const{cons}(\ti(A,
\mathit{X}), \mathit{Xs})$ with a type tag around $\mathit{X}$, effectively
protecting the term from an
ill-typed instantiation of $\mathit{X}$ that would result in the wrong type
argument being inferred for \const{cons}.

There is no need to protect the second argument, $\mathit{Xs}$,
since it is not part of the cover.
We call variables that occur in their enclosing symbol's cover
(and hence that ``carry'' some type arguments) ``undercover variables''.
It may seem dangerous to allow ill-typed terms to instantiate $\mathit{Xs}$,
but this is not an issue because such terms cannot contribute meaningfully to a
proof. At most, they can act as witnesses for the existence of terms of given
types, but even in that capacity they are not necessary.

\begin{defi}[Undercover Variable]\afterDot
The set of \emph{undercover variables} $\UV(\phi)$ of a formula~$\phi$ is
defined by the equations
\begin{align*}
\UV(f\LAN\bar\t\RAN(\bar t\,)) & \,=\,
    \CT{f}{\bar t\,} \mathrel\cup \UV(\bar t\,) &
  \UV(\mathit{X}) & \,=\, \emptyset \\[\betweentf]
\UV(p\LAN\bar\t\RAN(\bar t\,)) & \,=\,
    \CT{p}{\bar t\,} \mathrel\cup \UV(\bar t\,) &
  \UV(t_1 \eq t_2) & \,=\,
      (\{t_1, t_2\} \mathrel\cap \Vt) \mathrel\cup \UV(t_1, t_2) \\
\UV(\lnot\,p\LAN\bar\t\RAN(\bar t\,)) & \,=\,
    \CT{p}{\bar t\,} \mathrel\cup \UV(\bar t\,) &
  \UV(t_1 \not\eq t_2) & \,=\, \UV(t_1, t_2) \\
\UV(\phi_{1\!} \mathrel\land \phi_2) & \,=\, \UV(\phi_1, \phi_2) &
  \UV(\forall \mathit{X}\mathbin:\t.\;\, \phi) & \,=\, \UV(\phi)
\\
\UV(\phi_{1\!} \mathrel\lor \phi_2) & \,=\, \UV(\phi_1, \phi_2) &
  \UV(\exists \mathit{X}\mathbin:\t.\;\, \phi) & \,=\, \UV(\phi) - \{\mathit{X}^\sigma\}
\\[\lvminusbaselineskip]
\UV(\forall\a.\; \phi) & \,=\, \UV(\phi)
\end{align*}
where
$\CT{s}{\bar t\,} \,=\,
\{t_{\negvthinspace j} \mid j \in \Cover{s}\} \mathrel\cap \Vt$
and
$\UV(\bar t\,) = {\textstyle\bigcup_j {\UV(t_{\negvthinspace j})}}$.
\end{defi}

\subsection{Cover-Based Type Guards}
\label{ssec:cover-based-type-guards}

The cover-based encoding \guards\at{} is similar to the traditional encoding
\guards{}, except that it guards only undercover occurrences of variables.

%





\begin{defi}[Cover Guards \guards\at]\afterDot
The encoding $\guardsax{\phantom{i}}$ is defined similarly to the encoding
$\guardsx{\phantom{i}}$ (used for the traditional \guards{} encoding) except for the $\forall$ case in its formula
translation and the typing axioms.  Namely, the $\forall$ case adds guards
only for universally quantified variables that are undercover:
\begin{align*}
\guardsax{\forall \mathit{X} \mathbin: \t.\;\, \phi} & \,=\,
\forall \mathit{X} \mathbin: \t.\;
\begin{cases}
\guardsax{\phi} & \!\!\text{if $\mathit{X} \notin \UV(\phi)$} \\
\is\LAN\t\RAN(\mathit{X}) \rightarrow \guardsax{\phi} & \!\!\text{otherwise}
\end{cases}
\end{align*}
Moreover, the typing axioms take the cover into consideration:
\[\!\begin{aligned}[t]
& \textstyle \forall \bar\a.\; \mathit{\bar X} \mathbin: \bar\t.\;\,
\bigl(\bigwedge\nolimits_{\smash{j\in \Cover{f}}}\; \is\LAN\t_{\!j}\RAN(X_{\negvvthinspace j})\bigr) \rightarrow
\is\LAN\t\RAN(f\LAN\bar\a\RAN(\mathit{\bar X}))
  &\enskip& \text{for~}\ff : \forall\bar\a.\; \bar\t \to \t \in \mathcal{F} \\[\betwtypax]
& \forall \a.\; \exists \mathit{X} \mathbin: \a.\;\, \is\LAN\a\RAN(\mathit{X})
\end{aligned}\]
The \emph{cover-based type guards} encoding~\guards\at\
is defined as the composition $\encode{\phantom{i}}{\guards\at\tinycomma\args^\ninf\tinycomma\erased}$.
It translates a polymorphic problem $\PHI$ over $\SIGMA$
into an untyped problem $\encode{\PHI}{\guards\at\tinycomma\args^\ninf\tinycomma\erased}$
over $\SIGMA' = (\mathcal{F}' \mathrel{\uplus} \mathcal{K}' , \PP' \mathrel{\uplus} \{\is^2\})$, where
$\KK',\FF',\PP'$ are as for \args$^\ninf$. 
\end{defi}


\begin{exa}%
\label{ex:algebraic-lists-guards-at}%
If we choose the cover for $\const{cons}$ as in Convention \ref{cover-cons},
the \guards\at{} encoding of the algebraic list problem is identical
to the \guards{} encoding (Example~\ref{ex:algebraic-lists-guards}),
except that the
guard $\is(\const{list}(A), \mathit{Xs})$
is omitted in typing axiom and one of the problem axioms:
\begin{quotex}
$\forall A, \mathit{X}\!, \mathit{Xs}.\;\, \is(A, \mathit{X}) \rightarrow \is(\const{list}(A), \const{cons}(\mathit{X}\!, \mathit{Xs}))$ \\[\betweenaxs]
$\forall A, \mathit{X}\!, \mathit{Xs}.\;\,
  \is(A, \mathit{X}) \rightarrow
  \const{nil}(A) \not\eq \const{cons}(\mathit{X}\!, \mathit{Xs})$
\end{quotex}
By leaving $\mathit{Xs}$ unconstrained, the typing axiom for $\const{cons}$
gives a type to some
``ill-typed'' terms, such as $\const{cons}(0\typ{\mathit{nat}}, 0\typ{\mathit{nat}})$.
Intuitively, this is safe because such terms cannot be used to prove anything
useful that could not be proved with a ``well-typed'' term. What matters is that
``well-typed'' terms are associated with their correct type and that
``ill-typed'' terms are given at most one type.
\end{exa}

\begin{lem} \label{lem-infer-domains}
Let $\MM = \bigl(\dom,(k^\MM)_{k \in \KK},(\ff^{\,\MM})_{\ff \in \FF},(\pp^\MM)_{\pp \in \PP}\bigr)$ be such
that the domains in\/ $\dom$ are mutually disjoint and the type constructors $k^\MM$ are injective.
Assume\/ $\ov{\alpha} = (\alpha_1,\ldots,\alpha_m)$,
$\ov{\beta} = (\beta_1,\ldots,\beta_n)$,
$\ov{\sigma} = (\sigma_1,\ldots,\sigma_u)$,
$\ov{\tau} = (\tau_1,\ldots,\tau_v)$,
and\/
$s : \forall (\ov{\alpha},\ov{\beta}
).\;(\ov{\sigma},\ov{\tau}) \ra \varsigma$ in $\FF \mathrel{\uplus} \PP$,
such that
the last $n$ type arguments {\rm(}corresponding to $\ov{\beta}${\rm)} are inferable
and the first $u$ term arguments {\rm(}corresponding to $\ov{\sigma}${\rm)} constitute a cover for $s$.
Let\/
$(d_1,\ldots,d_u) \in (\bigcup_{D \in \dom} D)^u$.
Then there exists at most one tuple\/ $\ov{E} = (E_1,\ldots,E_n) \in \dom^n$ such that
each $d_i$ is in\/ $\SEM{\sigma_i}_\theta^\MM$
for some $\theta$ that maps $\ov{\beta}$ to $\ov{E}$.
%
\end{lem}
\begin{proof}
From the inferability and cover assumptions, we have the condition
$\TVars(\ov{\beta}) \subseteq \TVars(\ov{\sigma},\ov{\tau}) = \TVars(\ov{\sigma})$
on type variables.
Assume another such tuple
$\ov{E'}$ exists. Let $\theta'$ be its corresponding substitution, and let $j \in \{1,\ldots,n\}$.
By the type variable condition, there exists $i \in \{1,\ldots,u\}$ such that $\beta_j \in \TVars(\sigma_i)$;
since $d_i \in \SEM{\sigma_i}_\theta^\MM \,\cap\, \SEM{\sigma_i}_{\theta'}^\MM$ and distinct domains are disjoint,
we have $\SEM{\sigma_i}_\theta^\MM = \SEM{\sigma_i}_{\theta'}^\MM$, which, together with $\beta_j \in \TVars(\sigma_i)$ and
the injectivity of the type constructors, implies $\theta(\beta_j) = \theta'(\beta_j)$, hence $E_j = E'_{j}$;
and since $j$ was arbitrary, we obtain $\ov{E} = \ov{E'}$, as desired.
\end{proof}

\begin{thm}[Correctness of \guards\at]\afterDot
\label{thm:correctness-of-guards-at}%
The cover-based type guards encoding \hbox{\rm\guards\at} is correct.
\end{thm}

\begin{proof}
Let $\Sigma = (\KK,\FF,\PP)$ be the signature of a polymorphic problem $\PHI$.

\betweenitems\noindent
\textsc{Sound}:\enskip
Let $\MM = \bigl(\dom,(k^\MM)_{k \in \KK},(\ff^{\,\MM})_{\ff \in \FF},(\pp^\MM)_{\pp \in \PP}\bigr)$
be a model of $\PHI$.
By Lemma~\ref{lem-syntactic-domains}, we may assume that $\dom$ is disjoint from each of its elements, its elements are mutually disjoint,
and the type constructors $k^\MM$ satisfy distinctness and injectivity.


We define a structure $\NN = \bigl(\D',(\ff^{\,\NN})_{\ff \in \FF' \mathrel{\uplus} \KK'},(\pp^{\NN})_{\pp \in \PP' \mathrel{\uplus} \{\is^2\}}\bigr)$
for the untyped signature $\SIGMA' = (\FF' \mathrel{\uplus} \KK' , \PP' \mathrel{\uplus} \{\is^2\})$ as follows.
$\D'$ is the union of the domains of $\MM$ and their elements,
$\is^{\NN}$ is the set membership relation, and the symbols of $\NN$ common to those of $\MM$ try to emulate the $\MM$ interpretation
as closely as possible:\ $k^{\NN}$ acts like $k^{\MM}$ on the $\dom$ subset of $\D'$, and similarly for
function and predicate symbols~$s$ whose noninferable type arguments in $\MM$ become regular (term) arguments in $\NN$.
The reason why we can omit the inferable type arguments of $s$ is their recoverability from the type of $s$ and the domains
of the genuine term arguments.
Additionally, when applying $s^{\NN}$ to inputs outside $s$'s cover that do not belong to their proper
domains in $\MM$, we correct these inputs by replacing them with arbitrary inputs from the proper domains. We do this because,
in formulae, these inputs will not be guarded by the encoding, but we will nevertheless want to infer the encoded formula holding in $\NN$
from the original formula holding in $\MM$.

Formally,
we define the components of $\NN$ as follows.
First, $\D' = \dom \mathrel{\uplus} (\bigcup_{D \in \dom} D)$ 
and $\is^{\NN}(a,b) = (a \in \dom \mathrel\land b \in a)$.
Assume $k :: n$ is in $\KK$, meaning that $k$ is an $n$-ary function symbol in $\KK'$.
Then
$$
k^{\NN}(\ov{D}) = \left\{
\begin{array}{@{}ll@{}}
 k^{\MM}(\ov{D}) & \mbox{if $\ov{D} \in \dom^n$}
\\
 \eps{\dom} & \mbox{otherwise}
\end{array}
\right.
$$
Assume $\ov{\alpha} = (\alpha_1,\ldots,\alpha_m)$,
$\ov{\beta} = (\beta_1,\ldots,\beta_n)$,
$\ov{\sigma} = (\sigma_1,\ldots,\sigma_u)$,
$\ov{\tau} = (\tau_1,\ldots,\tau_v)$,
and
$s : \forall \ov{\alpha}, \ov{\beta}.\; \ov{\sigma} \times \ov{\tau} \ra \varsigma$ is in $\FF \mathrel{\uplus} \PP$,
such that
the first $m$ type arguments are noninferable, the last $n$ type arguments
are inferable, and the first $u$ term arguments constitute $s$'s cover.
(The general case, with arbitrary permutations of noninferable and cover arguments,
can be handled similarly, albeit with heavier notation.)
Then $s$ is an $(m+u+v)$-ary symbol in $\FF' \mathrel{\uplus} \PP'$.
Let $\ov{D} = (D_1,\ldots,D_m) \in {\D'}^m$, $\ov{d} = (d_1,\ldots,d_u) \in {\D'}^u$ and $\ov{e} = (e_1,\ldots,e_v) \in {\D'}^v$, and consider the
following condition on domains:
\begin{center}
$(D_1,\ldots,D_m) \in \dom^m$ and there exists $\ov{E} = (E_1,\ldots,E_n) \in \dom^n$ such that
each $d_i$ is in $\SEM{\sigma_i}_\theta^\MM$, where $\theta$ maps $(\ov{\alpha},\ov{\beta})$ to $(\ov{D},\ov{E})$
\end{center}
Assuming the condition holds, by Lemma~\ref{lem-infer-domains} there exists precisely one tuple $\ov{E}$ satisfying it.
We let $\ov{e}' = (e_1',\ldots,e_v')$, where
$$
e_i' =
\left\{
\begin{array}{@{}ll@{}}
 e_i & \mbox{if $e_i \in \SEM{\tau_i}^\MM_\theta$}
\\
 \eps{\SEM{\tau_i}^\MM_\theta} & \mbox{otherwise}
\end{array}
\right.
$$
Finally, we define
$$
s^{\NN}(\ov{D},\ov{d},\ov{e}) = \left\{
\begin{array}{@{}ll@{}}
 s^{\MM}(\ov{D},\ov{E})(\ov{d},\ov{e}') & \mbox{if the domain condition holds}
\\
 \eps{\dom} 
& \mbox{if the domain condition fails and $s \in \FF$}
\\
 \eps{\bool} 
& \mbox{if the domain condition fails and $s \in \PP$}
\end{array}
\right.
$$
We will show that $\NN$ is a model of $\encode{\PHI}{\guards\at\tinycomma\args^\ninf\tinycomma\erased}$.
Let $\xi' : \VV \ra \D'$ be a valuation that respects types, in the sense that $\xi'(\VV(\alpha)) \in \dom$ for all $\alpha \in \AAA$.
We define $\theta : \AAA \ra \dom$ by $\theta(\alpha) = \xi'(\VV(\alpha))$ and $\xi : \VV \ra \prod_{\sigma \in \Type_\Sigma} \SEM{\sigma}^\MM_\theta$ by
$\xi(\mathit{X})(\sigma) = \xi'(\mathit{X})$ if $\xi'(\mathit{X}) \in \SEM{\sigma}^\MM_\theta$ and $= \eps{\SEM{\sigma}^\MM_\theta}$ otherwise.
Facts (1)--(3) below follow by induction
on $\sigma$, $t$ or $\phi$ (for arbitrary $\xi'$), and (1$'$) is a consequence of (1) and the definition of $\is^{\NN}$:
\begin{enumerate}
\item[(1)] $\SEM{\typex{\sigma}}^{\NN}_{\xi'} = \SEM{\sigma}^\MM_\theta$;
\item[(1$'$)] $\SEM{\is(\typex{\sigma},\mathit{X})}^{\NN}_{\xi'} = (\xi'(\mathit{X}) \in \SEM{\sigma}^\MM_\theta)$;
\item[(2)] if $t \notin \Vt$ and $\xi(\mathit{X}) \in \SEM{\sigma}^\MM_\theta$ for all $\mathit{X}^\sigma \in \UV (t)$,
then $\SEM{\encode{t}{\guards\at\tinycomma\args^\ninf\tinycomma\erased}}^{\NN}_{\xi'} = \SEM{t}^\MM_{\theta,\xi}$;
\item[(3)] if $\xi'(\mathit{X}) \in \SEM{\sigma}^\MM_\theta$ for all $\mathit{X}^\sigma \in \UV (\phi)$,
then $\SEM{\phi}^\MM_{\theta,\xi}$ implies $\SEM{\encode{\phi}{\guards\at\tinycomma\args^\ninf\tinycomma\erased}}^{\NN}_{\xi'}$.
\end{enumerate}

Let us detail a few interesting cases in these proofs:

\begin{quote}
\textsc{Inductive case for (2):}\enskip
Assume $t = f\LAN\ov{\tau}\RAN(t_1,\ldots,t_n)$, and
let $i \in \{1,\ldots,n\}$.  If $t_i \not\in \Vt$, then the induction hypothesis applies to it,
yielding \[\SEM{\encode{t_i}{\guards\at\tinycomma\args^\ninf\tinycomma\erased}}^{\NN}_{\xi'} = \SEM{t_i}^\MM_{\theta,\xi}\]
If $t_i = \mathit{X}^\sigma$, then $\SEM{\encode{t_i}{\guards\at\tinycomma\args^\ninf\tinycomma\erased}}^{\NN}_{\xi'} = \xi'(\mathit{X})$,
$\SEM{t_i}^\MM_{\theta,\xi} = \xi(\mathit{X})(\sigma)$,
and we have two cases:
\begin{itemize}
\item If $\xi'(\mathit{X}) \in \SEM{\sigma}^\MM_\theta$, then $\xi'(\mathit{X})(\sigma) = \xi(\mathit{X})$.
\item Otherwise, by the assumptions $i \notin \Cover{f}$, and hence the definition of $f^{\NN}$ effectively
replaces the argument $\xi'(\mathit{X})$ by $\eps{\SEM{\sigma}^\MM_\theta}$, which also equals $\xi(\mathit{X})(\sigma)$.
\end{itemize}
The above shows $\SEM{\encode{t}{\guards\at\tinycomma\args^\ninf\tinycomma\erased}}^{\NN}_{\xi'} = \SEM{t}^\MM_{\theta,\xi}$.

\betweenitems\noindent
\textsc{Disequality case for (3):}\enskip
The subcase when one of the terms is a variable and the other is not.
Assume $\phi$ has the form $\mathit{X}^\sigma \not\eq t$.  If $\xi'(\mathit{X}) \in \SEM{\sigma}^\MM_\theta$, then
$\xi'(\mathit{X})(\sigma) = \xi(\mathit{X})$, hence
$\SEM{\encode{\phi}{\guards\at\tinycomma\args^\ninf\tinycomma\erased}}^{\NN}_{\xi'} = \SEM{\phi}^\MM_{\theta,\xi}$.
Otherwise, $\SEM{\encode{\phi}{\guards\at\tinycomma\args^\ninf\tinycomma\erased}}^{\NN}_{\xi'}$ since
$\SEM{\encode{t}{\guards\at\tinycomma\args^\ninf\tinycomma\erased}}^{\NN}_{\xi'} \in \SEM{\sigma}^\MM_\theta$
(and hence is different from $\xi'(\mathit{X})$).

\betweenitems\noindent
\textsc{Universal quantifier case for (3):}\enskip
Assume (A)~$\SEM{\forall \mathit{X} \mathbin: \sigma.\; \phi}^\MM_{\theta,\xi}$.
We must show $\SEM{\encode{\forall \mathit{X} \mathbin: \sigma.\; \phi}{\guards\at\tinycomma\args^\ninf\tinycomma\erased}}^{\NN}_{\xi'}$.
Fix $d \in \D$ and
let $\xi_1' = \xi'[\mathit{X} \ra d]$.
\begin{itemize}
\item Assume $\mathit{X} \in \UV(\phi)$. Then we assume $\SEM{\is(\typex{\sigma},\mathit{X})}^{\NN}_{\xi'}$, i.e.\
(B)~$d \in \SEM{\sigma}^\MM_\theta$, and need to show
$\SEM{\encode{\phi}{\guards\at\tinycomma\args^\ninf\tinycomma\erased}}^{\NN}_{\xi_1'}$.
From (B), we have $\xi(\mathit{X})(\sigma) = \xi'(\mathit{X})$, hence $\xi_1 = \xi[\mathit{X} \mapsto d]$;
with (A), this implies $\SEM{\phi}^\MM_{\theta,\xi_1}$, and hence the desired fact follows from the induction hypothesis.
\item Assume $\mathit{X} \not\in \UV(\phi)$.
We must show $\SEM{\encode{\phi}{\guards\at\tinycomma\args^\ninf\tinycomma\erased}}^{\NN}_{\xi_1'}$.
If $d \in \SEM{\sigma}^\MM_\theta$, then the argument goes the same as above. 
So assume $d \not\in \SEM{\sigma}^\MM_\theta$.
Then $\xi(\mathit{X})(\sigma) = \eps{\SEM{\sigma}^\MM_\theta}$, hence
$\xi_1 = \xi[\mathit{X} \mapsto \eps{\SEM{\sigma}^\MM_\theta}]$;
with (A), this implies $\SEM{\phi}^\MM_{\theta,\xi_1}$, and hence the desired fact follows from the induction hypothesis.
\end{itemize}
\end{quote}

\noindent From (1$'$) and the definition
of $f^{\NN}$, it follows that $\NN$ satisfies the necessary
axioms, namely, $\encode{\phi}{\args^\ninf\tinycomma\erased}$ for each typing axiom $\phi$.
It remains to show that $\NN$ satisfies $\encode{\phi}{\guards\at\tinycomma\args^\ninf\tinycomma\erased}$ for each $\phi \in \PHI$.
So let $\phi \in \PHI$.  Then $\SEM{\phi}^\MM$.
We pick any $\xi'$ that respects types and has the property
that $\xi'(\mathit{X}) \in \SEM{\sigma}^\MM_\theta$ for all $\mathit{X}^\sigma \in \UV (\phi)$;
by (3) and the fact that $\phi$ and $\encode{\phi}{\guards\at\tinycomma\args^\ninf\tinycomma\erased}$ are sentences,
it follows that $\SEM{\phi}^\MM_{\xi,\theta}$, hence
$\SEM{\encode{\phi}{\guards\at\tinycomma\args^\ninf\tinycomma\erased} }^{\NN}_{\xi'}$,
hence $\SEM{\encode{\phi}{\guards\at\tinycomma\args^\ninf\tinycomma\erased}}^{\NN}$.
\betweenitems\noindent
\textsc{Complete}:\enskip
This part is easy.
By Theorem \ref{thm:completeness-of-argsx}, it suffices to show that $\encode{\phantom{i}}{\guards\at}$
is complete. Let $\NN = \bigl(\dom',(k^{\NN})_{k \in \KK},(\ff^{\,\NN})_{\ff \in \FF},(\pp^{\NN})_{\pp \in \PP \mathrel{\uplus} \{\is\}}\bigr)$
be a model of $\encode{\PHI}{\guards\at}$, for which again we may assume that the domains are mutually disjoint.
We define $\MM = \bigl(\dom,(k^\MM)_{k \in \KK},(\ff^{\,\MM})_{\ff \in \FF},\allowbreak(\pp^\MM)_{\pp \in \PP}\bigr)$
by restricting the domains according to the guards and restricting the operations correspondingly.
Namely, for each $D \in \dom'$, let $D' = \{d \in D \mid \is^{\NN}(D)(d) \}$.  We take
\begin{itemize}
\item $\dom = \{D' \mid D \in \dom'\}$;
\item $k^{\MM}(\ov{D'}) = k^{\NN}(\ov{D})$;
\item $s^{\MM}(\ov{D'})(\ov{d}) = s^{\NN}(\ov{D})(\ov{d})$.
\end{itemize}
Thanks to the typing axioms, each $D'$ is nonempty.
Since the domains are disjoint, each $D'$ determines its $D$ uniquely; together with the typing axioms,
this also ensures that each $k^{\MM}$ and $s^{\MM}$ are well defined.
Now, $\SEM{U}^\MM_{\theta,\xi} = \SEM{\encode{U}{\guards\at}}^{\NN}_{\theta,\xi}$,
where $U$
is first a term and then a formula,
follows by induction on $U$ (for arbitrary $\theta$ and $\xi$).  From this, we obtain that $\MM$ is a model of
$\encode{\PHI}{\guards\at}$ by the usual route.
\QED
\end{proof}

\subsection{Cover-Based Type Tags}
\label{ssec:cover-based-type-tags}

The cover-based encoding \tags\at{} is similar to the traditional encoding
\tags{}, except that it tags only undercover occurrences of variables and
requires typing axioms to add or remove tags around function terms.

\begin{defi}[Cover Tags \tags\at]\afterDot
The encoding $\tagsax{\phantom{i}}$ translates polymorphic problems over
$\SIGMA = (\KK, \FF, \PP)$ to polymorphic problems over $(\KK, \FF \mathrel{\uplus} \{\ti : \forall \alpha.\;\alpha \ra \alpha\}, \PP)$.
Its term and formula translations are the following:
\begin{align*}
\tagsax{f\LAN\bar\t\RAN(\bar t\,)} & \,=\,
    f\LAN\bar\t\RAN({\CT{f}{\tagsax{\bar t\,}}}) \\[\betweentf]
\tagsax{p\LAN\bar\t\RAN(\bar t\,)}{} & \,=\,
    p\LAN\bar\t\RAN({\CT{p}{\tagsax{\bar t\,}}}) &
  \tagsax{t_{1\!} \eq t_2}{} & \,=\, \CT{\eq}{\tagsax{t_{1\!}}} \eq \CT{\eq}{\tagsax{t_{2}}} \\
\tagsax{\lnot\,p\LAN\bar\t\RAN(\bar t\,)}{} & \,=\,
    \lnot\,p\LAN\bar\t\RAN({\CT{p}{\tagsax{\bar t\,}}}) &
  \tagsax{\exists \mathit{X} \mathbin: \t.\;\, \phi} & \,=\,
    \exists \mathit{X} \mathbin: \t.\;\, \ti\LAN\t\RAN(\mathit{X}) \eq \mathit{X} \mathrel\land \tagsax{\phi}
\end{align*}
The auxiliary function
$\CT{s}{(t_1\typ{\t_1},\allowbreak\ldots,t_n\typ{\t_n})}$ returns a vector
$(u_1,\ldots,u_n)$ such that
\[u_{\negvvthinspace j} \,=\,
  \begin{cases}
    \ti\LAN\t_{\!j}\RAN(t_{\negvthinspace j}) & \!\!\text{if $j \in \Cover{s}$ and $t_{\negvthinspace j} \in \Vtun$
} \\
    t_{\negvthinspace j} & \!\!\text{otherwise}

  \end{cases}\]
taking $\Cover{\eq} = \{1, 2\}$.
The encoding adds the following typing axioms:
\[\!\begin{aligned}[t]
& \forall \bar\a.\; \forall \mathit{\bar X} \mathbin: \bar\t.\;\,
\ti\LAN\t\RAN(f\LAN\bar\a\RAN({\CT{f}{\mathit{\bar X}}}))
\eq f\LAN\bar\a\RAN({\CT{f}{\mathit{\bar X}}})
  &\enskip& \text{for~}\ff : \forall\bar\a.\; \bar\t \to \t \in \mathcal{F} \\[\betwtypax]
& \forall \a.\; \exists \mathit{X} \mathbin: \a.\;\, \ti\LAN\a\RAN(\mathit{X}) \eq \mathit{X}
\end{aligned}\]
The \emph{cover-based type tags} encoding \tags{}
is the composition $\encode{\phantom{i}}{\tags\at\tinycomma\args^\ninf\tinycomma\erased}$.
It translates a polymorphic problem over $\SIGMA$
into an untyped problem over $\SIGMA' = (\FF' \uplus
\KK' \mathrel{\uplus} \{\ti^2\},\PP')$, where
$\KK', \FF', \PP'$ are as for \args$^\ninf$.
\end{defi}

\begin{exa}%
\label{ex:algebraic-lists-tags-at}%
The \tags\at{} encoding of Example~\ref{ex:algebraic-lists} is as follows (again, choosing the cover for $\const{cons}$ as in
Convention \ref{cover-cons}):
\begin{quotex}
$\forall A.\;\, \ti(\const{list}(A), \const{nil}(A)) \eq \const{nil}(A)$ \\
$\forall A, \mathit{X}\!, \mathit{Xs}.\;\, \ti(\const{list}(A), \const{cons}(\ti(A, \mathit{X}), \mathit{Xs})) \eq
   \const{cons}(\ti(A, \mathit{X}), \mathit{Xs})$ \\
$\forall A, \mathit{Xs}.\;\, \ti(\const{list}(A), \const{hd}(\ti(\const{list}(A), \mathit{Xs}))) \eq \const{hd}(\ti(\const{list}(A), \mathit{Xs}))$ \\
$\forall A, \mathit{Xs}.\;\, \ti(A, \const{tl}(\ti(\const{list}(A), \mathit{Xs}))) \eq \const{tl}(\ti(\const{list}(A), \mathit{Xs}))$
  \\[\betweenaxs]
$\forall A, \mathit{X}\!, \mathit{Xs}.\;\,
  \const{nil}(A) \not\eq \const{cons}(\ti(A, \mathit{X}), \mathit{Xs})$ \\
$\forall A, \mathit{Xs}.\;\,
  \ti(\const{list}(A), \mathit{Xs}) \eq \const{nil}(A) \mathrel{\lor} {}$ \\
$\phantom{\forall A, \mathit{Xs}.\;\,}
  (\exists \mathit{Y}\!, \mathit{Ys}.\;\,
    \ti(A, \mathit{Y}) \eq \mathit{Y} \mathrel\land
    \ti(\const{list}(A), \mathit{Ys}) \eq \mathit{Ys} \mathrel\land \ti(\const{list}(A), \mathit{Xs}) \eq \const{cons}(\mathit{Y}\!, \mathit{Ys}))$ \\
$\forall A, \mathit{X}\!, \mathit{Xs}.\;\,
  \const{hd}(\const{cons}(\ti(A, \mathit{X}), \mathit{Xs})) \eq \ti(A, \mathit{X})
  \mathrel\land \const{tl}(\const{cons}(\ti(A, \mathit{X}), \mathit{Xs})) \eq \ti(\const{list}(A), \mathit{Xs})$ \\
$\exists \mathit{X}\!, \mathit{Y}\!, \mathit{Xs}, \mathit{Ys}.\;\,
  \ti(\const{\www}, \mathit{X}) \eq \mathit{X} \mathrel\land \ti(\const{\www}, \mathit{Y}) \eq \mathit{Y} \mathrel\land
  \ti(\const{list}(\const{\www}), \mathit{Xs}) \eq \mathit{Xs} \mathrel\land
  \ti(\const{list}(\const{\www}), \mathit{Ys}) \eq \mathit{Ys} \mathrel\land {}$ \\
$\phantom{\exists \mathit{X}\!, \mathit{Y}\!, \mathit{Xs}, \mathit{Ys}.\;\,}
\const{cons}(\mathit{X}\!, \mathit{Xs}) \eq \const{cons}(\mathit{Y}\!, \mathit{Ys})
\mathrel\land (\mathit{X} \not\eq \mathit{Y} \mathrel{\lor} \mathit{Xs} \not\eq \mathit{Ys})$
\end{quotex}
\end{exa}

\leftOut{%
In \longsect~\ref{sec:complete-monotonicity-based-encoding-of-polymorphism}, we
showed that it is sound to omit noninferable type arguments to constructors for
monotonicity-based encodings, yielding $\const{nil}$ instead of
$\const{nil}(A)$ in the translation.
For cover-based encodings, this would be unsound, as shown by the following
counterexample:
$\const{q}\LAN a\RAN(\const{nil}\LANX a\RAN) \mathrel\land \lnot\,\const{q}\LAN
b\RAN({\const{nil}\LANX b\RAN})$ is satisfiable, but the naive translation to
$\const q(\const{nil}) \mathrel\land \lnot\, \const q(\const{nil})$ is
unsatisfiable. The counterexample does not apply to monotonicity-based
encodings, because $\const q$ would be passed an encoded type argument:\
$\const q(\const a, \const{nil})\mathrel\land \lnot\, \const q(\const b, \const{nil})$
is satisfiable.
}

\begin{thm}[Correctness of \tags\at]\afterDot
\label{thm:correctness-of-tags-at}%
The cover-based type tags encoding\/ \hbox{\rm\tags\at} is correct.
\end{thm}
\begin{proof}
Let $\Sigma = (\KK,\FF,\PP)$ be the signature of a polymorphic problem $\PHI$.

\betweenitems\noindent
\textsc{Sound}:\enskip
Let $\MM = \bigl(\dom,(k^\MM)_{k \in \KK},(\ff^{\,\MM})_{\ff \in \FF},(\pp^\MM)_{\pp \in \PP}\bigr)$
be a model of $\PHI$.
By Lemma~\ref{lem-syntactic-domains}, we may assume that $\dom$ is disjoint from each of its elements, its elements are mutually disjoint,
and the type constructors $k^\MM$ satisfy distinctness and injectivity.

We define a structure $\NN = \bigl(\D',(\ff^{\,\NN})_{\ff \in \FF' \mathrel{\uplus} \KK' \mathrel{\uplus} \{\ti^2\}},(\pp^{\NN})_{\pp \in \PP'}\bigr)$
for $\SIGMA' = (\FF' \mathrel{\uplus} \KK' , \PP' \mathrel{\uplus} \{\is^2\})$ intended to be a model of $\encode{\PHI}{\tags\at\tinycomma\args^\ninf\tinycomma\erased}$.
The construction of $\NN$ proceeds very similarly to the case of guards from the proof of Theorem \ref{thm:correctness-of-guards-at}:\ $\D'$, $k^{\NN}$ for $k \in \KK'$,
and $s^{\NN}$ for $s \in \FF' \mathrel{\uplus} \PP'$ are all defined in the same way as in that proof.  It remains to define $\ti^{\NN} : \D' \times \D' \ra \D'$:
$$
\ti^{\NN}(a,b) =
\left\{
\begin{array}{@{}ll@{}}
 b & \mbox{if $a \in \dom$ and $b \in a$}
\\
 \eps{a} & \mbox{if $a \in \dom$ and $b \not\in a$}
\\
 \eps{\dom} & \mbox{if $a \not\in \dom$}
\end{array}
\right.
$$
In the proof of Theorem \ref{thm:correctness-of-guards-at}, $\is^{\NN}(a,b)$ captured the notion
that $a$ is a domain of $\MM$ and $b$ an element in it. The same can now be expressed by $\ti^{\NN}(a,b) = b$.
Additionally,
here $\ti^{\NN}$ is useful for redirecting any ill-typed element $b$
(one not in the first argument~$a$) to a well-typed element $\eps{a} \in a$.

To show that $\NN$ is a model of $\encode{\PHI}{\tags\at\tinycomma\args^\ninf\tinycomma\erased}$,
we proceed almost identically to the proof of Theorem \ref{thm:correctness-of-guards-at}. Starting
with a valuation $\xi' : \VV \ra \D'$ that respects types, we define $\xi$ and $\theta$ in the same way and state
facts (1)--(3) as for guards, but adding a condition about the ``well-typedness'' of the existential variables interpretation,
since only universal variables (i.e.\ variables in $\Vtun$) are tagged by \tags\at: 
\begin{enumerate}
\item[(1)] $\SEM{\typex{\sigma}}^{\NN}_{\xi'} = \SEM{\sigma}^\MM_\theta$;
\item[(2)] if $t \notin \Vt$ and $\xi(\mathit{X}) \in \SEM{\sigma}^\MM_\theta$ for all $\mathit{X}^\sigma \in \UV (t) \,\cup\, \Vtex$,
then $\SEM{\encode{t}{\tags\at\tinycomma\args^\ninf\tinycomma\erased}}^{\NN}_{\xi'} = \SEM{t}^\MM_{\theta,\xi}$;
\item[(3)] if $\xi'(\mathit{X}) \in \SEM{\sigma}^\MM_\theta$ for all $\mathit{X}^\sigma \in \UV (\phi) \,\cup\, \Vtex$,
then $\SEM{\phi}^\MM_{\theta,\xi}$ implies $\SEM{\encode{\phi}{\guards\at\tinycomma\args^\ninf\tinycomma\erased}}^{\NN}_{\xi'}$.
\end{enumerate}
Moreover,
(1$'$) from the proof of Theorem \ref{thm:correctness-of-guards-at} is replaced with the following two facts:
\begin{itemize}[label=(1'')]
\item[(1$'$)] $(\SEM{\ti(\typex{\sigma},\mathit{X})}^{\NN}_{\xi'} = \xi'(\mathit{X})) = (\xi'(\mathit{X}) \in \SEM{\sigma}^\MM_\theta)$;
\item[(1$''$)] $\SEM{\ti(\typex{\sigma},\mathit{X})}^{\NN}_{\xi'} \in \SEM{\sigma}^\MM_\theta\!$.
\end{itemize}
Again, (1)--(3) follow by induction on the involved type, term, or formula;
(1$'$) and (1$''$) follow from (1) and
the definition of $\ti^{\NN}$ (and are in turn used to establish (2) and (3)).
The fact that $\NN$ satisfies $\encode{\phi}{\tags\at\tinycomma\args^\ninf\tinycomma\erased}$ for each $\phi \in \PHI$ follows
in the same way as in Theorem \ref{thm:correctness-of-guards-at}.
It remains to show that
$\NN$ satisfies the necessary axioms, namely, the $\encode{\phantom{i}}{\args^\ninf\tinycomma\erased}$-translations
of the $\encode{\phantom{i}}{\tags\at}$ typing axioms. This follows from (1$'$)
and the definition of $f^{\NN}$.
\betweenitems\noindent
\textsc{Complete}:\enskip
This part is again similar to the corresponding one for guards.
By Theorem~\ref{thm:completeness-of-argsx}, it suffices to show $\encode{\phantom{i}}{\tags\at}$ complete.
Let $\NN = \bigl(\dom',(k^{\NN})_{k \in \KK},(\ff^{\,\NN})_{\ff \in \FF \mathrel{\uplus} \{\ti\}},(\pp^{\NN})_{\pp \in \PP}\bigr)$
be a model of $\encode{\PHI}{\tags\at}$, for which we may assume that the domains are mutually disjoint.
We define $\MM = \bigl(\dom,(k^\MM)_{k \in \KK},(\ff^{\,\MM})_{\ff \in \FF},(\pp^\MM)_{\pp \in \PP}\bigr)$
by restricting the domains not according to the guards as in the proof of Theorem \ref{thm:correctness-of-guards-at},
but according to the property that tags be identity:\
for each $D \in \dom'$, let $D' = \{d \in D \mid \ti^{\NN}(D)(d) = d \}$;
we define $\dom = \{D' \mid D \in \dom'\}$.  The other components of $\MM$ are defined as in Theorem \ref{thm:completeness-of-argsx},
and the proof is analogous to that for guards.
\leftOut{
Again, the typing axioms ensure that each $D'$ is nonempty and that the interpretations in $\MM$ are well defined,
again $\SEM{\gamma}^\MM_{\theta,\xi} = \SEM{\encode{\gamma}{\guards\at}}^{\NN}_{\theta,\xi}$,
where $\gamma$
is first a term and then a formula,
follow by induction on $\gamma$ (for arbitrary $\theta$ and $\xi$).  Again, we obtain that $\MM$ is a model of
$\encode{\PHI}{\guards\at}$ by the usual route.
}
\QED
\end{proof}

Unlike in the case of guards, the traditional type tag encoding \tags{} is not a
particular case of the cover-based encoding \tags\at, since only \tags\at{}
introduces typing axioms, and also \tags{} further restricts the type arguments
to phantom arguments. Nevertheless, we can provide a somewhat similar argument
for \tags{}'s correctness. Although the encoding is well known, we are not aware
of any soundness proof in the literature.

\medskip

\noindent
\emph{Proof of Theorem \ref{thm:correctness-of-tags}} (Correctness of \tags).

\betweenitems\noindent
\textsc{Sound}:\enskip
Let $\MM$ be as in the proof of Theorem \ref{thm:correctness-of-tags-at}.
We define $\NN$ similar to there, but with the following difference:\
$\D'$ contains not only the domains in $\dom$ and their elements, but also the set $P$ of ``polymorphic values'', i.e.\
functions that take a domain $D \in \dom$ and return an element of $D$.
The reason for the polymorphic values is that, due to the switch from noninferable arguments to phantom arguments,
the domain of the result for $f^{\NN}$
with $f : \forall\ov{\alpha}.\;\ov{\sigma} \ra \sigma$ in $\FF$
will no longer be completely inferable from the arguments, but will miss precisely the
nonphantom noninfrebale arguments, which correspond to type variables belonging to the result type $\sigma$ but not to the argument types $\ov{\sigma}$.
Consequently, the result of applying $f^{\NN}$ to ``well-typed'' arguments
will be polymorphic values that wait for the result domain $D$ and, if it has the
form $\SEM{\sigma}_{\theta}$ for an appropriate $\theta$, return the result from $\SEM{\sigma}_{\theta}$ according to the
interpretation of $f^{\NN}$.  Then the interpretation of the tag applied to polymorphic values
will select the desired element by providing the domain.
The reason why this approach works with \tags{} but not with \tags\at{} is that in \tags{} tags are applied everywhere,
thus resolving immediately any polymorphic value emerging from an application of $f^{\NN}$.

Formally,
we define the components of $\NN$ as follows.
%
First, $k^{\NN}$ and $p^{\NN}$ are as in the proofs of Theorems
\ref{thm:correctness-of-guards-at} and \ref{thm:correctness-of-tags-at}.
(For predicate symbols, the notions of phantom and noninferable type arguments coincide.)
$\D' = \dom \,\cup\, E \,\cup\, P$, where $E = \bigcup_{D \in \dom} D$ (the elements) and $P = \prod_{D \in \dom} D$ (the polymorphic values).
Moreover, $\ti^{\NN} : \D' \times \D' \ra \D'$ is defined as follows:
$$
\ti^{\NN}(a,b) =
\left\{
\begin{array}{@{}ll@{}}
 b & \mbox{if $a \in \dom$ and $b \in a$}
\\
 \eps{a} & \mbox{if $a \in \dom$, $b \in E$ and $b \not\in a$}
\\
 b(a) & \mbox{if $a \in \dom$ and $b \in P$}
\\
 \eps{\dom} & \mbox{otherwise}
\end{array}
\right.
$$
Assume $\ov{\alpha} = (\alpha_1,\ldots,\alpha_m)$,
$\ov{\gamma} = (\gamma_1,\ldots,\gamma_r)$,
$\ov{\beta} = (\beta_1,\ldots,\beta_n)$,
$\ov{\sigma} = (\sigma_1,\ldots,\sigma_u)$,
$\ov{\tau} = (\tau_1,\ldots,\tau_v)$,
and
$s : \forall (\ov{\alpha},\ov{\gamma},\ov{\beta}).\;(\ov{\sigma},\ov{\tau}) \ra \varsigma$ is in $\FF \mathrel{\uplus} \PP$,
such that
the first $m$ type arguments are phantom,
the middle $r$ arguments 
are noninferable nonphantoms,
the last $n$ type arguments are inferable,
and the first $u$ term arguments constitute $s$'s cover.
(Again, the general case with arbitrary permutations of the arguments
can be handled similarly.)
Then $s$ is an $(m+u+v)$-ary symbol in $\FF' \mathrel{\uplus} \PP'$.
Let $\ov{D} = (D_1,\ldots,D_m) \in {\D'}^m$, $\ov{d} = (d_1,\ldots,d_u) \in {\D'}^u$ and $\ov{e} = (e_1,\ldots,e_v) \in {\D'}^v$, and consider the
following condition on domains:
\begin{center}
$(D_1,\ldots,D_m) \in \dom^m$ and there exists $\ov{E} = (E_1,\ldots,E_n) \in \dom^n$ such that
each $d_i$ is in $\SEM{\sigma_i}_\theta^\MM$, where $\theta$ maps $(\ov{\alpha},\ov{\beta})$ to $(\ov{D},\ov{E})$
\end{center}
Assuming the condition holds, again by Lemma~\ref{lem-infer-domains} (taking $\ov{\alpha}$ from there to be $(\ov{\alpha},\ov{\gamma})$)
there exists precisely one tuple $\ov{E}$ satisfying it,
and we define the ``correction'' vector $\ov{e}'$ from $\ov{e}$ as in the proof of
Theorems \ref{thm:correctness-of-guards-at} and \ref{thm:correctness-of-tags-at}.
We now define
$$
f^{\NN}(\ov{D},\ov{d},\ov{e}) = \left\{
\begin{array}{@{}ll@{}}
 \pi_{\ov{D},\ov{d},\ov{e}} & \mbox{if the domain condition holds}
\\
 \eps{\dom} 
& \mbox{otherwise}
\end{array}
\right.
$$
where the polymorphic value $\pi_{\ov{D},\ov{d},\ov{e}} \in P$ is defined as follows:
$$
\pi_{\ov{D},\ov{d},\ov{e}}(D) = \left\{
\begin{array}{@{}ll@{}}
 s^{\MM}(\ov{D},\ov{G},\ov{E})(\ov{d},\ov{e}') & \mbox{if $D$ has the form $\SEM{\sigma}_\theta^\MM$ and $\theta$
assigns $(\ov{\alpha},\ov{\gamma},\ov{\beta})$ to $(\ov{D},\ov{G},\ov{E})$}
\\
 \eps{\dom} 
& \mbox{otherwise}
\end{array}
\right.
$$
To show that $\NN$ is a model of $\encode{\PHI}{\tags\tinycomma\args^\phan\tinycomma\erased}$,
we proceed almost identically to the proof of Theorem \ref{thm:correctness-of-tags-at}. Starting
with a valuation $\xi' : \VV \ra \D'$ that respects types (as in the proofs of Theorems \ref{thm:correctness-of-guards-at} and \ref{thm:correctness-of-tags-at}).
We define $\theta : \AAA \ra \dom$ by $\theta(\alpha) = \xi'(\VV(\alpha))$ and $\xi : \VV \ra \prod_{\sigma \in \Type_\Sigma} \SEM{\sigma}^\MM_\theta$ by
$\xi(\mathit{X})(\sigma) = \ti^{\NN}(\SEM{\sigma}^{\MM}_\theta,\xi'(\mathit{X}))$.
Facts (1)--(3) below follow by induction
on $\sigma$, $t$, or $\phi$ (for arbitrary $\xi'$): 
\begin{enumerate}
\item[(1)] $\SEM{\typex{\sigma}}^{\NN}_{\xi'} = \SEM{\sigma}^\MM_\theta$;
\item[(2)] $\SEM{\encode{t}{\tags\tinycomma\args^\phan\tinycomma\erased}}^{\NN}_{\xi'} = \SEM{t}^\MM_{\theta,\xi}$;
\item[(3)] $\SEM{\phi}^\MM_{\theta,\xi} = \SEM{\encode{\phi}{\tags\tinycomma\args^\phan\tinycomma\erased}}^{\NN}_{\xi'}$\!.
\end{enumerate}
It follows by the usual route that $\NN$ is a model of each $\encode{\phi}{\tags\tinycomma\args^\phan\tinycomma\erased}$ with $\phi \in \PHI$,
hence of $\encode{\PHI}{\tags\tinycomma\args^\phan\tinycomma\erased}$.
\betweenitems\noindent
\textsc{Complete}:\enskip
By Theorem \ref{thm:completeness-of-argsx}, it suffices to show $\encode{\phantom{i}}{\tags}$ complete.
Let $\NN = \bigl(\dom',(k^{\NN})_{k \in \KK},$ $(\ff^{\,\NN})_{\ff \in \FF \mathrel{\uplus} \{\ti\}},(\pp^{\NN})_{\pp \in \PP}\bigr)$
be a model of $\encode{\PHI}{\tags}$, for which we may assume that the domains are mutually disjoint.
We must construct a model $\MM = \bigl(\dom,(k^\MM)_{k \in \KK},(\ff^{\,\MM})_{\ff \in \FF},(\pp^\MM)_{\pp \in \PP}\bigr)$ of $\PHI$
by somehow removing the tags from $\NN$.
Unlike for \tags\at{}, here $\encode{\phantom{i}}{\tags}$ does not contain typing axioms ensuring that
a the restriction of $\NN$ to elements for which the tag interpretation is the identity forms a valid submodel; thus $\NN$
cannot be defined that way.
On the other hand, we know that the formulae in $\encode{\phantom{i}}{\tags}$ are fully tagged; in particular,
all variables are accessed through tags. This means we can take the domains of $\MM$ by restricting those of $\NN$ to the images
of the tag interpretations. The result will not be a submodel, so cannot take the functions and predicates $s^\MM$ of $\MM$ to be restrictions of those
of $\NN$. Instead, $s^\MM$ will be defined by applying $s^{\NN}$ through the tag ``interface''.

Formally, let, for each $D \in \dom'$, let $D'$ be the image of $\ti^{\NN}(D)$ (which is a subset of~$D$).
Each $D'$ uniquely determines its $D$.
We define the auxiliary function $\smash{T : \prod_{D \in \dom'} D \ra D}$, which only applies $\ti^{\NN}$ if the element is not already in the image:
$$
T_{D}(d) = \left\{
\begin{array}{@{}ll@{}}
 d & \mbox{if $d \in D'$} \\
 \ti^{\NN}(D)(d) & \mbox{otherwise}
\end{array}
\right.
$$
Notice that $T_D$ adds no tags around an existing tag:
$T_D(\ti^{\NN}(D)(d)) = \ti^{\NN}(D)(d)$.
%
We define the components of $\MM$ as follows:
%
$\dom = \{D' \mid D \in \dom'\}$;
$k^{\MM}(\ov{D'}) = k^{\NN}(\ov{D})$;
if $f : \forall\ov{\alpha}.\; \ov{\sigma} \ra \sigma$ is in $\FF$, then
$f^{\MM}(\ov{D'})(d_1,\ldots,d_n) =
\ti^{\NN}(\SEM{\sigma}^{\NN}_\theta)(f^{\NN}(\ov{D})(T_{\SEM{\sigma_1}^{\NN}_\theta}(d_1),\allowbreak\ldots,\allowbreak T_{\SEM{\sigma_n}^{\NN}_\theta}(d_n)))$,
where $\theta$ maps each $\alpha_i$ to $D_i$;
if $p : \forall\ov{\alpha}.\; \ov{\sigma} \ra \bool$ is in $\PP$, then
$p^{\MM}(\ov{D'})(d_1,\ldots,d_n) =
p^{\NN}(\ov{D})(T_{\SEM{\sigma_1}^{\NN}_\theta}(d_1),\ldots,T_{\SEM{\sigma_n}^{\NN}_\theta}(d_n))$,
where $\theta$ maps each $\alpha_i$ to $D_i$.
%
The interpretation $f^{\MM}$ applies tags $\ti^{\NN}$ at the top,
whereas the interpretations of $f^{\MM}$ and $p^{\MM}$ apply the modified tag $T$ at the bottom. This is to ensure
that the interpretations of terms and atomic formulae in $\MM$ do not add several consecutive layers of tags.
For example, when interpreted in $\MM$ within $f(\mathit{X}^\sigma)$, $f$ should add a tag around $\mathit{X}^\sigma$, as well as one around $f(\mathit{X}^\sigma)$;
however, when interpreted in $\MM$ within $f(g(\mathit{X}^\sigma))$,
$f$ should not add a tag around $g(\mathit{X}^\sigma)$, since $g$ already adds one.

Given $\theta : \AAA \ra \dom$ and
$\xi' : \VV \ra \prod_{\sigma \in \Type_{\Sigma'}} \SEM{\sigma}^{\NN}_{\theta'}$,
we define $\theta' : \AAA \ra \dom'$ by letting $\theta'(\alpha)$ be the unique $D \in \dom'$ such that $D' = \theta(\alpha)$,
and $\xi : \VV \ra \prod_{\sigma \in \Type_\Sigma} \SEM{\sigma}^{\MM}_\theta$ by
$\xi(\mathit{X})(\sigma) = \ti^{\NN}(\SEM{\sigma}^{\NN}_\theta)(\xi'(\mathit{X})(\sigma))$.
The next facts follow by induction on $\sigma$, $t$, or $\phi$ (for arbitrary $\theta$ and $\xi'$):
\begin{enumerate}
\item[(1)] $\SEM{\sigma}^{\MM}_\theta$ is in the image of $\ti^{\NN}(\SEM{\sigma}^{\NN}_{\theta'})$;
\item[(2)] $\SEM{t}^{\MM}_{\theta,\xi} = \SEM{\encode{t}{\tags}}^{\NN}_{\theta',\xi}$;
\item[(3)] $\SEM{\phi}^{\MM}_{\theta,\xi} = \SEM{\encode{\phi}{\tags}}^{\NN}_{\theta',\xi}$.
\end{enumerate}
Let $\phi \in \PHI$.  To show that $\MM$ is a model of $\phi$, let $\theta : \AAA \ra \dom$ and
$\xi : \VV \ra \prod_{\sigma \in \Type_\Sigma} \SEM{\sigma}^{\MM}_\theta$.  By (1),
there exists $\xi' : \VV \ra \prod_{\sigma \in \Type_{\Sigma'}} \SEM{\sigma}^{\NN}_{\theta'}$
such that $\xi(\mathit{X})(\sigma) = \ti^{\NN}(\SEM{\sigma}^{\NN}_\theta)(\xi'(\mathit{X})(\sigma))$ for all $\mathit{X}$ and $\sigma$---simply take
$\xi'(\mathit{X})(\sigma)$ to be any element $d'$ in $\SEM{\sigma}^{\NN}_{\theta'}$ such that $\ti^{\NN}(\SEM{\sigma}^{\NN}_{\theta'})(d') = \xi(\mathit{X})(\sigma)$.
Since $\SEM{\encode{\phi}{\tags}}^{\NN}_{\theta',\xi}$ holds, it follows from (3) that $\SEM{\phi}^{\MM}_{\theta,\xi}$, i.e.\
$\SEM{\phi}^{\MM}$, also holds, as desired.
\qed

\subsection{Polymorphic L\"owenheim--Skolem}

The previous completeness results can be used to prove
L\"owenheim--Skolem-like properties for polymorphic first-order logic by
appealing to the more manageable (and better known) monomorphic first-order
logic. 

\begin{lem} \afterDot
\label{lem-low-sko-aux-poly}
If a polymorphic $\Sigma$-problem $\PHI$ has a model, it also has a model $\MM = \bigl(\dom,\allowbreak
(k^\MM)_{k \in \KK},\allowbreak\_,\_\bigr)$ such that the following conditions are met\/{\rm:}\footnote{These
correspond to the
conditions from Lemma \ref{lem-syntactic-domains} plus countability of each $D \in \dom$.}
\begin{enumerate}
\item each $D \in \dom$ is countable\/{\rm;}
\item each $k^\MM$ is injective, and \smash{$k^\MM(\ov{D}) \not= {k'}^\MM(\ov{E})$} whenever $k \not= k'${\rm;}
\item $\dom = \{\SEM{\tau}^\MM \mid \tau \in \GType_\Sigma\}${\rm;}
\item the type interpretation function\/ $\SEM{\phantom{i}}^\MM$ is a bijection between\/ $\GType$ and\/ $\dom${\rm;}
\item $\dom$ is countable\/{\rm;}
\item $\dom$ is disjoint from each $D \in \dom$, and any distinct $D_1,D_2 \in \dom$ are disjoint.
\end{enumerate}
\end{lem}
\begin{proof}
%
%
Assume $\PHI$ has a model, and recall that by definition the polymorphic signatures that we consider
are countable.  By the soundness of \hbox{\rm\guards\at} (Theorem \ref{thm:correctness-of-guards-at}),
$\PHIii{\guards\at}$ also has a model and its signature is countable; hence by classical
L\"owenheim--Skolem \cite{hodges-1993}, $\PHIii{\guards\at}$ has a countable model $\NN$. It follows
from the proof of completeness of \hbox{\rm\guards\at} that $\PHI$ has a model
$\MM = (\dom,\_,\_,\_)$ constructed from $\NN$ in such a way
that 
each $D \in \dom$ is countable---hence $\NN$ satisfies (1).
%
%
Now, (2)--(6) follow by applying the same reasoning as in Lemma \ref{lem-syntactic-domains} and noticing that
all the structure modifications from there do not alter the countability of the domains $D \in \dom$.
\QED
\end{proof}

\begin{lem}\afterDot
\label{lem-low-sko-poly} \label{lem:downward-loewenheim-skolem}
If\/ $\PHI$ has a model $\MM = (\dom,\_,\_,\_)$ where all\/ $\SEM{\tau}^\MM$ are infinite for all\/ $\tau \in \GType$,
it also has a model $\NN = (\dom',\_,\_,\_)$ where\/ $\dom'$
and all\/ $D' \in \dom'$ are countably infinite.
\end{lem}
\begin{proof}
Let $\PHI'$ be $\PHI$ extended with axioms $\Ax =
\{ \exists X_1 : \t,\,\ldots,\,X_n : \t.\;\, \bigwedge_{i<j} X_i \not\eq X_j \mid\allowbreak \t \in \GType_\Sigma,\allowbreak\, n \in \mathbb{N} \}$
that state that all ground types are infinite. Then $\MM =
(\dom,\_,\_,\_)$ is a model of $\PHI'$.
By Lemma~\ref{lem-low-sko-aux-poly}, $\PHI'$ also has a model $\NN = (\dom',\_,\_,\_)$
with $\dom' = \{\SEM{\tau}^{\NN} \mid \tau \in \GType_\Sigma\}$
(which is, in particular, countably infinite)
and each $D' \in \dom'$ {\relax at most} countably infinite, i.e.\
due to $\Ax$ and the fact that each $D' \in \dom$ has the form $\SEM{\tau}^{\NN}$, {\relax actually} countably infinite.
$\NN$ is the desired model of $\PHI$.
\QED
\end{proof}

\section{Monotonicity-Based Type Encodings --- The Monomorphic Case}
\label{sec:monotonicity-based-type-encodings-the-monomorphic-case}

The cover-based encodings of
Section~\ref{sec:alternative-cover-based-encoding-of-polymorphism} removed some of the
clutter associated with type arguments, which are in general necessary to
encode polymorphism soundly. Another family of encodings focuses on the quantifiers
and exploit monotonicity.
For types that are inferred monotonic, 
the translation can omit the type information on term variables of these types.
Informally, a type is monotonic in a problem when, for any model of that
problem, we can increase the size of the domain that interprets the type while
preserving satisfiability.

This section focuses on the monomorphic case, where the input problem contains
no type variables or polymorphic symbols. This case is interesting in its own
right and serves as a stepping stone towards polymorphic monotonicity-based
encodings (\longsect~\ref{sec:complete-monotonicity-based-encoding-of-polymorphism}).

Before we start, we need to define variants of the traditional \erased{}, \tags{}, and \guards{}
encodings that operate on monomorphic problems.
Since the monomorphic \erased{} is essentially identical to the polymorphic \erased{} restricted to
monomorphic signatures, we use the same notation for both.
The monomorphic encodings \mono\tags{} and \mono\guards{}
coincide with \tags{}~and~\guards{}
except that the polymorphic function $\ti\LAN\t\RAN(t)$ and
predicate $\is\LAN\t\RAN(t)$ are replaced by type-indexed families of unary
functions $\ti_{\,\t}(t)$ and predicates $\is_{\,\t}(t)$, 
as is customary in the literature \cite[\S4]{wick-mccune-1989}.

\subsection{Monotonicity}
\label{ssec:monotonicity-monomorphic}

The concept of monotonicity used by Claessen et al.\
\cite[\S2.2]{claessen-et-al-2011}
declares a type $\tau$ monotonic for a finite problem $\PHI$ if for any model $\MM = \bigl((\D_\sigma)_{\sigma \in \Type},\_,\_\bigr)$ of $\PHI$
such that $\D_\tau$ is finite, there exists another model $\NN = \bigl((\D'_\sigma)_{\sigma \in \Type},\_,\_\bigr)$ of $\PHI$
such that $\left|\smash{\D'_\tau}\right| = \left|\smash{\D_\tau}\right| + 1$ and
$\left|\smash{\D'_\sigma}\right| = \left|\smash{\D_\sigma}\right|$ for all $\sigma \not= \tau$.
%
%
\leftOut{
\begin{defi}[Finite Monotonicity]\afterDot
\label{def:finite-monotonicity}%
Let $\t$ be a ground type and $\PHI$ be a 
problem. The type $\t$ is \emph{\vthinspace finitely monotonic} in $\PHI$ if for
all models
$\mathcal M$ of $\PHI$ such that $\SEM{\t}^{\mathcal M}$ is finite,
there exists a model $\mathcal M'$ where $\left|\smash{\SEM{\t}^{\mathcal M\smash{'}}}\right|
= \left|\smash{\SEM{\t}^{\mathcal M}}\right| + 1$ and
$\left|\smash{\SEM{\u}^{\mathcal M\smash{'}}}\right|
= \left|\smash{\SEM{\u}^{\mathcal M}}\right|$ for all $\u \not= \t$.
The problem $\PHI$ is \emph{\vthinspace finitely monotonic}
if all its types are monotonic.
\end{defi}
}
Their notion, which we call {\em finite monotonicity},
is designed to ensure that it is possible to produce a model having all types interpreted by countably infinite sets,
and finally a model having all types interpreted as the same set, so that type information can be soundly erased.
In this article, we take directly the infinite-interpretation property as
definition of monotonicity and extend the notion to sets of types:
%

\begin{defi}[Monotonicity]\afterDot
\label{def:infinite-monotonicity}%
Let $S$ be a set of 
types and $\PHI$ be a 
problem.
The set $S$ is 
\emph{monotonic} in
$\PHI$ if for all models $\MM = \bigl((\D_\sigma)_{\sigma \in \Type},\_,\_\bigr)$ of $\PHI$,
there exists a model $\NN = \bigl((\D'_\sigma)_{\sigma \in \Type},\_,\_\bigr)$ of $\PHI$  such that
for all 
types $\sigma$,
$\D'_\sigma$ is infinite if $\sigma \in S$ and
$\left|\smash{\D'_\sigma}\right| = \left|\smash{\D_\sigma}\right|$ otherwise.
A type $\t$ is 
\emph{monotonic} if $\{\t\}$ is monotonic.
The problem~$\PHI$ is 
\emph{monotonic} if the set $\Type$ of all types is monotonic.
\end{defi}

Full type erasure is sound for monotonic monomorphic problems.
The intuition is that a model of such a problem can be extended
into a model where all types are interpreted as sets of the same cardinality,
which can be merged to yield an untyped model. 

\begin{exa}%
\label{ex:monkey-village-monotonicity}%
The monkey village of Example~\ref{ex:monkey-village} is monotonic
because any model with finitely many bananas can be extended to a model with
infinitely many, and any model with infinitely many bananas and finitely many
monkeys can be extended to one where monkeys and bananas have the same infinite
cardinality (cf.\ Example~\ref{ex:monkey-village-erased}). However, the type of
monkeys is not monotonic on its own, as we argued intuitively in
Example~\ref{ex:monkey-village}, nor is it finitely monotonic.
\end{exa}



\leftOut{
\begin{lem}[Same-Cardinality Erasure]\afterDot
\label{lem:same-cardinality-erasure}%
Let $\PHI$ be a monomorphic problem.
If $\PHI$ has a model where all domains have cardinality $k$,
$\erasedx{\PHI}$ has a model where the unique domain has cardinality $k$.
\end{lem}
\begin{proof}
See Theorem~1 in Bouillaguet et al.~\cite[\S4]{bouillaguet-et-al-2007}
or Lemma~1 in Claessen et al.\ \cite[\S1]{claessen-et-al-2011}.\,\hbox{}\QED
\end{proof}
}

\begin{thm}[Monotonic Erasure]\afterDot
\label{thm:monotonic-erasure-monomorphic}%
Full type erasure is sound for monotonic monomorphic problems.
\end{thm}

\begin{proof}
Let\/ $\PHI$ be such a problem, let $\Sigma = (\Type,\FF,\PP)$ be its signature, and let $\Sigma' = (\FF',\PP')$
be the target signature of \erased{}.
If $\PHI$ is satisfiable, it has a model where all domains are infinite
by definition of monotonicity.
Since $\Sigma$ is countable, by the downward L\"owenheim--Skolem theorem \cite{hodges-1993}, $\PHI$ also has a model
where all domains are countably infinite, hence also a model
$\MM = \bigl((\D_\sigma)_{\sigma \in \Type},(\ff^{\,\MM})_{\ff \in \FF},(\pp^\MM)_{\pp \in \PP}\bigr)$
where all domains are interpreted as the same set $D$, i.e.\ each $\D_\sigma$ is $D$.
We define the $\Sigma'$-structure $\NN = \bigl(\D,(\ff^{\,\NN})_{\ff \in \FF'},(\pp^{\NN})_{\pp \in \PP'}\bigr)$
by taking $\D$ to be $D$ and each $s^{\NN}$ to be $s^\MM$.
For each $\xi' : \VV \ra D$, we define $\xi : \VV \ra \prod_{\sigma \in \Type} \D_\sigma$,
i.e.\ $\xi : \VV \ra \Type \ra D$, by $\xi(\mathit{X})(\sigma) = \xi'(\mathit{X})$.
The next facts follow by induction on $t$ and $\phi$ (for arbitrary $\xi$):
\begin{enumerate}
\item[(1)] $\SEM{\encode{t}{\erased}}^{\NN}_{\xi'} = \SEM{t}^\MM_{\xi}$;
\item[(2)] $\SEM{\encode{\phi}{\erased}}^{\NN}_{\xi'} = \SEM{\phi}^\MM_{\xi}$\!.
\end{enumerate}
It follows by the usual route that $\NN$ is a model of $\encode{\PHI}{\erased}$.
%
\QED
\end{proof}

\begin{rem}
An alternative to invoking the L\"owenheim--Skolem theorem
would have been to require countable infinity in the definition of monotonicity.
Although it makes no difference in this article, the more general definition
would also apply in the presence of uncountable interpreted types (e.g.\ for the
real numbers).
The proof of
Theorem~\ref{thm:monotonic-erasure-monomorphic} can be adapted to go beyond
countable infinity if desired.
\end{rem}

It is often convenient to determine the monotonicity of single types separately,
viewed as singleton sets.
This is enough to make the set of all types, i.e.\ the problem, monotonic:

\begin{lem}[Global Monotonicity from Separate Monotonicity]\afterDot
\label{lem:monotonicity-preservation-by-union}%
If $\sigma$ is monotonic in $\PHI$ for each type $\sigma$, then $\PHI$ is monotonic.
\end{lem}

%
\begin{proof}
Let $\Sigma = (\Type,\FF,\PP)$ be the signature of $\PHI$, and assume $\PHI$ is satisfiable.
Let $\sigma_1,\sigma_2,\ldots$ 
be an enumeration of $\Type$, and let $\phi_k$ be a formula stating that for all $i \in \{1,\ldots,k\}$,
$\sigma_i$'s domain has at least $k$ elements. Finally,
let $\PHI' = \{\phi_i \mid i \in \mathbb{N}\}$.
By the monotonicity of the individual types,
it follows by induction on $k$ that each $\PHI \,\cup\, \{\phi_i \mid i \in \{1,\ldots,k\}\}$ is satisfiable,
and hence that each finite subset of $\PHI \,\cup\, \PHI'$ is satisfiable.
By the compactness theorem \cite{hodges-1993}, it follows that $\PHI \,\cup\, \PHI'$ is itself satisfiable, and hence
$\PHI$ has a model with only infinite domains. 
\QED
\end{proof}

\leftOut{
\begin{lem}[Monotonicity Preservation by Union]\afterDot
\label{lem:monotonicity-preservation-by-union}%
Let\/ $\mathcal S$ be a set of sets of ground types and\/ $\PHI$ be a problem.
If every\/ $S \in \mathcal S$ is monotonic in\/ $\PHI$, then\/
$\bigcup \mathcal S$ is monotonic in~$\PHI$.
\end{lem}

\begin{proof}
If $\mathcal S$ is empty or has 1 element, the lemma is trivial.

Suppose $\left|\smash{\mathcal S}\right| = 2$, that is $\mathcal S
= \left\{\t, \u\right\}$. Then given a model of $\PHI$, we first use
$\t$'s monotonicity to get a model where $\SEM{\t}$ is infinite, and
all other domains are unchanged. Then we use $\u$'s monotonicity
to get a model where $\SEM{\t}$ and $\SEM{\u}$ are infinite, and
all other domains are unchanged. This is the required model. By
induction, this establishes the lemma for finite $\mathcal S$.

If $\mathcal S$ is infinite, we use the same compactness argument
as in Theorem~\ref{thm:subsumption-of-finite-monotonicity}.
Suppose $\mathcal M$ is a model of $\PHI$.
Let the set of formulae $K$ consist of, for all types
$\t \in \mathcal S$, the statement that $\SEM{\t}$ is infinite, and
for all types $\t \not\in \mathcal S$, the statement that the
cardinality of $\SEM{\t}$ is the same as in $\mathcal M$.

To show that S is monotonic in $\PHI$, we need to show that
$\PHI \cup K$ is satisfiable. Since any finite subset of
$\mathcal S$ is monotonic in $\PHI$, the union of $\PHI$ and any
finite subset of $K$ is satisfiable. Hence, by compactness,
$\PHI \cup K$ is satisfiable. \QED
\end{proof}
}

\subsection{Monotonicity Inference}
\label{ssec:monotonicity-inference-monomorphic}

Claessen et al.\ introduced a simple calculus to infer finite monotonicity for
monomorphic first-order logic \cite[\S2.3]{claessen-et-al-2011}. The
definition below generalises it from clause normal form to negation normal form.
The generalisation is straightforward; we present it because we later
adapt it to polymorphism.
The calculus is based on the observation that a type $\t$ must be
monotonic if the problem expressed in NNF contains no
positive literal of the form $\mathit{X}\typ{\t} \eq t$ or $t \eq \mathit{X}\typ{\t\!}$,
where $\mathit{X}$ is universal. We call such an occurrence of $\mathit{X}$ a naked occurrence.
Naked variables are unavoidable to express upper bounds on the cardinality of
types in first-order logic.

\begin{defi}[Naked Variable]\afterDot
\label{def:naked-variable}%
The set of \emph{naked variables} $\NV(\phi)$
of a formula $\phi$ is defined as follows:
\begin{align*}
\NV(p(\bar t\,)) & \,=\, \emptyset &
  \NV(t_{1\!} \eq t_2) & \,=\, \{t_1, t_2\} \mathrel\cap \Vt \\
\NV(\lnot\,p(\bar t\,)) & \,=\, \emptyset &
  \NV(t_{1\!} \not\eq t_2) & \,=\, \emptyset \\
\NV(\phi_{1\!} \mathrel\land \phi_2) & \,=\, \NV(\phi_1) \mathrel\cup \NV(\phi_2) &
  \NV(\forall \mathit{X}\mathbin:\t.\;\, \phi) & \,=\, \NV(\phi) \\
\NV(\phi_{1\!} \mathrel\lor \phi_2) & \,=\, \NV(\phi_1) \mathrel\cup \NV(\phi_2) &
  \NV(\exists \mathit{X}\mathbin:\t.\;\, \phi) & \,=\, \NV(\phi) - \left\{ \mathit{X}^\sigma \right\}
%
\end{align*}
For a problem $\PHI$, we define $\NV(\PHI) = \bigcup_{\phi \in \PHI} \NV(\phi)$.
%
\end{defi}

%
We see from the $\exists$ case that existential variables never occur naked in
sentences.

Variables of types other than $\t$ are irrelevant when inferring whether
$\t$ is monotonic; a variable is problematic only if it occurs naked and
has type $\t$. Annoyingly, a single naked variable of type $\t$,
such as $\mathit{X}$ on the right-hand side of the equation
$\const{hd}_{\www}(\const{cons}_{\vvthinspace\www}(\mathit{X}\!, \mathit{Xs})) \eq \mathit{X}$
from Example~\ref{ex:algebraic-lists-monomorphised},
will cause us to classify $\t$ as possibly nonmonotonic.
We regain some precision by extending the calculus with an
infinity analysis, as suggested by Claessen et al\hbox{.}:\ trivially, all types
with no finite models are monotonic.
%
\newcommand{\rhdbf}{{\rhd}\hspace{-1.71ex}\lower.01ex\hbox{$\rhd$}\hspace{-1.685ex}\raise.01ex\hbox{$\rhd$}}

\begin{conv}\label{inftypes-mono}
Abstracting over the specific analysis used to detect infinite types
(e.g.\ Infinox \cite{claessen-lilliestrom-2011}), we fix a set $\Inf(\PHI)$ of
types whose interpretations are guaranteed to be infinite in all models of
$\PHI$, More precisely, the following property is assumed to hold:\
if $\tau \in \Inf(\PHI)$ and $\MM = \bigl((\D_\sigma)_{\sigma \in \Type},\_,\_\bigr)$ is a model of $\PHI$,
then $\D_\tau$ is infinite.
\end{conv}

Our monotonicity calculus takes $\Inf(\PHI)$ into account:


\begin{defi}[Monotonicity Calculus $\MONO{}$]\afterDot
\label{def:monotonicity-calculus-monomorphic}%
Let $\PHI$ be a monomorphic
problem over $\SIGMA = (\mathcal K, \mathcal F, \mathcal P)$.
A judgement $\MONO{\t} \PHI$ indicates
that the ground type $\t$ is inferred monotonic in~$\PHI$.
The \emph{monotonicity calculus} consists of the following rules:\strut

\vskip\abovedisplayskip

\noindent
\hfill
\AXC{$\strut\t \in \Inf(\PHI)$\strut}
\UIC{$\MONO{\t} \PHI$\strut}
\DP
\btwrules
\AXC{$\NV(\PHI) \mathrel\cap \mathcal \{\mathit{\mathit{X}^\sigma} \mid \mathit{X} \in \VV\} = \emptyset$\strut}
\UIC{$\MONO{\t} \PHI$\strut}
\DP
\hfill
\hbox{}

\vskip\abovedisplayskip

\noindent
We write $\MONO{\t} \PHI$ to indicate
that the judgement is derivable and $\NONMONO{\t} \PHI$ otherwise.
\end{defi}

Claessen et al.\ designed a second, more
powerful calculus that extends their first calculus to detect predicates that
act as guards for naked variables.
Whilst the calculus proved successful on a
subset of the TPTP benchmarks \cite{sutcliffe-tptp},
we assessed its suitability on about 1000 problems generated by Sledgehammer
and found no improvement on the simple calculus.
For this reason, we restrict our attention to the first calculus.

\begin{thm}[Soundness of $\MONO{}$]\afterDot
\label{thm:soundness-of-calculus-monomorphic}%
Let\/ $\PHI$ be a monomorphic problem. If\/ $\MONO{\tau} \PHI$, then\/ $\tau$ is
monotonic in\/ $\PHI$.
\end{thm}

\begin{proof}
Let $\Sigma = (\Type,\FF,\PP)$ be the signature of $\PHI$, and let
$\tau \in \Type$ such that $\MONO{\tau} \PHI$.
Assume $\PHI$ has a model 
$\MM = \bigl((\D_\sigma)_{\sigma \in \Type},(\ff^{\,\MM})_{\ff \in \FF},(\pp^\MM)_{\pp \in \PP}\bigr)$.
We will construct a model $\NN = \bigl((\D'_\sigma)_{\sigma \in \Type},(\ff^{\,\NN})_{\ff \in \FF},(\pp^{\NN})_{\pp \in \PP}\bigr)$
with the same domains as $\MM$ for $\sigma \not= \tau$
and with $\D'_\tau$ infinite.  If $\tau \in \Inf(\Phi)$, then $\D_\tau$ is already infinite, so can we take $\NN = \MM$.
Otherwise, $\tau \not\in \Inf(\Phi)$, which means the second rule of the calculus must have been applied
and no variables of type $\tau$ occur naked,
i.e.\
$\NV(\PHI) \mathrel\cap \mathcal \{\mathit{\mathit{X}^\tau} \mid \mathit{X} \in \VV\} = \emptyset$.
We define $\D'_\tau$ by extending $\D_\tau$ with an infinite number of
fresh elements $e_1,e_2,\ldots$
and define the functions and predicates of $\NN$ to treat these in the same
way as $\eps{\D_\tau}$.
Intuitively, $\NN$ also satisfies $\PHI$ because,
since no variables of type $\tau$ occur naked, the formulae of $\PHI$ cannot tell the difference
between a ``clone'' $e_i$ and the original $\eps{\D_\tau}$ except for the case of disequality---and there the formula
instantiated with clones is more likely to be true than the one instantiated with $\eps{\D_\tau}$
(since clones are mutually disequal and disequal to $\eps{\D_\tau}$).

Formally, we let $E = \{e_1,e_2,\ldots\}$ be an infinite set disjoint from $\D_\tau$, and define
$$\D'_\sigma = \left\{
\begin{array}{@{}ll@{}}
 \D_\sigma \,\cup\, E & \mbox{if $\sigma = \tau$}
\\
 \D_\sigma & \mbox{otherwise}
\end{array}
\right.$$
Given $s : \ov{\sigma} \ra \varsigma$, we let
$s^{\NN}(\ov{d}) = s^{\MM}(\ov{d}')$,
where
$$
d'_i =
\left\{
\begin{array}{@{}ll@{}}
 \eps{\D_\tau} & \mbox{if $\sigma_i = \tau$ and $d_i \in E$}
\\
 d_i & \mbox{otherwise}
\end{array}
\right.$$
Given $\xi' : \VV \ra \prod_{\sigma \in \Type} \D'_\sigma$, we define
$\xi : \VV \ra \prod_{\sigma \in \Type} \D_\sigma$ by
$$
\xi(\mathit{X})(\sigma)(d) =
\left\{
\begin{array}{@{}ll@{}}
 \eps{\D_\tau} & \mbox{if $\sigma = \tau$ and $d \in E$}
\\
 \xi'(\mathit{X})(\sigma)(d) & \mbox{otherwise}
\end{array}
\right.
$$
The next facts follow by induction on $t$ or $\phi$ (for arbitrary $\xi'$):
\begin{enumerate}
\item[(1)] $t \not\in \{\mathit{X}^\tau \mid \mathit{X} \in \VV\}$ implies $\SEM{t}^{\NN}_{\xi'} = \SEM{t}^{\MM}_{\xi}$;
\item[(2)] $\NV(\phi) \mathrel\cap \mathcal \{\mathit{\mathit{X}^\sigma} \mid \mathit{X} \in \VV\} = \emptyset$ and
$\SEM{\phi}^{\MM}_{\xi}$ imply $\SEM{\phi}^{\NN}_{\xi'}\!$.
\end{enumerate}
(Notice that (2) is an implication, not an equivalence. An inductive proof of the converse would fail
for the case of disequalities.)
In particular, thanks to the absence of naked variables of type $\tau$ in all $\phi \in \PHI$, $\NN$ is the desired model of $\PHI$.
\QED
\end{proof}

In the light of the above soundness result, we will
allow ourselves to
write that $\t$ is monotonic if $\MONO{\t} \PHI$ and possibly nonmonotonic if
$\NONMONO{\t} \PHI$.

\subsection{Encoding Nonmonotonic Types}
\label{ssec:encoding-nonmonotonic-types}

Monotonic types can be soundly erased when
translating to untyped first-order logic,
by Theorem~\ref{thm:monotonic-erasure-monomorphic}.
Nonmonotonic types in general cannot.
Claessen et al.\ \cite[\S3.2]{claessen-et-al-2011} point out that adding
sufficiently many protectors to a nonmonotonic problem will make it monotonic,
at which point its types can be erased.
Thus the following general two-stage procedure translates monomorphic problems
to untyped first-order logic:
\begin{enumerate}[label=\arabic*.]
\item[1.]\strut Selectively introduce protectors (tags or guards) without erasing any types:
\begin{enumerate}[label=1.\arabic*.]
  \item[1.1.]\strut Infer monotonicity to identify the possibly nonmonotonic types
            in the problem.
  \item[1.2.]\strut Introduce protectors for the universal variables of possibly nonmonotonic types.
  \item[1.3.]\strut If necessary (depending on the encoding), generate \relax{typing axioms} for any function symbol whose
            result type is possibly nonmonotonic, to make it possible to remove
            protectors for terms with the right type.
\end{enumerate}

\betweenitems

\item[2.]\strut Erase all the types.
\end{enumerate}

The purpose of stage~1 is to make the problem monotonic while preserving
satisfiability. This paves the way for the sound type erasure of
stage~2. The following lemmas will help us prove such two-stage encodings
correct.

\begin{lem}[Correctness Conditions]\afterDot
\label{lem:correctness-conditions-monomorphic}%
Let\/ $\PHI$ be a monomorphic problem, and let\/ $\xx$ be a monomorphic encoding.
The problems\/ $\PHI$\/ and\/ $\encode{\PHI}{\xx\tinycomma\erased}$ are
equisatisfiable provided that the following conditions hold\/{\rm:}
\begin{itemize}
\item[] \textsc{Mono}{\rm:}\enskip $\PHIii\xx$\/ is monotonic.

\betweenitems

\item[] \textsc{Sound}{\rm:}\enskip If\/ $\PHI$ is satisfiable, so is\/ $\PHIii\xx$.

\betweenitems

\item[] \textsc{Complete}{\rm:}\enskip If\/ $\PHIii\xx$ is satisfiable, so is\/ $\PHI$.
\end{itemize}
\end{lem}
\begin{proof}
Immediate from Theorems \ref{thm:completeness-of-erased} and
\ref{thm:monotonic-erasure-monomorphic}.\QED
\end{proof}

\leftOut{
\begin{lem}[Submodel]\afterDot
\label{lem:submodel}%
Let $\PHI$ be a 
problem, let\/ $\mathcal M$ be a model of $\PHI$, and let\/ $\mathcal M'$ be a substructure of $\mathcal M$
{\upshape(}i.e.\ a structure constructed from\/ $\mathcal M$ by removing some domain
elements while leaving the interpretations of functions and predicates intact
for the remaining elements{\upshape)}. This\/ $\mathcal M'$ is a model of\/ \PHI{}
provided that it does not remove any witness for an existential variable.
\end{lem}
\begin{proof}
Suppose $\PHI$ is skolemised. If it is \emph{false} in $\mathcal M'$
then there must be a valuation---an assignment of domain elements to
the variables of $\PHI$---which makes it false. This same valuation
makes $\PHI$ false in $\mathcal M$.

If $\PHI$ is not skolemised, let $\Psi$ be its skolemisation.
From $\mathcal M$ we can construct a model $\mathcal M_\Psi$ of $\Psi$
by giving interpretations to all the Skolem functions in $\Psi$.
We remove domain elements from $\mathcal M$ to get $\mathcal M'$;
we can remove the same domain elements from $\mathcal M_\Psi$
to get $\mathcal M'_\Psi$. This gives us a substructure of $\mathcal M_\Psi$
as long as we do not remove any domain elements that are in the image
of a Skolem function in $\mathcal M_\Psi$, which is to say, we do not
remove any witness for an existential variable in $\mathcal M$.
By the result above, $\mathcal M'_\Psi$ is a model of $\Psi$,
hence $\mathcal M'$ is a model of $\PHI$.
\end{proof}
}

\subsection{Monotonicity-Based Type Tags}
\label{ssec:monotonicity-based-type-tags-monomorphic}

The monotonicity-based encoding \mono\tags\query{}
specialises the above
procedure for tags. It is similar to the traditional encoding
\mono\tags{} (the monomorphic version of \tags),
except that it omits the tags for types that are inferred
monotonic. By wrapping all naked variables (in fact, all terms) of possibly
nonmonotonic types in a function term, stage~1 yields a monotonic problem.

\begin{defi}[Lightweight Tags \monobf\tags\query{}]\afterDot
The encoding $\encode{\phantom{i}}{\mono\tags\query}$ translates monomorphic problems over
$\SIGMA = (\Type, \FF, \PP)$ to monomorphic problems over $(\Type, \FF \mathrel{\uplus} \{\ti_\sigma : \sigma \ra \sigma \mid \sigma \in \Type \},\allowbreak \PP)$.
It adds no axioms, and its term and formula translations are defined as follows:
\begin{align*}
\mtagsqx{f(\bar t\,)} & \,=\, \CT{}{f(\mtagsqx{\bar t\,})} &
\mtagsqx{\mathit{X}} & \,=\, \CT{}{\mathit{X}} &
\text{with}\;\,\CT{}{t\typ\t} & \,=\, {\begin{cases}
  t & \!\!\text{if $\MONO{\t} \PHI$} \\
  \ti_{\,\t}(t) & \!\!\text{otherwise}
\end{cases}}
\end{align*}
The \emph{monomorphic lightweight type tags} encoding \mono\tags\query{}
is the composition $\encode{\phantom{i}}{\mono\tags\query\tinycomma\erased}$.
It translates a monomorphic problem over $\SIGMA$
into an untyped problem over $\SIGMA' = (\FF'
\uplus \{\ti^1_{\,\t} \mid \sigma \in \Type\}, \PP')$,
where $\mathcal{F}', \Pp$ are as for \erased.
\end{defi}

\begin{exa}%
\label{ex:algebraic-lists-mono-tags-query}%
For the algebraic list problem of Example~\ref{ex:algebraic-lists-monomorphised},
the type $\mathit{list\IUS}\www$
is monotonic by virtue of being infinite, whereas $\www$
cannot be inferred monotonic. The \mono\tags\query{} encoding of the problem
follows:
\begin{quotex}
$\forall \mathit{X}\!, \mathit{Xs}.\;\,
  \const{nil}_{\www} \not\eq \const{cons}_{\vvthinspace\www}(\ti_{\vvthinspace\www}(\mathit{X}), \mathit{Xs})$ \\
$\forall \mathit{Xs}.\;\,
  \mathit{Xs} \eq \const{nil}_{\www} \mathrel{\lor}
  (\exists \mathit{Y}\!, \mathit{Ys}.\;\,
  \mathit{Xs} \eq \const{cons}_{\vvthinspace\www}(\ti_{\vvthinspace\www}(\mathit{Y}), \mathit{Ys})
  )$ \\
$\forall \mathit{X}\!, \mathit{Xs}.\;\,
  \ti_{\vvthinspace\www}(\const{hd}_{\www}(\const{cons}_{\vvthinspace\www}(\ti_{\vvthinspace\www}(\mathit{X}), \mathit{Xs}))) \eq \ti_{\vvthinspace\www}(\mathit{X})
\mathrel\land
  \const{tl}_{\www}(\const{cons}_{\vvthinspace\www}(\ti_{\vvthinspace\www}(\mathit{X}), \mathit{Xs})) \eq \mathit{Xs}$ \\
$\exists \mathit{X}\!, \mathit{Y}\!, \mathit{Xs}, \mathit{Ys}.\;\,
  \const{cons}_{\vvthinspace\www}(\ti_{\vvthinspace\www}(\mathit{X}), \mathit{Xs}) \eq \const{cons}_{\vvthinspace\www}(\ti_{\vvthinspace\www}(\mathit{Y}), \mathit{Ys}) \mathrel\land (\ti_{\vvthinspace\www}(\mathit{X}) \not\eq \ti_{\vvthinspace\www}(\mathit{Y}) \mathrel\lor \mathit{Xs} \not\eq \mathit{Ys})$
\end{quotex}
\end{exa}

The \mono\tags\query{} encoding treats all variables of the same type uniformly.
Hundreds of axioms can suffer because of one unhappy formula that uses a type
nonmonotonically (or in a way that cannot be inferred monotonic).
To address this, we introduce a lighter encoding:\
if a universal variable does not occur naked in a formula, its tag
can safely be omitted.%
\footnote{This is related to the observation that only
paramodulation from or into a variable can cause ill-typed instantiations in a
resolution prover \cite[\S4]{wick-mccune-1989}.}

This new encoding, called \mono\tags\qquery{}, protects only naked variables
and introduces equations $\ti_{\,\t}(\ff(\mathit{\bar X})\typ\t) \eq
\ff(\mathit{\bar X})$ to add or remove tags around each function
symbol $\ff$ whose result type $\t$ is possibly nonmonotonic, and
similarly for existential variables.

\begin{defi}[Featherweight Tags \monobf\tags\qquery]\afterDot
The encoding $\encode{\phantom{i}}{\mono\tags\qquery}$
translates monomorphic problems over
$\SIGMA = (\Type, \FF, \PP)$ to monomorphic problems over $(\Type, \FF \mathrel{\uplus} \{\ti_\sigma : \sigma \ra \sigma \mid \sigma \in \Type \},\allowbreak \PP)$.
Its term and formula translations are defined as follows:
\begin{align*}
\mtagsqqx{t_{1\!} \eq t_2} & \,=\, \CT{}{\mtagsqqx{t_{1\!}}{}} \eq \CT{}{\mtagsqqx{t_2}{}} \\
\mtagsqqx{\exists \mathit{X} \mathbin: \t.\;\, \phi}{} & \,=\,
  \exists \mathit{X} \mathbin: \t.\;
  \begin{cases}
    \mtagsqqx{\phi}{} & \!\!\text{if~$\MONO{\t} \PHI$} \\
    \ti_{\,\t}(\mathit{X}) \eq \mathit{X} \mathrel\land \mtagsqqx{\phi}{} & \!\!\text{otherwise} \\
  \end{cases}
\end{align*}
\\[-.5\baselineskip] with \\[-.5\baselineskip] 
\[\CT{}{t\typ\t} \,=\, \begin{cases}
  t & \!\!\text{if $\MONO{\t} \PHI$ or $t \notin \Vtun$} \\
  \ti_{\,\t}(t) & \!\!\text{otherwise}
\end{cases}\]
The encoding adds the following typing axioms:
\[\!\begin{aligned}[t]
& \textstyle \forall \mathit{\bar X} \mathbin: \bar\t.\; \ti_{\,\t}(f(\mathit{\bar X})) \eq f(\mathit{\bar X})
  &\enskip& \text{for~}\ff : \bar\t \to \t \in \mathcal{F}
     \text{~such that~} \NONMONO{\t} \PHI \\[\betwtypax]
& \exists \mathit{X} \mathbin: \t.\; \ti_{\,\t}(\mathit{X}) \eq \mathit{X}
  && \text{for~} \NONMONO\t \PHI \text{~that is not the
result type of a symbol in}~\mathcal F
\end{aligned}\]
The \emph{monomorphic featherweight type tags} encoding \mono\tags\qquery{}
is the composition $\encode{\phantom{i}}{\mono\tags\qquery\tinycomma\erased}$.
The target signature for \mono\tags\qquery{} is the same as for \mono\tags\query.
\end{defi}


The axioms are necessary for the completeness of \mono\tags\qquery{}.
They would have been harmless for \mono\tags\query{}:\ for soundness,
we can think of the $\ti_{\,\t}$ functions as identities.
The side condition of the last axiom is a minor optimisation:\ it
avoids asserting that $\t$ is inhabited if the symbols in $\mathcal F$
already witness $\t$'s inhabitation.

\begin{exa}%
\label{ex:algebraic-lists-mono-tags-qquery}%
The \mono\tags\qquery{} encoding of Example~\ref{ex:algebraic-lists-monomorphised}
requires fewer tags than \mono\tags\query{},
at the cost of a typing axiom for \const{hd} and typing equations for the
existential variables of type $\www$:
\begin{quotex}
$\forall \mathit{Xs}.\;\, \ti_{\vvthinspace\www}(\const{hd}_{\www}(\mathit{Xs})) \eq \const{hd}_{\www}(\mathit{Xs})$ \\[\betweenaxs]
$\forall \mathit{X}\!, \mathit{Xs}.\;\,
  \const{nil}_{\www} \not\eq \const{cons}_{\vvthinspace\www}(\mathit{X}\!, \mathit{Xs})$ \\
$\forall \mathit{Xs}.\;\,
  \mathit{Xs} \eq \const{nil}_{\www} \mathrel{\lor}
  (\exists \mathit{Y}\!, \mathit{Ys}.\;\,
  \ti_{\vvthinspace\www}(\mathit{Y}) \eq \mathit{Y} \mathrel\land \mathit{Xs} \eq \const{cons}_{\vvthinspace\www}(\mathit{Y}\!, \mathit{Ys})
  )$ \\
$\forall \mathit{X}\!, \mathit{Xs}.\;\,
  \const{hd}_{\www}(\const{cons}_{\vvthinspace\www}(\mathit{X}\!, \mathit{Xs})) \eq \ti_{\vvthinspace\www}(\mathit{X})
  \mathrel\land \const{tl}_{\www}(\const{cons}_{\vvthinspace\www}(\mathit{X}\!, \mathit{Xs})) \eq \mathit{Xs}$ \\
$\exists \mathit{X}\!, \mathit{Y}\!, \mathit{Xs}, \mathit{Ys}.\;\,
  \ti_{\vvthinspace\www}(\mathit{X}) \eq \mathit{X} \mathrel\land \ti_{\vvthinspace\www}(\mathit{Y}) \eq \mathit{Y} \mathrel\land
  \const{cons}_{\vvthinspace\www}(\mathit{X}\!, \mathit{Xs}) \eq \const{cons}_{\vvthinspace\www}(\mathit{Y}\!, \mathit{Ys}) \mathrel\land {}$ \\
$\phantom{\exists \mathit{X}\!, \mathit{Y}\!, \mathit{Xs}, \mathit{Ys}.\;\,}
  (\mathit{X} \not\eq \mathit{Y} \mathrel\lor \mathit{Xs} \not\eq \mathit{Ys})$
\end{quotex}
\end{exa}

\begin{thm}[Correctness of \monobf\tags\query, \monobf\tags\qquery]\afterDot
\label{thm:correctness-of-mono-tags-query}%
The monomorphic type tags encodings\/ \hbox{\rm\mono\tags\query} and\/ \hbox{\rm\mono\tags\qquery}
are correct.
\end{thm}

\begin{proof}
It suffices to show the three conditions of Lemma~\ref{lem:correctness-conditions-monomorphic}.

\betweenitems\noindent
\textsc{Mono}:\enskip
  By induction on $\phi$, it follows that
\begin{enumerate}
\item[(1)] $\mathit{X}^\sigma \in \NV(\encode{\phi}{\mono\tags\query})$ implies $\sigma \in \Inf(\PHI)$;
\item[(2)] $\mathit{X}^\sigma \in \NV(\encode{\phi}{\mono\tags\qquery})$ implies $\sigma \in \Inf(\PHI)$.
\end{enumerate}
(Indeed, while transforming $\phi$ into $\encode{\phi}{\mono\tags\qquery}$, \mono\tags\qquery{} tags precisely the variables that would cause
the condition (1) to fail; and \mono\tags\query{} tags even more.)
Moreover, the typing axioms of \mono\tags\qquery{} have no naked variables, and hence
$\MONO{\sigma} \encode{\phi}{\mono\tags\query}$ and
$\MONO{\sigma} \encode{\phi}{\mono\tags\qquery}$ hold for all types $\sigma$.
Therefore,
by Theorem~\ref{thm:soundness-of-calculus-monomorphic} and Lemma~\ref{lem:monotonicity-preservation-by-union},
$\MONO{\sigma} \encode{\phi}{\mono\tags\query}$ and
$\MONO{\sigma} \encode{\phi}{\mono\tags\qquery}$ are monotonic.

\betweenitems\noindent
\textsc{Sound}:\enskip
  This is immediate for both $\encode{\phantom{i}}{\mono\tags\query}$ and $\encode{\phantom{i}}{\mono\tags\qquery}$:
  given a model of $\PHI$, we extend it to a model
  of the encoded $\PHI$ by interpreting all type tags as the identity.\strut

\betweenitems\noindent
\textsc{Complete for $\encode{\phantom{i}}{\mono\tags\query}$}:\enskip
The proof is analogous to the corresponding case for \tags{} (Theorem~\ref{thm:correctness-of-tags}),
the differences being that here we do not face the complication of interpreting polymorphic types,
and only terms of certain types are tagged.
Let $\NN = \bigl((\D'_\sigma)_{\sigma \in \Type},\allowbreak
(\ff^{\,\NN})_{\ff \in \FF \mathrel{\uplus} \{\ti_\sigma \mid \NONMONO{\sigma} \PHI\}},\allowbreak
(\pp^{\NN})_{\pp \in \PP}\bigr)$
be a model of $\encode{\PHI}{\mono\tags\query}$.
For each $\sigma \in \Type$, we let
$$
\D_\sigma = \left\{
\begin{array}{@{}ll@{}}
 D'_\sigma & \mbox{if $\MONO{\sigma} \PHI$} \\
 \{d \in \D'_\sigma \mid d \mbox{ is in the image of $\ti_\sigma^{\NN}$}\} & \mbox{otherwise}
\end{array}
\right.
$$
We define
$T_\sigma : \D'_\sigma \ra \D'_\sigma$ so as to apply $\ti_\sigma^{\NN}$ only if the element is not already in its image:
$$
T_\sigma(d) = \left\{
\begin{array}{@{}ll@{}}
 d & \mbox{if $d \in \D_\sigma$} \\
 \ti_\sigma^{\NN}(d) & \mbox{otherwise}
\end{array}
\right.
$$
%
%
%
Let $\MM = \bigl((\D_\sigma)_{\sigma \in \Type},(\ff^{\,\MM})_{\ff \in \FF},(\pp^\MM)_{\pp \in \PP}\bigr)$, where
\begin{itemize}
\item
$f^{\MM}(d_1,\ldots,d_n) =
\ti_\sigma^{\NN}(f^{\NN}(T_{\sigma_1}(d_1),\ldots,T_{\sigma_n}(d_n)))$
for $f : \ov{\sigma} \ra \sigma$ is in $\FF$;
\item $p^{\MM}(d_1,\ldots,d_n) =
p^{\NN}(T_{\sigma_1}(d_1),\ldots,T_{\sigma_n}(d_n))$
if $p : \ov{\sigma} \ra \bool$ is in $\FF$.
\end{itemize}
Given
$\xi' : \VV \ra \prod_{\sigma \in \Type} \D'_\sigma$,
we define $\xi : \VV \ra \smash{\prod_{\sigma \in \Type} \D_\sigma}$ as follows:
$$\xi(\mathit{X})(\sigma) =
\left\{
\begin{array}{@{}ll@{}}
 \xi'(\mathit{X})(\sigma) & \mbox{if $\MONO{\sigma} \PHI$} \\
 \ti_\sigma(\xi'(\mathit{X})(\sigma)) & \mbox{otherwise}
\end{array}
\right.$$
The next facts follow by induction on $t$ or $\phi$ (for arbitrary $\xi'$):
\begin{enumerate}
\item[(1)] $\SEM{t}^{\MM}_{\xi} = \SEM{\encode{t}{\tags}}^{\NN}_{\xi'}\!$;
\item[(2)] $\SEM{\phi}^{\MM}_{\xi} = \SEM{\encode{\phi}{\tags}}^{\NN}_{\xi'}\!$.
\end{enumerate}
Let $\phi \in \PHI$.  To show that $\MM$ is a model of $\phi$, let
$\xi : \VV \ra \prod_{\sigma \in \Type} \D_\sigma$.  By the definition of $\D_\sigma$,
there exists $\xi' : \VV \ra \prod_{\sigma \in \Type} \D'_\sigma$
such that the defining property of $\xi$ holds. 
Now, since $\SEM{\encode{\phi}{\mono\tags\query}}^{\NN}_{\xi'}$ holds, it follows from (2) that $\SEM{\phi}^{\MM}_{\xi}$, i.e.\
$\SEM{\phi}^{\MM}$ also holds, as desired.

\betweenitems\noindent
\textsc{Complete for $\encode{\phantom{i}}{\mono\tags\qquery}$}:\enskip
The proof is analogous to the corresponding case for \tags\at{} (Theorem~\ref{thm:correctness-of-tags-at}).
Starting with a model $\NN = \bigl((\D'_\sigma)_{\sigma \in \Type},(\ff^{\,\NN})_{\ff \in \FF \mathrel{\uplus} \{\ti_\sigma \mid \NONMONO{\sigma} \PHI\}},(\pp^{\NN})_{\pp \in \PP}\bigr)$
of $\encode{\PHI}{\mono\tags\qquery}$, we define a model
$\MM$ of $\PHI$ as follows, where the typing axioms ensure that each $D_\sigma$ is nonempty and each $s^{\MM}$ is well defined:
\begin{itemize}
\item $\D_\sigma = \left\{
\begin{array}{@{}ll@{}}
 D'_\sigma & \mbox{if $\MONO{\sigma} \PHI$} \\
 \{d \in \D'_\sigma \mid \ti_\sigma^{\NN}(d) = d\} & \mbox{otherwise}
\end{array}
\right.$
\item $s^{\MM}$ is the restriction of $s^{\NN}\!$.
\end{itemize}
The next facts follow by induction on $t$ or $\phi$ (for arbitrary $\xi : \VV \ra \prod_{\sigma \in \Type} \D_\sigma$):
\begin{enumerate}
\item[(1)] $\SEM{t}^{\MM}_{\xi} = \SEM{\encode{t}{\tags}}^{\NN}_{\xi'}\!$;
\item[(2)] $\SEM{\phi}^{\MM}_{\xi} = \SEM{\encode{\phi}{\tags}}^{\NN}_{\xi'}\!$.
\end{enumerate}
Then, by the usual route, it follows that $\MM$ is a model of $\PHI$.
\end{proof}


\leftOut{
Moreover, we may add further typing axioms to $\PHIii\xx$\dash for example, equations
$\ff(\mathit{\bar U}\negvthinspace, \ti_{\,\t}(\mathit{X}),
\mathit{\bar V}) \eq \ff(\mathit{\bar U}\negvthinspace, \mathit{X}\!, \mathit{\bar V})$ to add or remove
tags around well-typed arguments of a 
symbol \ff, or the idempotence law
$\ti_{\,\t}(\ti_{\,\t}(\mathit{X})) \eq \ti_{\,\t}(\mathit{X})$\dash provided that they hold for
canonical models
and preserve monotonicity.
}


\subsection{Monotonicity-Based Type Guards}
\label{ssec:monotonicity-based-type-guards-monomorphic}

The \mono\guards\query{} and \mono\guards\qquery{} encodings are
defined analogously to \mono\tags\query{} and \mono\tags\qquery{} but
using type guards.
The \mono\guards\query{} encoding omits the guards
for types that are inferred monotonic, whereas \mono\guards\qquery{} omits more
guards that are not needed to make the intermediate problem monotonic.

\newcommand\MGSQXCHEAT{\kern.75ex} 

\begin{defi}[Lightweight Guards \monobf\guards\query]\afterDot
The encoding $\encode{\phantom{i}}{\mono\guards\query}$
translates monomorphic problems over
$\SIGMA = (\Type, \FF, \PP)$ to monomorphic problems over $(\Type, \FF , \PP \mathrel{\uplus} \{\is_\sigma : \sigma \ra \bool \mid \sigma \in \Type \})$.
Its term and formula translations are defined as follows:
\begin{align*}
\mguardsqx{\forall \mathit{X} \mathbin: \t.\;\, \phi} & \,=\,
\forall \mathit{X} \mathbin: \t.\;
\,\relax{\begin{cases}
\mguardsqx{\phi} & \!\!\text{if $\MONO\t \PHI$} \\
\is_{\,\t}(\mathit{X}^\sigma) \rightarrow \mguardsqx{\phi} & \!\!\text{otherwise}
\end{cases}} \\
\mguardsqx{\exists \mathit{X} \mathbin: \t.\;\, \phi} & \,=\,
\exists \mathit{X} \mathbin: \t.\;
\begin{cases}
\mguardsqx{\phi} & \MGSQXCHEAT\!\!\text{if $\MONO\t \PHI$} \\
\is_{\,\t}(\mathit{X}^\sigma) \mathrel\land \mguardsqx{\phi} & \MGSQXCHEAT\!\!\text{otherwise}
\end{cases}
\end{align*}
The encoding adds the following typing axioms:
\[\!\begin{aligned}[t]
& \textstyle \forall \mathit{\bar X} \mathbin: \bar\t.\;\, \is_{\,\t}(f(\mathit{\bar X}^{\ov{\sigma}}))
  &\enskip& \text{for~}\ff : \bar\t \to \t \in \mathcal{F}
     \text{~such that~} \NONMONO\t \PHI \\[\betwtypax]
& \exists \mathit{X} : \t.\;\, \is_{\,\t}(\mathit{X}^\sigma)
  && \text{for~} \NONMONO\t \PHI \text{~that is not the result type of a symbol in}~\mathcal F
\end{aligned}\]
The \emph{monomorphic lightweight type guards} encoding \mono\guards\query{} is
the composition $\encode{\phantom{i}}{\mono\guards\query\tinycomma\erased}$.
It translates a monomorphic problem over $\SIGMA$
into an untyped problem over $\SIGMA' = (\FF', \PP' \mathrel{\uplus} \{\is^1_{\,\t} \mid \sigma \in \Type\})$,
where $\FF', \PP'$ are as for \erased.
\end{defi}


\begin{exa}%
\label{ex:algebraic-lists-mono-guards-query}%
The \mono\guards\query{} encoding of
Example~\ref{ex:algebraic-lists-monomorphised} is as follows:
\begin{quotex}
$\forall \mathit{Xs}.\;\, \is_{\vvthinspace\www}(\const{hd}_{\www}(\mathit{Xs}))$ \\[\betweenaxs]
$\forall \mathit{X}\!,\; \mathit{Xs}.\;\,
  \is_{\vvthinspace\www}(\mathit{X}) \rightarrow
  \const{nil}_{\www} \not\eq \const{cons}_{\vvthinspace\www}(\mathit{X}\!, \mathit{Xs})$ \\
$\forall \mathit{Xs}.\;\,
  \mathit{Xs} \eq \const{nil}_{\www} \mathrel{\lor}
  (\exists \mathit{Y}\!,\; \mathit{Ys}.\;\,
  \is_{\vvthinspace\www}(\mathit{Y}) \mathrel\land
  \mathit{Xs} \eq \const{cons}_{\vvthinspace\www}(\mathit{Y}\!, \mathit{Ys})
  )$ \\
$\forall \mathit{X},\; \mathit{Xs}.\;\,
  \is_{\vvthinspace\www}(\mathit{X}) \rightarrow
  \const{hd}_{\www}(\const{cons}_{\vvthinspace\www}(\mathit{X}\!, \mathit{Xs})) \eq \mathit{X} \mathrel\land
  \const{tl}_{\www}(\const{cons}_{\vvthinspace\www}(\mathit{X}\!, \mathit{Xs})) \eq \mathit{Xs}$ \\
$\exists \mathit{X}\!, \mathit{Y}\!, \mathit{Xs}, \mathit{Ys}.\;\,
  \is_{\vvthinspace\www}(\mathit{X}) \mathrel\land
  \is_{\vvthinspace\www}(\mathit{Y}) \mathrel\land
\const{cons}_{\vvthinspace\www}(\mathit{X}\!, \mathit{Xs}) \eq \const{cons}_{\vvthinspace\www}(\mathit{Y}\!, \mathit{Ys})
\mathrel\land (\mathit{X} \not\eq \mathit{Y} \mathrel{\lor} \mathit{Xs} \not\eq \mathit{Ys})$
\end{quotex}%
Notice that the $\const{tl}_\www$ equation is needlessly in the scope of the guard.
The encoding is more precise if the problem is clausified.
\end{exa}

Our novel encoding \mono\guards\qquery{} omits the guards for variables that do
not occur naked, regardless of whether they are of a monotonic type.

\begin{defi}[Featherweight Guards \monobf\guards\qquery]\afterDot
The \emph{monomorphic featherweight type guards} encoding \mono\guards\qquery{}
is identical to the lightweight encoding \mono\guards\query{} except that the
condition ``if~$\MONO\t \PHI$'' 
in the $\forall$ case is weakened to
``if~$\MONO\t \PHI$ or $\mathit{X} \notin \NV(\phi)$''.
\end{defi}

\begin{exa}%
\label{ex:algebraic-lists-mono-guards-qquery}%
The \mono\guards\qquery{} encoding of the algebraic list problem is identical
to \mono\guards\query{} 
except that the {}$\const{nil}_{\www} \not\eq \const{cons}_{\vvthinspace\www}${}
axiom does not have any guard.
\end{exa}

\begin{thm}[Correctness of \monobf\guards\query, \monobf\guards\qquery]\afterDot
\label{thm:correctness-of-mono-guards-query}%
The monomorphic type guards encodings\/ \hbox{\rm\mono\guards\query}
and\/ \hbox{\rm\mono\guards\qquery} are correct.
\end{thm}
\begin{proof}
It suffices to show the three conditions of Lemma~\ref{lem:correctness-conditions-monomorphic}.

\betweenitems\noindent
\textsc{Mono}:\enskip
  By Lemma \ref{lem:monotonicity-preservation-by-union}, it
 suffices to show that each type $\tau$ such that $\NONMONO{\sigma} \PHI$ is monotonic in the encoded problem.
  By Theorem~\ref{thm:soundness-of-calculus-monomorphic}, types are monotonic unless
  they are possibly finite and variables of their
  types occur naked in the original problem. Both encodings
  guard all such variables\dash \mono\guards\qquery{} guards exactly those
  variables, while \mono\guards\query{} guards more.
  The typing axioms contain no naked variables.
  We cannot use Theorem~\ref{thm:soundness-of-calculus-monomorphic} directly,
  because guarding a naked variable does not make it less naked---but we
  can generalise the proof slightly to exploit the guards.
  Given a model
  $\mathcal M$ of the encoded problem
  $\MM = \bigl((\D_\sigma)_{\sigma \in \Type},(\ff^{\,\MM})_{\ff \in \FF},(\pp^\MM)_{\pp \in \PP \mathrel{\uplus} \{\is_\sigma \mid \NONMONO{\sigma} \PHI\}}\bigr)$,
  we construct a model $\NN = \bigl((\D'_\sigma)_{\sigma \in \Type},\allowbreak(\ff^{\,\NN})_{\ff \in \FF},\allowbreak(\pp^{\NN})_{\pp \in \PP}\bigr)$ with all types infinite
  as in the proof of Theorem~\ref{thm:soundness-of-calculus-monomorphic}, but defining the $\tau$-guard interpretation to be false
  on newly added elements (recall that $\D'_\tau = \D_\tau \mathrel{\uplus} E$ for some countably infinite set $E$):
$$
\is_\sigma^{\NN}(d) =
\left\{
\begin{array}{@{}ll@{}}
 \is_\sigma^{\MM}(d) & \mbox{if $\sigma \not= \tau$ or $\sigma = \tau$ and $d \in \D_\sigma$}
\\
 \mbox{false} & \mbox{if $\sigma = \tau$ and $d \in E$}
\end{array}
\right.
$$
Given $\xi' : \VV \ra \prod_{\sigma \in \Type} \D'_\sigma$, we define
$\xi : \VV \ra \prod_{\sigma \in \Type} \D_\sigma$ by
$$
\xi(\sigma)(d) =
\left\{
\begin{array}{@{}ll@{}}
 \eps{D_\tau} & \mbox{if $\sigma = \tau$ and $d \in E$}
\\
 \xi'(\sigma)(d) & \mbox{otherwise}
\end{array}
\right.
$$
The encoded problem has the form $\PHI_1 \,\cup\, \PHI_2$, where $\PHI_1$ are the added axioms
and $\PHI_2$ are the formula translations from $\PHI$.  It is easy to check that $\NN$ satisfies $\PHI_1$.
We say that a formula is $\tau$-\emph{guarded} if
all its subformulae of the form $\forall \mathit{X} \mathbin: \tau.\; \phi$
have $\phi$ of the form $\lnot\, \is_\tau(\mathit{X}^\tau) \mathrel\lor \chi$
and
all its subformulae of the form $\exists \mathit{X} \mathbin: \tau.\; \phi$
have $\phi$ of the form $\is_\tau(\mathit{X}^\tau) \mathrel\land \chi$.
All the formulae in $\PHI_2$ are clearly guarded.
The next facts follow by induction on $t$ or $\phi$ (for arbitrary $\xi'$):
\begin{enumerate}
\item[(1)] $\SEM{t}^{\NN}_{\xi'} = \SEM{t}^{\MM}_{\xi}$;
\item[(2)] If $\phi$ is guarded and $\is_\tau\,(\xi\,(\tau)\,(\mathit{X}))$ is false whenever $\mathit{X}^\tau \in \NV(\phi) \,\cap\, \FVars(\phi)$,
then $\SEM{\phi}^{\MM}_{\xi}$ implies $\SEM{\phi}^{\NN}_{\xi'}$.%
\end{enumerate}
It follows that $\NN$ satisfies $\PHI_2$ as well. Hence it is the desired model of the encoded problem.

\betweenitems\noindent
\textsc{Sound}:\enskip
  Just like for tags, this is immediate for both $\encode{\phantom{i}}{\mono\guards\query}$ and $\encode{\phantom{i}}{\mono\guards\qquery}$:
  given a model of $\PHI$, we extend it to a model
  of the encoded $\PHI$ by interpreting all type guards as everywhere-true predicates.\strut

\betweenitems\noindent
\textsc{Complete}:\enskip
  The proofs for both $\encode{\phantom{i}}{\mono\guards\query}$ and $\encode{\phantom{i}}{\mono\guards\qquery}$ are very similar to that for
\tags\qquery{} (from Theorem \ref{thm:correctness-of-mono-tags-query})---the only change is the replacement of
the condition $\ti_\sigma^{\NN}(d) = d$ by $\is_\sigma^{\NN}\!(d)$ in the definition of the model $\MM$.  (Just like for \tags\qquery{},
the typing axioms ensure that the structure is well defined.)
\QED
\end{proof}


A simpler but less instructive way to prove \textsc{Mono} is to observe that the
second monotonicity calculus by Claessen et al.\
\cite[\S2.4]{claessen-et-al-2011} can infer monotonicity of all
problems generated by \mono\guards, \mono\guards\query, and \mono\guards\qquery.

\begin{rem} \label{rem-stronger-mono-tags-guards}
The proofs of Theorems \ref{thm:correctness-of-mono-tags-query} and \ref{thm:correctness-of-mono-guards-query}
remain valid as they are even if the generated problems contain \relax{more} tags or guards
than inferred as unnecessary by the monotonicity calculus, i.e.\ if in the definitions of
$\encode{\phantom{i}}{\mono\tags\query}$, $\encode{\phantom{i}}{\mono\tags\qquery}$,
$\encode{\phantom{i}}{\mono\guards\query}$, and $\encode{\phantom{i}}{\mono\guards\qquery}$
we replace the condition ``$\MONO{\sigma} \PHI$'' by ``$\sigma \in J$\vthinspace'', where
$J \subseteq \{\sigma \in \Type_\Sigma \mid \MONO{\sigma} \PHI\}$. (The replacement should be
done everywhere, including in the axioms.)
\end{rem}

\subsection{Heuristic Monomorphisation}
\label{ssec:heuristic-monomorphisation}

\longsect~\ref{sec:complete-monotonicity-based-encoding-of-polymorphism}
will show how to translate polymorphic types soundly and completely. If we are
willing to sacrifice completeness,
an easy way to extend \mono\tags\query{}, \mono\tags\qquery{},
\mono\guards\query{}, and \mono\guards\qquery{} to polymorphism is to perform
\emph{\vthinspace heuristic monomorphisation} on the polymorphic problem:\
\begin{enumerate}[label=\arabic*.]
\item[1.] Heuristically instantiate all type variables with suitable ground
types, taking finitely many copies of each formula if desired.

\betweenitems

\item[2.] Map each ground occurrence $s\LANY\bar\a\vthinspace\rho\RAN$ of a polymorphic symbol
$s : \forall\bar\a.\; \bar\t \to \t$ to a fresh monomorphic symbol
$s_{\vvthinspace\bar\a\vthinspace\rho} : \lfloor\bar\t\vthinspace\rho\rfloor \to \lfloor\t\vthinspace\rho\rfloor$,
where $\rho$ is a ground type substitution
(a~function from type variables to ground types)
and $\lfloor\phantom{i}\rfloor$ is an injection from ground types to nullary type
constructors (e.g.\ $\{\www \mapsto \www,\: \mathit{list}(\www) \mapsto \mathit{list\IUS}\www\}$).
\end{enumerate}
Heuristic monomorphisation is generally incomplete \cite[\S2]{bobot-paskevich-2011}
and often overlooked in the literature, but by eliminating
type variables it considerably simplifies the generated formulae, leading to
very efficient encodings.
It also provides a simple and effective way to exploit the native
support for monomorphic types in some automatic provers. 

%
%


\section{Complete Monotonicity-Based Encodings of Polymorphism}
\label{sec:complete-monotonicity-based-encoding-of-polymorphism}

Heuristic monomorphisation is simple and effective, but its incompleteness
can be a cause for worry, and its nonmodular nature makes it unsuitable for some applications
that need to export an entire polymorphic theory independently of any
conjecture. Here we adapt the
type encodings to a polymorphic setting.
%

We start by defining a correct but infinitary translation
of a polymorphic into a monomorphic problem, called complete monomorphisation. 
Then we address the genuinely polymorphic
issue of encoding type arguments, proving conditions under which
composition of an encoding $\xx$ with the type arguments encoding \args\  is correct:\ $\xx$ must be complete
and produce monotonic problems. 
Finally, we define polymorphic counterparts of the guard and tag encodings and show that
they satisfy the required conditions by reducing (most of) the problem to the monomorphic case
via complete monomorphisation.


%

\subsection{Complete Monomorphisation}
\label{subsec:complete-mono}

The main insight behind complete monomorphisation is that a
polymorphic formula having all its type quantification at the top
is equisatisfiable to the (generally infinite) set of its
monomorphic instances.
Complete monomorphisation does not obey our convention about encodings,
since it translates each polymorphic
formula not into a single monomorphic formula but into a set of formulae.

\begin{defi}[Instance Type and Most General Instance]\afterDot
\label{def-mgi}
A type $\tau$ is an \emph{instance} of a type $\sigma$ if there exists a type substitution
$\rho$ such that $\tau = \sigma\vthinspace\rho$. If this is the case, we also say that $\tau$ is
{\em less general} than $\sigma$ and that $\sigma$ is {\em more general} than $\tau$, and we write $\tau \leq \sigma$.
Given two types $\sigma$ and $\tau$,
if they have a common instance (i.e.\ a unifier), then they also have a most general common instance, which
we denote by $\mgi(\sigma,\tau)$. 
\end{defi}

\begin{defi}[Complete Monomorphisation \cmono]\afterDot
We define the encoding $\encode{\phantom{i}}{\cmono}$ that translates polymorphic problems over
$\SIGMA = (\KK, \FF, \PP)$ to monomorphic problems over $\Sigma' = (\Type',\FF',\PP')$, where
\begin{itemize}
\item $\Type' = \GType_\Sigma$;
\item for each $s : \forall\ov{\alpha}.\;\ov{\sigma} \ra \varsigma$ in $\FF \mathrel{\uplus} \PP$
and each $\rho$ such that each $\alpha_i$ is mapped to a ground type $\tau_i \in \GType_{\Sigma}$,
$\FF' \mathrel{\uplus} \PP'$ contains a symbol $s_{\,\ov{\tau}} : \ov{\sigma}\vthinspace\rho \ra \sigma\vthinspace\rho$.
\end{itemize}
We first define the translation of terms and formulae that contain no type variables, i.e.\
that have no type quantifiers and
such that all the occurring types in applications of the function or predicate symbols are ground:
\begin{align*}
\cmonox{f\LAN\bar\t\RAN(\bar t\,)} & \,=\, f_{\ov{\sigma}}(\cmonox{\bar t\,}) &
\cmonox{\mathit{X}^\sigma} &\,=\, \mathit{X}^\sigma
\\[\betweentf]
\cmonox{p\LAN\bar\t\RAN(\bar t\,)} & \,=\, p_{\ov{\sigma}}(\cmonox{\bar t\,}) &
  \cmonox{\forall \mathit{X} \mathbin: \t.\;\, \phi} & \,=\, \forall \mathit{X} \mathbin: \t.\; \cmonox{\phi}
\\
\cmonox{\lnot\,p\LAN\bar\t\RAN(\bar t\,)} & \,=\, \lnot\,p_{\ov{\sigma}}(\cmonox{\bar t\,})
& \cmonox{\exists \mathit{X} \mathbin: \t.\;\, \phi} & \,=\, \exists \mathit{X} \mathbin: \t.\; \cmonox{\phi}
\end{align*}
Now, given a sentence $\forall \ov{\alpha}.\; \phi$ where $\ov{\alpha}$ indicates all
its universally quantified types, we encode it as the set of encodings of its monomorphic instances (via ground type substitutions~$\rho$):
\begin{align*}
\cmonox{\forall \ov{\alpha}.\; \phi} & \,=\, \{\cmonox{\phi\vthinspace\rho} \mid \mbox{$\rho : \AAA \ra \GType$}\}
\end{align*}
Finally, the encoding of a problem is the union of the encoding of its formulae:
\begin{align*}
\textstyle\cmonox{\PHI} & \,=\, \textstyle\bigcup_{\phi \in \PHI} \cmonox{\phi}
\end{align*}
\end{defi}

\begin{conv} \label{con-alpha-phi}
Whenever we write a polymorphic formula as $\forall \ov{\alpha}.\; \phi$, we implicitly assume
that $\ov{\alpha}$ indicates all its universally quantified types, so that $\phi$ has no type quantifiers.
\end{conv}

\begin{lem}[Correctness of \cmono]\afterDot
\label{lem:correctness-of-cmono}%
The complete monomorphisation encoding \cmono{} is correct.
\end{lem}
\begin{proof}

First, observe that in any model $\MM$ of a polymorphic signature, the interpretation
$\SEM{\phantom{i}}_{\theta,\xi}^\MM$ of a term or formula that does not contain type variables
does not depend on $\theta$ and, from $\xi$, it
only depends on the restriction of $\xi$ to ground types,
$\xi' : \VV \ra \prod_{\sigma \in \GType_{\Sigma}}\,\SEM{\sigma}^\MM$. We can therefore write
$\SEM{\phantom{i}}_{\xi'}^\MM$ instead of $\SEM{\phantom{i}}_{\theta,\xi'}^\MM$.

\betweenitems\noindent
\textsc{Sound}:\enskip
Assume $\PHI$ is satisfiable and let $\Sigma = (\KK,\FF,\PP)$ be its polymorphic signature.
By Lemma~\ref{lem-low-sko-aux-poly}, $\PHI$ also has a model
$\MM = \bigl(\dom,(\ff^{\,\MM})_{\ff \in \FF},(\pp^\MM)_{\pp \in \PP}\bigr)$ for which
$\SEM{\phantom{i}}^\MM : \GType_\Sigma \ra \dom$ is a bijection---let $v : \dom \ra \GType_\Sigma$
be its inverse.  We define a structure
$\NN = \bigl((\D'_\tau)_{\tau \in \Type'},(\ff^{\,\NN})_{\ff \in \FF},(\pp^{\NN})_{\pp \in \PP'}\bigr)$ for $\Sigma'$ as follows:
\begin{itemize}
\item $\D'_\tau = \SEM{\tau}^\MM$;
\item if $s : \forall\ov{\alpha}.\;\ov{\sigma} \ra \varsigma$ is in $\FF \mathrel{\uplus} \PP$
and $\ov{\tau} = (\tau_1,\ldots,\tau_m)$ are ground instances of $\ov{\alpha} = (\alpha_1,\ldots,\alpha_m)$ via $\rho$,
we define
by $s_{\ov{\tau}}^{\NN}(\ov{d}) = s^{\MM}(v(\tau_1),\ldots,v(\tau_m))(\ov{d})$.
\end{itemize}
The next fact follows by induction on the term or formula $\delta$ (for arbitrary $\xi : \VV \ra \prod_{\tau \in \Type_{\Sigma'}}\,\D'_\tau$):
\begin{enumerate}
\item[(1)] If $\delta$ contains no type variables, then $\SEM{\cmonox{\delta}}^{\NN}_{\xi'} = \SEM{\delta}^{\MM}_\xi$.
\end{enumerate}
%
%
In particular, for a sentence $\phi$, we have $\SEM{\cmonox{\phi}}^{\NN} = \SEM{\phi}^{\MM}$.

Now let $\forall \ov{\alpha}.\;\phi \in \PHI$.  Since $\MM$ is a model of $\forall \ov{\alpha}.\;\phi$, $\MM$ is also a model
of all its ground-type instances $\phi\vthinspace\rho$, and hence, since $\SEM{\cmonox{\phi}}^{\NN} = \SEM{\phi}^{\MM}$,
$\NN$ is a model of each formula in $\cmonox{\forall \ov{\alpha}.\;\phi}$.
Thus, $\NN$ is a model of $\cmonox{\PHI}$, as desired.

\betweenitems\noindent
\textsc{Complete}:\enskip
Let $\NN = \bigl((\D'_\tau)_{\tau \in \Type'},(\ff^{\,\NN})_{\ff \in \FF},(\pp^{\NN})_{\pp \in \PP'}\bigr)$ be a model of $\cmonox{\PHI}$,
for which we can assume without loss of generality that $\D'_\tau \,\cap \D'_{\tau'} = \emptyset$ if $\tau \not= \tau'$.
We define a structure $\MM = \bigl(\dom,(\ff^{\,\MM})_{\ff \in \FF},(\pp^\MM)_{\pp \in \PP}\bigr)$ for $\Sigma$ as follows:
\begin{itemize}
\item $\dom = \{\D'_\tau \mid \tau \in \Type'\}$---for each
$D \in \dom$, we let $v(D)$ be the unique $\tau \in \Type' = \GType_\Sigma$ such that $D = \D_\sigma$;
\item if $\ov{\alpha} = (\alpha_1,\ldots,\alpha_m)$, $s : \forall\ov{\alpha}.\;\ov{\sigma} \ra \varsigma \in \FF \mathrel{\uplus} \PP$, and
$\ov{D} \in \dom^m$,
we define $s^{\MM}(\ov{D}) = s_{(v(D_1),\ldots,v(D_m))}^{\NN}$.
\end{itemize}
The next fact follows by induction on $\tau$:
\begin{enumerate}
\item[(1)] $\tau \in \GType$ implies $\SEM{\tau}^{\MM} = \D'_\tau$.
\end{enumerate}
The next fact follows by induction on the term or formula $\delta$ (for arbitrary $\xi : \VV \ra \prod_{\tau \in \GType_{\Sigma}}\,\SEM{\tau}^{\MM}$):
\begin{enumerate}
\item[(2)] If $\delta$ contains no type variables, then $\SEM{\delta}^{\MM}_\xi = \SEM{\cmonox{\delta}}^{\NN}_{\xi'}$.
\end{enumerate}
%
%
In particular, for a sentence $\phi$, we have $\SEM{\phi}^{\MM} = \SEM{\cmonox{\phi}}^{\NN}$.
Now let $\forall \ov{\alpha}.\;\phi \in \PHI$, and assume by absurdity that $\MM$ is not a model of $\forall \ov{\alpha}.\;\phi$.
Then, since by (1) and the definition of $\dom$ we have that $\SEM{\phantom{i}}^\MM : \GType \ra \dom$ is surjective,
we obtain a ground-type instance $\phi\vthinspace\rho$ of $\phi$ such that $\MM$ is not a model of $\phi\vthinspace\rho$.
By $\SEM{\phi}^{\MM} = \SEM{\cmonox{\phi}}^{\NN}\!$, $\NN$ is not a model of $\SEM{\cmonox{\phi\vthinspace\rho}}^{\NN}$, hence not a model of $\cmonox{\forall \ov{\alpha}.\;\phi}$,
which is a contradiction.
Therefore $\MM$ is a model of $\forall \ov{\alpha}.\;\phi$.
We obtain that $\MM$ is a model $\PHI$, as desired.
\QED
\end{proof}

\subsection{Monotonicity}
\label{ssec:monotonicity-polymorphic}
The definition of monotonicity from \longsect~\ref{ssec:monotonicity-monomorphic}
(Definition~\ref{def:infinite-monotonicity}) must be adapted to the polymorphic
case.

\leftOut{
Definition \ref{def:infinite-monotonicity} and Lemmas
\ref{lem:downward-loewenheim-skolem},
\ref{lem:monotonicity-preservation-by-union}, and
\ref{lem:submodel}
cater for polymorphic problems. Other results from
\longsect~\ref{sec:monotonicity-based-type-encodings-the-monomorphic-case} must
be adapted to the polymorphic case.

\begin{defi}[Monotonicity of Polymorphic Type]\afterDot
A polymorphic type is (\emph{infinitely}) \emph{monotonic} if all its
ground instances are monotonic.
\end{defi}
}

\begin{defi}[Monotonicity]\afterDot
\label{def:infinite-monotonicity-polymorphic}%
Let $S$ be a set of types and $\PHI$ be a polymorphic
problem.
The set $S$ is 
\emph{monotonic} in
$\PHI$ if for all models $\MM$ of $\PHI$,
there exists a model $\NN$ of $\PHI$ such that
for all ground types $\sigma$,
$\SEM{\sigma}^{\NN}$
is infinite if $\sigma$ is an instance of a type in $S$ and
$\left|\smash{\SEM{\sigma}^{\NN}}\right| = \left|\smash{\SEM{\sigma}^\MM}\right|$ otherwise.
A type $\t$ is 
\emph{monotonic} if $\{\t\}$ is monotonic.
The problem~$\PHI$ is 
\emph{monotonic} if the set $\Type$ of all types is monotonic.
\end{defi}

\begin{exa}
For the algebraic list problem of Example~\ref{ex:algebraic-lists},
$S = \{\mathit{list}(\a)\}$ is monotonic because all models $\MM$ of the problem
necessarily interpret all of the ground instances of $\mathit{list}(\a)$
(e.g.\ $\mathit{list}(\www)$, $\mathit{list}(\mathit{list}(\www))$) by
infinite domains. Monotonicity is trivially witnessed by taking $\NN = \MM$.
\end{exa}

\begin{thm}[Monotonic Erasure]\afterDot
\label{thm:monotonic-erasure-polymorphic}%
The traditional type arguments encoding \args\
is sound for
monotonic polymorphic problems.
\end{thm}

\begin{proof}
Let\/ \PHI{} be such a problem and $\Sigma = (\KK,\FF,\PP)$ be its signature, and assume
\PHI{} is satisfiable. By monotonicity and Lemma~\ref{lem-low-sko-poly},
it also has a model $\MM$ where all $D \in \dom$
are countably infinite.
From this model, we construct a
model $\mathcal M'$ of $\encode{\PHI{}}{\args
\tinycomma\erased}$ with a countably infinite domain $E$
that interprets the encoded types as distinct elements of $E$. The function and
predicate tables for $\mathcal M'$ are based on those
from~$\mathcal M$, with encoded type arguments corresponding to actual type
arguments.

More precisely, let $E$ be an countably infinite set and consider the
following functions:
\begin{itemize}
\item the mutually inverse bijections $u : \dom \ra E$ and $v : E \ra \dom$;
\item for each $D \in \dom$, the mutually inverse bijections $u_D : D \ra E$ and $v_D : E \ra D$.
\end{itemize}
Let $\Sigma' = (\FF' \mathrel{\uplus} \KK',\PP')$ be the target untyped signature of \args{}.
We define the structure $\NN$ for $\Sigma'$ as follows:
\begin{itemize}
\item $\D = E$.
\item Assume $k$ is an $n$-ary function in $\KK'$, meaning $k :: n$ is in $\KK$.
Then $k^{\NN}(\ov{e}) = u(k^\MM(v(\ov{e})))$.
\item Assume $\ov{\alpha} = \alpha_1, \ldots, \alpha_m$, $\ov{\sigma} = \sigma_1, \ldots, \sigma_n$,
and $s : \forall \ov{\alpha}.\;\ov{\sigma} \ra \varsigma \in \FF \mathrel\uplus \PP$, meaning that
$s$ is an $(m+n)$-ary symbol in $\FF' \mathrel{\uplus} \PP'$.
Given $\ov{c} = (c_1,\ldots,c_m) \in E^m$ and $\ov{e} = (e_1,\ldots,e_n) \in E^n$, we define
$s^{\NN}(\ov{c},\ov{e}) =
\smash{u_{\SEM{\sigma}^\MM_{\theta}}(s^\MM(v(\ov{c}))(v_{\SEM{\sigma_1}^\MM_{\theta}}(e_1),\ldots,v_{\SEM{\sigma_n}^\MM_{\theta}}(e_n)))}$,
where $\theta$ maps each $\alpha_i$ to $v(c_i)$.
\end{itemize}
Given $\xi' : \VV \ra E$, we define $\theta : \AAA \ra \dom$
by $\theta(\alpha) = v(\xi'(\VV(\alpha)))$
and
$\xi : {\mathcal V} \ra \smash{\prod_{\sigma \in \Type_\Sigma} \SEM{\sigma}^\MM_\theta}$
by $\xi(\mathit{X})(\sigma) = v_{\SEM{\typex{\sigma}}_{\xi'}^{\NN}}(\xi'(\mathit{X}))$, where
$\mathcal{V}(\a)$ and $\typex{\sigma}$ are as in Definition~\ref{def-term-enc-types}.

The next facts follow by structural induction on $\sigma$, $t$, and $\phi$ (for arbitrary $\xi'$):
\begin{itemize}
\item $\SEM{\typex{\sigma}}^{\NN}_{\xi'} = u(\SEM{\sigma}^\MM_{\theta})$;
\item $t : \sigma$ implies $\SEM{\encode{t}{\args
\tinycomma\erased}}^{\NN}_{\xi'} = u_{\SEM{\sigma}^\MM_{\theta}}(\SEM{t}^\MM_{\theta,\xi})$;
\item $\SEM{\encode{\phi}{\args
\tinycomma\erased}}^{\NN}_{\xi'} = \SEM{\phi}^\MM_{\theta,\xi}$.
\end{itemize}
In particular, for sentences, we have $\SEM{\encode{\phi}{\args
\tinycomma\erased}}^{\NN} = \SEM{\phi}^\MM$;
and since $\MM$ is a model of~$\PHI$, it follows that $\NN$ is a model of $\encode{\PHI}{\args
\tinycomma\erased}$.
\QED
\end{proof}

\begin{lem}[Correctness Conditions]\afterDot
\label{lem:correctness-conditions-polymorphic}%
Let\/ $\PHI$ be a polymorphic problem, and let\/ $\xx$ be a polymorphic encoding. The
problems\/ $\PHI$\/ and\/ $\encode{\PHI}{\xx\tinycomma\args
\tinycomma\erased}$ are
equisatisfiable provided that the following conditions hold\/{\rm:}\strut
\begin{itemize}
\item[] \textsc{Mono}{\rm:}\enskip $\PHIii\xx$\/ is monotonic.

\betweenitems

\item[] \textsc{Sound}{\rm:}\enskip If\/ $\PHI$ is satisfiable, so is\/ $\PHIii\xx$.

\betweenitems

\item[] \textsc{Complete}{\rm:}\enskip If\/ $\PHIii\xx$ is satisfiable, so is\/ $\PHI$.
\end{itemize}
\end{lem}
\begin{proof}
Immediate from Theorems \ref{thm:completeness-of-argsx} and
\ref{thm:monotonic-erasure-polymorphic}.\QED
\end{proof}

\begin{lem}[Monotonicity Preservation and Reflection by \cmono]\afterDot
\label{lem:pres-of-cmono}%
A polymorphic problem $\PHI$ is monotonic iff its monomorphic encoding $\cmonox{\PHI}$ is monotonic.
\end{lem}
\begin{proof}
Let $\phi_k$ be the polymorphic formula stating that all type domains have at least $k$ elements,
and let $\PHI' = \{\phi_k \mid k \in \mathbb{N}\}$.
We note that the following four statements are equivalent:
\begin{enumerate}
\item[(1)] $\PHI$ has a model with all domains infinite;
\item[(2)] $\PHI \,\cup\, \PHI'$ has a model;
\item[(3)] $\cmonox{\PHI} \,\cup\, \cmonox{\PHI'}$ has a model;
\item[(4)] $\cmonox{\PHI}$ has a model with all domains infinite.
\end{enumerate}
The equivalences ``(1) iff (2)'' and ``(3) iff (4)'' are obvious, and ``(2) iff
(3)'' follows from the correctness of \cmono{} (Lemma
\ref{lem:correctness-of-cmono}).
%
\QED
\end{proof}

\subsection{Monotonicity Inference}
\label{ssec:monotonicity-inference-polymorphic}

The monotonicity inference of
\longsect~\ref{ssec:monotonicity-inference-monomorphic} must be adapted to the
polymorphic setting.
We start by generalising the notion of naked variable.
\begin{defi}[Polymorphic Naked Variable]\afterDot
\label{def:naked-variable-poly}%
The set of \emph{naked variables} $\NV(\phi)$
of a polymorphic formula $\phi$ is defined similarly to the monomorphic case
(Definition\ \ref{def:naked-variable}). The following equations are new or
slightly different from there:
\begin{align*}
\NV(p\LAN\bar\t\RAN(\bar t\,)) & \,=\, \emptyset &
  \NV(\lnot\,p\LAN\bar\t\RAN(\bar t\,)) & \,=\, \emptyset &
\NV(\forall\a.\;\phi) & \,=\, \NV(\phi)
\end{align*}
%
Again, we take $\NV(\PHI) = \bigcup_{\phi \in \PHI} \NV(\phi)$.
\end{defi}

\leftOut{
Note that the undercover variables from Section~\ref{sec:alternative-cover-based-encoding-of-polymorphism} are reminiscent of naked variables. This is no coincidence:\
naked variables effectively carry the implicit type argument of
${\eq} : \forall\a.\; \a \times \a \to \bool$, and
undercover variables generalise this to arbitrary predicate and function
symbols.
Equality is special in two respects, though:\ its only sound cover is
effectively $\{1, 2\}$ because of the equality axioms (which are built
into the provers, without any protectors), and the negative case is optimised to
take advantage of disequality's fixed semantics.
}


Note from Definitions \ref{def:naked-variable} and \ref{def:naked-variable-poly}
that the naked variables of a polymorphic formula occur at the same positions as
in all its monomorphic instances. Moreover, the types of the naked variables in
the completely monomorphised problem are simply the ground instances of the types
of the naked variables in the original problem:

\begin{lem}\label{lem-NV-transport}
$\NV(\cmonox{\PHI}) = \{\mathit{X}^\tau \,{\mid}\; \tau \in \negvthinspace\GType\negvthinspace \mbox{ and there exists $\sigma \mathbin\geq \tau$ such that $\mathit{X}^\sigma\negvthinspace \mathbin\in \NV(\PHI)$}\}$.
\end{lem}

The calculus presented below
captures the insight that a polymorphic type is
monotonic if each of its common instances with the type of any
naked variable is an instance of an infinite type. Similarly to its monomorphic counterpart,
it is parameterised by a fixed set $\Inf(\PHI)$ of (not necessarily ground) types for which any interpretation $\SEM{\sigma}^\MM_\theta$
(for any type valuation~$\theta$) in any model $\MM$
of $\PHI$ is known to be infinite. 

\begin{conv}\label{con-Inf-transport}
It is easy to see that if $\sigma \in \Inf(\PHI)$ and $\rho : \AAA \ra \GType_\Sigma$, then
$\sigma\vthinspace\rho$ is infinite in all models of $\cmonox{\PHI}$. We fix the ``known as infinite'' types of $\cmonox{\PHI}$, $\Inf(\cmonox{\PHI})$,
to be $\{\sigma\vthinspace\rho \mid \sigma \in \Inf(\PHI) \mbox{ and } \rho : \AAA \ra \GType_\Sigma\}$.
%
\end{conv}

\begin{defi}[Polymorphic Monotonicity Calculus $\MONO{}$]\afterDot
\label{def:monotonicity-calculus-polymorphic}%
A judgement $\MONO{\t} \PHI$ indicates
that the type $\t$ is inferred monotonic in~$\PHI$.
The \emph{monotonicity calculus} consists of the single rule\strut

\vskip\abovedisplayskip

\noindent
\hfill
\AXC{for all $\mathit{X}^{\sigma'} \mathbin\in \NV(\PHI)$, if $\sigma$ and $\sigma'$ have a common instance, then $\textsf{mgi}(\sigma, \sigma') \in \InfX(\PHI)$\strut}
\UIC{$\MONO{\t} \PHI$\strut}
\DP
\hfill
\hbox{}

\vskip\belowdisplayskip

\noindent
where
$\InfX(\PHI) = \{\sigma \in \Type_\Sigma \mid \mbox{there exists }\sigma' \in \Inf(\PHI) \mbox{ such that } \sigma \leq \sigma'\}$
consists of all instances of all types in $\Inf(\PHI)$.
\end{defi}

\begin{exa}
For the algebraic list problem of Example~\ref{ex:algebraic-lists},
the only naked variables are $\mathit{X}\typ{\a}$ and
$\mathit{Xs}\typ{\mathit{list(\a)}}$, i.e.\ $\,\NV(\PHI) = \{\mathit{X}\typ{\a}, \mathit{Xs}\typ{\mathit{list(\a)}}\}$.
Assume $\Inf(\PHI) = \{\mathit{list}(\a)\}$.
Then $\InfX(\PHI) = \{\mathit{list}(\sigma) \mid \sigma \in \TypesOf_\Sigma\}$.
The type $\mathit{list}(\a)$ is inferred monotonic by $\MONO{}$ because its most general
common instance with $\a$ (the type of $\mathit{X}$) and $\mathit{list}(\a)$
(the type of $\mathit{Xs}$) is in both cases $\mathit{list}(\a)$, which is
known to be infinite. In contrast, $\a$ and $\www$ cannot be
inferred monotonic:\ each of them is its most general common instance with $\a$,
and neither of them is among the types that are known to be infinite.
\end{exa}

\leftOut{
\begin{rem}
Although type variables
occurring in declarations and formulae are bound by a $\forall$ quantifier,
for monotonicity inference polymorphic types
$\t[\bar\a]$, where $\bar\a$
are the type variables occurring in $\t$,
are viewed as being implicitly bound (i.e.\
``$\forall\bar\a.\;\t[\bar\a]$'' to abuse notation). Type variables are assumed
to be fresh in distinct types, comparison between types is modulo $\a$-renaming,
and substitution is capture-avoiding. Thus, we have $\mathit{list}(\a) = \mathit{list}(\b)$ (since
``$\forall\a.\> \mathit{list}(\a) = \forall\b.\> \mathit{list}(\b)$''), and
$\mathit{map}(\a, d)$ can be unified
with $\mathit{map}(c, \a)$ by renaming the second $\a$ to $\b$ and applying
$\rho = \{\a \mapsto c,\> \b \mapsto d\}$.
\end{rem}

Our proof strategy is to reduce the polymorphic case to the already proved
monomorphic case. Our Herbrandian motto is,
\begin{quote}
A polymorphic formula is equisatisfiable to the (generally infinite) set of its
mono\-morphic instances.
\end{quote}
This complete form of monomorphisation is not to be confused with the finitary,
heuristic monomorphisation algorithm presented in
\longsect~\ref{ssec:heuristic-monomorphisation}.
}


\begin{lem}\label{lem-calc-transport}
The following properties hold for any polymorphic signature $\Sigma$
and any $\Sigma$-problem~$\PHI${\rm:}
\begin{enumerate}
\item[(1)] If $\MONO{\sigma} \PHI$ and $\tau \in \GType_\Sigma$ such that $\tau \leq \sigma$, then $\MONO{\tau} \cmonox{\PHI}$.
\item[(2)] If $\tau \in \GType_\Sigma$, then $\MONO{\tau} \cmonox{\PHI}$ iff $\MONO{\tau} \PHI$.
\item[(3)] If $\MONO{\sigma} \PHI$ and $\sigma' \leq \sigma$, then $\MONO{\sigma'} \PHI$.
\end{enumerate}
\end{lem}
\begin{proof}
(1):\enskip
Assume $\MONO{\sigma} \PHI$ and $\sigma \geq \tau \in \GType_\Sigma$.
To show $\MONO{\tau} \cmonox{\PHI}$, it suffices to
let $\mathit{X}^\tau \in \NV(\cmonox{\PHI})$ and prove $\tau \in \Inf(\cmonox{\PHI})$.
By Lemma~\ref{con-Inf-transport}, we obtain $\sigma' \geq \tau$ such that $\mathit{X}^{\sigma'} \in \NV(\PHI)$.
Since $\sigma$ and $\sigma'$ have a common instance, by the first assumption we have $\mgi(\sigma,\sigma') \in \InfX(\PHI)$,
and hence, since $\tau \leq \mgi(\sigma,\sigma')$, we have $\tau \in \InfX(\PHI)$---and since $\tau$ is
ground, we obtain $\tau \in \Inf(\cmonox{\PHI})$, as desired.

\betweenitems\noindent
(2):\enskip
Let (A) $\tau \in \GType_\Sigma$.  One implication follows immediately from (1).  For the second,
assume (B) $\MONO{\tau} \cmonox{\PHI}$.  To prove $\MONO{\tau} \PHI$, we fix
(C) $\mathit{X}^\sigma \in \NV(\PHI)$ such that
$\sigma$ and $\tau$ have a common instance (i.e.\ by (A), $\tau \leq \sigma$) and prove $\mgi(\sigma,\tau) \in \InfX(\PHI)$,
i.e.\ by (A), $\tau \in \InfX(\PHI)$, i.e.\ again by (A), $\tau \in \Inf(\cmonox{\PHI})$.
From (C), $\tau \leq \sigma$ and Lemma~\ref{con-Inf-transport}, we obtain $\mathit{X}^\tau \in \NV(\cmonox{\PHI})$; hence,
by (A), we have $\tau \in \Inf(\cmonox{\PHI})$, as desired.

\betweenitems\noindent
(3):\enskip
Immediate from the definition.
\QED
\end{proof}

Note that property (1) of the above lemma states that if a type is inferred monotonic in the polymorphic calculus,
all its ground instances can be inferred monotonic in the monomorphic calculus associated
to the completely monomorphised problem.  This allows us to prove soundness of the former from soundness of the latter:

\begin{thm}[Soundness of $\MONO{}$]\afterDot
\label{thm:soundness-of-calculus-polymorphic}%
Let\/ $\PHI$ be a polymorphic problem. If\/ $\MONO{\t} \PHI$ for all $\sigma \in \Type_\Sigma$,
then $\PHI$ is monotonic.
\end{thm}

\begin{proof}
Assume $\MONO{\sigma} \PHI$ for all $\sigma \in \Type_\Sigma$.
By Lemma~\ref{lem-calc-transport}, $\MONO{\tau} \cmonox{\PHI}$ for all ground instances $\tau$ of types in $\Type_\Sigma$,
i.e.\ for all $\tau$ in the signature of $\cmonox{\PHI}$. By Theorem \ref{thm:soundness-of-calculus-monomorphic}
and Lemma~\ref{lem:monotonicity-preservation-by-union}, this makes $\cmonox{\PHI}$ monotonic.
Then, by Lemma~\ref{lem:correctness-of-cmono}, $\PHI$ is also monotonic.
\QED
\end{proof}

\subsection{General Strategy}
\label{sec-genStrat}

We will define polymorphic versions of the featherweight and lightweight tags and guard encodings, altogether four encodings:\
\tags\query{}, \tags\qquery{}, \guards\query{}, and \guards\qquery{}.  If $\xx$ ranges over the polymorphic encodings,
$\xmono \in \{\mono\tags\query{}, \mono\tags\qquery{}, \mono\guards\query{}, \mono\guards\qquery{}\}$ denotes
its monomorphic counterpart.
We base the correctness of each $\xx$ on the correctness of $\xmono$, by proving the problems
$\encode{\PHI}{\cmono\tinycomma\xmono}$ and $\encode{\PHI}{\xx\tinycomma\cmono}$
equisatisfiable in a monotonicity-preserving way.
The following lemma implicitly assumes that the signatures
of $\encode{\PHI}{\cmono\tinycomma\xmono}$ and $\encode{\PHI}{\xx\tinycomma\cmono}$ coincide,
which is easy to check for each encoding $\xx$.

\begin{lem}[Correctness Conditions via Complete Monomorphisation]\afterDot
\label{lem-transport}%
Let\/ $\PHI$ be a polymorphic problem and let\/ $\xx \in \{\tags\query{}, \tags\qquery{}, \guards\query{}, \guards\qquery{}\}$.
The
problems\/ $\PHI$\/ and\/ $\encode{\PHI}{\xx\tinycomma\args
\tinycomma\erased}$ are
equisatisfiable provided that the following conditions hold\/{\rm:}
\begin{itemize}
\item[] \textsc{Sound}{\rm:}\enskip If\/ $\encode{\PHI}{\cmono\tinycomma\xmono}$ has a model, $\encode{\PHI}{\xx\tinycomma\cmono}$
has a model with the same interpretation of the type constructors.

\betweenitems

\item[] \textsc{Complete}{\rm:}\enskip If\/ $\encode{\PHI}{\xx\tinycomma\cmono}$ has a model, $\encode{\PHI}{\cmono\tinycomma\xmono}$
has a model with the same interpretation of the type constructors.
\end{itemize}
\end{lem}
\begin{proof}
By the correctness of $\cmono$ (Lemma \ref{lem:correctness-of-cmono}) and
the corresponding correctness theorems for $\xmono$
(more precisely, from the \textsc{Sound} and \textsc{Complete} statements
from the proofs of Theorems \ref{thm:correctness-of-mono-tags-query} and \ref{thm:correctness-of-mono-guards-query}),
we have that $\PHI$ and $\encode{\PHI}{\cmono\tinycomma\xmono}$ are equisatisfiable.
Together with the assumptions, this implies that (A) $\PHI$ and
$\encode{\PHI}{\xx\tinycomma\cmono}$ are equisatisfiable.
Moreover, from Lemma \ref{lem:pres-of-cmono} and monotonicity of $\xmono$
(more precisely, from the \textsc{Mono} statements from the proofs of
Theorems \ref{thm:correctness-of-mono-tags-query} and \ref{thm:correctness-of-mono-guards-query}),
we know that $\encode{\PHI}{\cmono\tinycomma\xmono}$ is monotonic.
Hence, by the assumptions, (B) $\encode{\PHI}{\xx\tinycomma\cmono}$ is also monotonic.
The desired fact follows from
(A), (B), and Lemma \ref{lem:correctness-conditions-polymorphic}
(with $\xx$ instantiated with $\xx\tinycomma\cmono$).
\QED
\end{proof}

To apply the above lemma, we need to provide an $\xx$ such that the
equation $\encode{\PHI}{\xx\tinycomma\cmono} =
\encode{\PHI}{\cmono\tinycomma\xmono}$ almost holds, in the sense that the two problems
are equisatisfiable without changing the type constructor interpretation, but
possibly changing some of the function or predicate symbol interpretation.

Given $\xmono$, we will come up with an encoding $\xx$ in a systematic way. But
first we need to define some relevant sets of types for a polymorphic
$\Sigma$-problem $\PHI$:

\begin{defi}\label{def-S-T}
Let
\begin{itemize}
\item $\TypesOf_\PHI = \{\sigma \in \Type_\Sigma \mid \mbox{$\sigma$ is the type of a subterm occurring in $\PHI$}\}$;
\item $S_\PHI = \{\tau \in \GType_\Sigma \mid \NONMONO{\tau} \encode{\PHI}{\cmono}
\mbox{ and there exists } \sigma \in \TypesOf_\PHI \mbox{ such that } \tau \leq \sigma\}$;
\item $T_\PHI = \{\tau \in \GType_\Sigma \mid
\mbox{there exists } \sigma \in \TypesOf_\PHI \mbox{ such that } \NONMONO{\sigma} \PHI \mbox{ and } \tau \leq \sigma\}$.
\end{itemize}
\end{defi}

Here is how we proceed with the definition of $\xx$:
\begin{enumerate}
\item[(1)] We emulate the definition of $\xmono$, using the polymorphic monotonicity calculus instead of the monomorphic one.
\betweenitems
\item[(2)] Step (1) will cause $\encode{\PHI}{\xx\tinycomma\cmono}$ to introduce more protectors than
$\encode{\PHI}{\cmono\tinycomma\xmono}$. This is because the former protects all ground types $\tau$ that are instances of
types $\sigma \in \TypesOf_\PHI$ such that $\NONMONO{\sigma} \PHI$ (namely, $T_\PHI$),
whereas the latter protects all ground types $\tau$
that are instances of types in $\TypesOf_\PHI$
and satisfy $\NONMONO{\tau} \encode{\PHI}{\cmono}$ (namely, $S_\PHI$).
We have $S_\PHI \subseteq T_\PHI$ but generally not vice versa.  To repair this mismatch, we add axioms that semantically
eliminate the protectors
for the types in $T_\PHI$ but not in $S_\PHI$.  To achieve this, we must characterise $T_\PHI - S_\PHI$ (which is an infinite set
even for finite problems $\PHI$) as the set of ground instances of a set $U_\PHI$ of types so that $U_\PHI$ is finite whenever
$\PHI$ is.
\end{enumerate}

We define $U_\PHI$ next:

\begin{defi}[Cap]\afterDot
Given a set of types $S$, a \emph{cap} for it is a set $S' \subseteq S$ such that all types in $S$ are instances of types in $S'$.
A cap is
\emph{minimal} if it contains no two distinct types $\sigma$ and $\sigma'$ such that $\sigma \leq \sigma'$.
For each $\sigma$, let $U_\sigma$ be a minimal cap of the set
$$\begin{array}{l@{}l}
  U'_\sigma =  \{\sigma' \leq \sigma \mid {} & \sigma'\in \InfX(\PHI) \mbox{ or there exists no $\mathit{X}^{\sigma''} \in \NV(\PHI)$} \\
               & \mbox{such that $\sigma'$ and $\sigma''$ have a common instance}\}
\end{array}
$$
We let $U_\PHI$ be an arbitrary cap of $\bigcup \{U_\sigma \mid \sigma \in \TypesOf_\PHI \mbox{ and }\NONMONO{\sigma} \PHI \}$.
\end{defi}

The set $U_\PHI$ is both a subset of the monotonic types 
and a \relax{precise} characterisation of the sets of types whose ground instances give
our difference of interest, $T_\PHI - S_\PHI$.

\begin{lem}\label{lem-method}
The following properties hold\/{\rm:}
\begin{enumerate}
\item[(1)] $S_\PHI \subseteq T_\PHI${\rm;}
\item[(2)] If $\sigma' \in U'_\sigma$, then $\MONO{\sigma'} \PHI${\rm;}
\item[(3)] If $\sigma \in U_\PHI$, then $\MONO{\sigma} \PHI${\rm;}
\item[(4)] If $\tau \in T_\PHI - S_\PHI$, then $\MONO{\tau} \cmonox{\PHI}${\rm;}
\item[(5)] $T_\PHI - S_\PHI = \{\tau \in \GType_\Sigma \mid \mbox{there exists } \sigma \in U_\PHI \mbox{ such that } \tau \leq \sigma\}$.
\end{enumerate}
\end{lem}
\begin{proof}
(1):\enskip Assume $\tau \in S_\PHI$ and let $\sigma \geq \tau$ as in the definition of $S_\PHI$.  By Lemma \ref{lem-calc-transport}(2),
we have $\NONMONO{\tau} \PHI$, and hence by Lemma \ref{lem-calc-transport}(3) we have $\NONMONO{\sigma} \PHI$, ensuring that
$\tau \in T_\PHI$.

\betweenitems\noindent
(2):\enskip It is clear that the condition defining $U'_\sigma$ is a strengthening of that defining $\MONO{\sigma} \PHI$.

\betweenitems\noindent
(3):\enskip Immediate from (2).

\betweenitems\noindent
(4):\enskip Immediate from the definitions of $S_\PHI$ and $T_\PHI$.

\betweenitems\noindent
(5):\enskip Let $\tau \in T_\PHI$ such that (A) $\tau \not\in S_\PHI$.
There exists $\sigma$ such that $\tau \leq \sigma \in \TypesOf_\PHI$ and $\NONMONO{\sigma} \PHI$.
From (4), we have $\MONO{\tau} \cmonox{\PHI}$; hence, by Lemma \ref{lem-calc-transport}(2),  $\MONO{\tau} \PHI$.
Since $\tau$ is ground,
either $\tau \in \InfX(\PHI)$ or there exists no $\mathit{X}^{\sigma''} \in \NV(\PHI)$ such that $\tau \leq \sigma''$,
implying $\tau \in U'_\sigma$.
It follows that $\tau$ is an instance of an element of $U_\sigma$, hence of an element of $U_\PHI$, as desired.
Now, assume $\tau$ is ground such that $\tau \leq \sigma' \in U_\PHI$.
We obtain $\sigma \in \TypesOf_\PHI$ such that $\NONMONO{\sigma} \PHI$ and $\sigma' \in U_\sigma$.
In particular, we have $\sigma' \leq \sigma$, hence $\tau \in T_\PHI$.
Moreover, by $\sigma' \in U_\sigma$ and (2), we have $\MONO{\sigma'} \PHI$; hence by Lemma \ref{lem-calc-transport}(1),
we have $\MONO{\tau} \cmonox{\PHI}$, implying $\tau \not\in S_\PHI$, as desired.
\QED
\end{proof}

The following results will hold for any choice of set $V_\PHI$
between $U_\PHI$ and $\{\sigma \in \Type_\Sigma \mid \MONO{\sigma} \PHI\}$. The smaller the chosen set, the lighter the
encoding.

\begin{conv}\label{con-V}
We fix $V_\PHI$ such that $U_\PHI \subseteq V_\phi \subseteq \{\sigma \in \Type_\Sigma \mid \MONO{\sigma} \PHI\}$.
\end{conv}

\subsection{Monotonicity-Based Type Tags}
\label{ssec:monotonicity-based-type-tags-polymorphic}

We are now equipped to present the definitions of the polymorphic
monotonicity-based encodings and prove their correctness.
The polymorphic \tags\query{} encoding can be seen as a hybrid between
traditional tags (\tags{}) and monomorphic lightweight tags (\mono\tags\query):\
as in \tags{}, tags take the form of a function $\ti\LAN\t\RAN(t)$;
as in \mono\tags\query{}, tags are omitted for types that are inferred
monotonic.

The main novelty concerns the typing axioms. The
\mono\tags\query{} encoding omits all typing axioms for monotonic types. In the
polymorphic case, the monotonic type $\t$ might be an instance of a more general,
potentially nonmonotonic type for which tags are generated. For example, if $\a$
is tagged (because it is possibly nonmonotonic) but its instance
$\mathit{list}(\a)$ is not (because it is infinite and hence monotonic),
there will be mismatches between tagged and untagged terms.
Our solution is to add the typing axiom
$\ti\LAN\mathit{list}(\a)\RAN(\mathit{Xs}) \eq \mathit{Xs}$, which allows the
prover to add or remove a tag for the infinite type $\mathit{list}(\a)$. Such
an axiom is sound for any monotonic~type.

\begin{defi}[Lightweight Tags \tags\query]\afterDot \label{def-light-tags-poly}
The encoding \tags\query{} translates polymorphic problems over
$\SIGMA = (\Type, \FF, \PP)$ to polymorphic problems over $(\Type, \FF \mathrel{\uplus} \{\ti : \forall \alpha.\; \alpha \ra \alpha\}, \PP)$.
Its term and formula translations are defined as follows:
\begin{align*}
\tagsqx{f\LAN\t\RAN(\bar t\,)\typ\t} & \,=\, \CT{}{f\LAN\t\RAN(\tagsqx{\bar t\,})} &
\tagsqx{\mathit{X}\typ\t} & \,=\, \CT{}{\mathit{X}} &
\text{with}\;\,\CT{}{t\typ\t} & \,=\, \begin{cases}
  t & \!\!\text{if $\MONO\t \PHI$} \\
  \ti\LAN\t\RAN(t) & \!\!\text{otherwise}
\end{cases}
\end{align*}
The encoding adds the following typing axioms: 
\[\!\begin{aligned}[t]
& \forall\vthinspace\TVars(\sigma).\allowbreak\; \forall \mathit{X} \mathbin: \sigma
.\;\, \ti\LAN \sigma \RAN(\mathit{X}^{\sigma}) \eq \mathit{X}^{\sigma}
  && \text{for~} \sigma \in V_\PHI
\end{aligned}
\]
The \emph{polymorphic lightweight type tags} encoding \tags\query{}
is the composition $\encode{\phantom{i}}{\tags\query\tinycomma\args\tinycomma\erased}$.
It translates a polymorphic problem
over $\SIGMA$
into an untyped problem over $\SIGMA' = (\mathcal{F}' \mathrel{\uplus} \{\ti^2\},\Pp)$, where $\mathcal{F}', \Pp$ are as for
\args.
%
\end{defi}

\begin{exa}%
\label{ex:algebraic-lists-tags-query}%
The \tags\query{} encoding of of Example~\ref{ex:algebraic-lists} follows:
\begin{quotex}
$\forall A, \mathit{Xs}.\;\, \ti(\const{list}(A), \mathit{Xs}) \eq \mathit{Xs}$ \\[\betweenaxs]
$\forall A, \mathit{X}\!, \mathit{Xs}.\;\,
  \const{nil}(A) \not\eq \const{cons}(A, \ti(A, \mathit{X}), \mathit{Xs})$ \\
$\forall A, \mathit{Xs}.\;\,
  \mathit{Xs} \eq \const{nil}(A) \mathrel{\lor}
  (\exists \mathit{Y}\!, \mathit{Ys}.\;\,
  \mathit{Xs} \eq \const{cons}(A, \ti(A, \mathit{Y}), \mathit{Ys})
  )$ \\
$\forall A, \mathit{X}\!, \mathit{Xs}.\;\,
  \ti(A, \const{hd}(A, \const{cons}(A, \ti(A, \mathit{X}), \mathit{Xs}))) \eq \ti(A, \mathit{X})
  \mathrel\land {}$ \\
$\phantom{\forall A, \mathit{X}\!, \mathit{Xs}.\;\,}
  \const{tl}(A, \const{cons}(A, \ti(A, \mathit{X}), \mathit{Xs})) \eq \mathit{Xs}$ \\
$\exists \mathit{X}\!, \mathit{Y}\!, \mathit{Xs}, \mathit{Ys}.\;\,
  \const{cons}(\const{\www}, \ti(\const{\www}, \mathit{X}), \mathit{Xs}) \eq \const{cons}(\const{\www}, \ti(\const{\www}, \mathit{Y}),
\mathit{Ys}) \mathrel\land {}$ \\
$\phantom{\exists \mathit{X}\!, \mathit{Y}\!, \mathit{Xs}, \mathit{Ys}.\;\,} (\ti(\const{\www}, \mathit{X}) \not\eq \ti(\const{\www}, \mathit{Y}) \mathrel\lor \mathit{Xs} \not\eq \mathit{Ys})$
\end{quotex}
The typing axiom allows any term to be typed as $\mathit{list}(\a)$,
which is sound because $\mathit{list}(\a)$ is infinite.
It would have been equally correct to provide separate axioms for \const{nil},
\const{cons}, and \const{tl}. Either way, the axioms are needed to remove
the $\ti(A, \,\mathit{X})$ tags in case the proof requires reasoning about
$\mathit{list}(\mathit{list}(\a))$.
\end{exa}

The lighter encoding \tags\qquery{} protects only naked variables and introduces
equations of the form
$\ti\LAN\t\RAN(\ff\LAN\bar\a\RAN(\mathit{\bar X})) \eq \ff\LAN\bar\a\RAN(\mathit{\bar X})$
to add or remove tags around each function symbol $\ff$ of a possibly
nonmonotonic type $\t$, and similarly for existential variables.

\begin{defi}[Featherweight Tags \tags\qquery]\afterDot \label{def-qquery-poy}
The encoding \tags\qquery{} translates polymorphic problems over
$\SIGMA = (\Type, \FF, \PP)$ to polymorphic problems over $(\Type, \FF \mathrel{\uplus} \{\ti : \forall \alpha.\; \alpha \ra \alpha\}, \PP)$.
Its term and formula translations are defined as follows:
\begin{align*}
\tagsqqx{t_{1\!} \eq t_2} & \,=\, \CT{}{\tagsqqx{t_{1\!}}{}} \eq \CT{}{\tagsqqx{t_2}{}} \\
\tagsqqx{\exists \mathit{X} \mathbin: \t.\;\, \phi}{} & \,=\,
  \exists \mathit{X} \mathbin: \t.\;
  \begin{cases}
    \tagsqqx{\phi}{} & \!\!\text{if~$\MONO\t \PHI$} \\
    \ti\LAN\t\RAN(\mathit{X}) \eq \mathit{X} \mathrel\land \tagsqqx{\phi}{} & \!\!\text{otherwise} \\
  \end{cases}
\end{align*}
\\[-.5\baselineskip] with \\[-.5\baselineskip] 
\[\CT{}{t\typ\t} \,=\, \begin{cases}
  t & \!\!\text{if $\MONO\t \PHI$ or $t \notin \Vtun$} \\
  \ti\LAN\t\RAN(t) & \!\!\text{otherwise}
\end{cases}\]
The encoding adds the typing axioms of \tags\query{} (from Definition \ref{def-light-tags-poly}) and the following:
\[\!\begin{aligned}[t]
& \textstyle \forall\bar\a.\; \forall \mathit{\bar X} \mathbin: \bar\t.\;\, \ti\LAN\t\RAN(f\LAN\bar\a\RAN(\mathit{\bar X}^{\ov{\sigma}}))
\eq f\LAN\bar\a\RAN(\mathit{\bar X}^{\ov{\sigma}})
  &\enskip& \text{for~}\ff : \forall\bar\a.\; \bar\t \to \t \in \mathcal{F}
  \text{~such that~} \NONMONO{\sigma} \PHI
\\[\betwtypax]
& \forall\vthinspace\TVars(\t).\allowbreak\; \exists \mathit{X} \mathbin: \nobreak \t.\;\, \ti\LAN\t\RAN(\mathit{X}^\sigma) \eq \mathit{X}^\sigma
  && \!\begin{aligned}[t] 
        & \text{for~} \t \in \TypesOf_\PHI \text{~such that } \NONMONO{\t} \PHI \text{~ and $\t$ is not } \\[-\jot]
        & \text{an instance of the result type of some~} f \in \mathcal F 
        \end{aligned}
\end{aligned}\]
The \emph{polymorphic featherweight type tags} encoding \tags\qquery{} is the composition
$\encode{\phantom{i}}{\tags\qquery\tinycomma\args
\tinycomma\erased}$.
The target signature of \tags\qquery{} is the same as that of \tags\query{}.
\end{defi}

\begin{exa}%
\label{ex:algebraic-lists-tags-qquery}%
The \tags\qquery{} encoding of Example~\ref{ex:algebraic-lists}
requires fewer tags than \mono\tags\query{}, at the cost of two additional typing axioms
and two typing equations for the existential variables of type~$\www$:
\begin{quotex}
$\forall A, \mathit{Xs}.\;\, \ti(A, \const{hd}(A, \mathit{Xs})) \eq \const{hd}(A, \mathit{Xs})$ \\
$\forall A, \mathit{Xs}.\;\, \ti(\const{list}(A), \mathit{Xs}) \eq \mathit{Xs}$ \\
$\forall A.\; \exists \mathit{X}.\;\, \ti(A, \mathit{X}) \eq \mathit{X}$ \\[\betweenaxs]
$\forall A, \mathit{X}\!, \mathit{Xs}.\;\,
  \const{nil}(A) \not\eq \const{cons}(A, \mathit{X}\!, \mathit{Xs})$ \\
$\forall A, \mathit{Xs}.\;\,
  \mathit{Xs} \eq \const{nil}(A) \mathrel{\lor}
  (\exists \mathit{Y}\!, \mathit{Ys}.\;\,
  \ti(A, \mathit{Y}) \eq \mathit{Y} \mathrel\land \mathit{Xs} \eq \const{cons}(A, \mathit{Y}\!, \mathit{Ys})
  )$ \\
$\forall A, \mathit{X}\!, \mathit{Xs}.\;\,
  \const{hd}(A, \const{cons}(A, \mathit{X}\!, \mathit{Xs})) \eq \ti(A, \mathit{X})
  \mathrel\land \const{tl}(A, \const{cons}(A, \mathit{X}\!, \mathit{Xs})) \eq \mathit{Xs}$ \\
$\exists \mathit{X}\!, \mathit{Y}\!, \mathit{Xs}, \mathit{Ys}.\;\,
  \ti(\const{\www}, \mathit{X}) \eq \mathit{X} \mathrel\land \ti(\const{\www}, \mathit{Y}) \eq \mathit{Y} \mathrel\land
  \const{cons}(\const{\www}, \mathit{X}\!, \mathit{Xs}) \eq \const{cons}(\const{\www}, \mathit{Y}\!,
\mathit{Ys}) \mathrel\land {}$ \\
$\phantom{\exists \mathit{X}\!, \mathit{Y}\!, \mathit{Xs}, \mathit{Ys}.\;\,}
  (\mathit{X} \not\eq \mathit{Y} \mathrel\lor \mathit{Xs} \not\eq \mathit{Ys})$
\end{quotex}
\end{exa}

\begin{thm}[Correctness of \tags\query, \tags\qquery]\afterDot
\label{thm:correctness-of-tags-query}%
The polymorphic type tags encodings\/ \hbox{\rm\tags\query} and\/ \hbox{\rm\tags\qquery}
are correct.
\end{thm}

\begin{proof}
First we discuss the case of \tags\query{}.
The following property can be routinely checked:
\begin{itemize}[label=(A)]
\item[(A)] $\encode{\PHI}{\tags\query\tinycomma\cmono} =
  \{\phi^\#  \mid \phi \in \encode{\PHI}{\cmono\tinycomma\mono\tags\query}\} \,\cup\, \{\forall \mathit{X}^\tau.\; \ti_\tau(\mathit{X}^\tau) \eq \mathit{X}^\tau \mid \tau \in
\GInst(V_\PHI)\}$
\end{itemize}
where $\GInst(V_\PHI)$ denotes the set of ground instances of types in $V_\PHI$
and $\phi^\#$ denotes the modification of $\phi$ obtained by adding,
for each $\tau \in T_\PHI - S_\PHI$, a tag $\ti_\tau$ around every term of type $\tau$.
Recall that $U_\PHI \subseteq V_\phi \subseteq \{\sigma \in \Type_\Sigma \mid \MONO{\sigma} \PHI\}$.
We have the following:
\begin{enumerate}
\item[(1)] $T_\PHI - S_\PHI \subseteq \GInst(V_\PHI)$;
\item[(2)] $\GInst(V_\PHI) \subseteq \{\sigma \in \GType_\Sigma \mid \MONO{\sigma} \PHI\}$;
\item[(3)] if $\tau \in V_\PHI$, the symbol $\ti_\tau$ does not occur in $\encode{\phantom{i}}{\cmono\tinycomma\mono\tags\query}$.
\end{enumerate}
Property (1) follows from Lemma \ref{lem-method}(5),
(2) from Lemma \ref{lem-calc-transport}(3),
and (3) from (2) and the definition of $\encode{\phantom{i}}{\mono\tags\query}$.
Now we are ready to check the conditions of Lemma~\ref{lem-transport}.

\betweenitems\noindent
\textsc{Sound}:\enskip
  Given a model $\MM$ of $\encode{\PHI}{\cmono\tinycomma\mono\tags\query}$, we modify it into a structure $\NN$ by
  reinterpreting, for all $\tau \in \GInst(V_\PHI)$, the symbols $\ti_\tau$ as identity. Thanks to (3), $\NN$ is still a model
  of $\encode{\PHI}{\cmono\tinycomma\mono\tags\query}$ and, thanks to (A) and (1), it is also a model of $\encode{\PHI}{\tags\query\tinycomma\cmono}$.

\betweenitems\noindent
\textsc{Complete}:\enskip
 Any model of $\encode{\PHI}{\tags\query\tinycomma\cmono}$ is also a model of
$\encode{\PHI}{\cmono\tinycomma\mono\tags\query}$, since, thanks to (1),
in the presence of $\{\forall \mathit{X}^\tau.\; \ti_\tau(\mathit{X}^\tau) \eq \mathit{X}^\tau \mid \tau \in
\GInst(V_\PHI)\}$ each $\phi^\#$ is equivalent to $\phi$.

\betweenitems\indent
The case of \tags\qquery{} is similar, with the following modifications:
\begin{itemize}[label=(A)]
\item[(A)] $\encode{\PHI}{\tags\query\tinycomma\cmono} =
  \{\phi^\#  \mid \phi \in \encode{\PHI}{\cmono\tinycomma\mono\tags\query}\} \,\cup\, \{\forall \mathit{X}^\tau.\; \ti_\tau(\mathit{X}^\tau) \eq \mathit{X}^\tau \mid \tau \in
\GInst(V_\PHI)\} \,\cup\, \Ax$
\end{itemize}
where $\Ax$ consists of all ground instances $\phi\,\rho$ of the additional
axioms $\forall \ov{\alpha}.\; \phi$ from Definition~\ref{def-qquery-poy} where the occurring $\ti_\tau$ is such that $\tau \in T_\PHI - S_\PHI$.
What we retain about $\Ax$ is that $\Ax$ is satisfied by any structure that interprets as identity each $\ti_\tau$ with $\tau \in T_\PHI - S_\PHI$,
and hence, by (1), we have
\begin{enumerate}
\item[(4)] $\Ax$ is satisfied by any structure that interprets as identity each $\ti_\tau$ with $\tau \in \GInst(V_\PHI)$.
\end{enumerate}
Moreover, $\phi^\#$ is modified to add tags to $\phi$ in fewer places---not to
arbitrary terms, but only to naked universal variables.
Then the proof for \tags\query{} works here too, additionally invoking (4) in the proof of soundness.
\QED
\end{proof}

\subsection{Monotonicity-Based Type Guards}
\label{ssec:monotonicity-based-type-guards-polymorphic}

Analogously to \tags\query{}, the \guards\query{} encoding is best
understood as a hybrid between traditional guards (\guards{}) and monomorphic
lightweight guards (\mono\guards\query):\ as in \guards{}, guards take the form
of a predicate $\is\LAN\t\RAN(t)$; as in \mono\guards\query{}, guards are
omitted for types that are inferred monotonic.

Once again, the main novelty concerns the typing axioms. The
\mono\guards\query{} encoding omits all typing axioms for monotonic types. In the
polymorphic case, the monotonic type $\t$ might be an instance of a more general,
potentially nonmonotonic type for which guards are generated. Our solution
is to add the typing axiom $\is\LAN\t\RAN(\mathit{X})$, which allows the
prover to discharge any guard for the monotonic type $\t$.

\begin{defi}[Lightweight Guards \guards\query]\afterDot
The encoding \guards\query{} translates polymorphic problems over
$\SIGMA = (\Type, \FF, \PP)$ to polymorphic problems over $(\Type, \FF , \PP \mathrel{\uplus} \{\is : \forall \alpha.\; \alpha \ra \bool\})$.
Its term and formula translations $\encode{\phantom{i}}{\guards\query\tinycomma\args
\tinycomma\erased}$ are defined as follows:
\newcommand\GSQXCHEAT{\kern.75ex} 
\[\begin{aligned}[t]
& \guardsqx{\forall \mathit{X} \mathbin: \t.\;\, \phi} & \,=\,
\forall \mathit{X} \mathbin: \t.\;
\begin{cases}
\guardsqx{\phi} & \!\!\text{if $\MONO\t \PHI$} \\
\is\LAN\t\RAN(\mathit{X}^\sigma) \rightarrow \guardsqx{\phi} & \!\!\text{otherwise}
\end{cases} \\
& \guardsqx{\exists \mathit{X} \mathbin: \t.\;\, \phi} & \,=\,
\exists \mathit{X} \mathbin: \t.\;
\begin{cases}
\guardsqx{\phi} & \GSQXCHEAT\!\!\text{if $\MONO\t \PHI$} \\
\is\LAN\t\RAN(\mathit{X}^\sigma) \mathrel\land \guardsqx{\phi} & \GSQXCHEAT\!\!\text{otherwise}
\end{cases}
\end{aligned}\]
The encoding adds the following typing axioms:
\[\!\begin{aligned}[t]
& \textstyle \forall\vthinspace\TVars(\sigma).\; \forall \mathit{X} \mathbin: \bar\sigma.\;\, \is\LAN\sigma\RAN(\mathit{X}^\sigma)
  && \text{for~} \sigma \in V_\PHI \\
& \textstyle \forall\bar\a.\; \forall \mathit{\bar X} \mathbin: \bar\t.\;\, \is\LAN\t\RAN(f\LAN\bar\a\RAN(\mathit{\bar X}^{\ov{\sigma}}))
  &\enskip& \text{for~}\ff : \forall\bar\a.\; \bar\t \to \t \in \mathcal{F}
  \text{~such that~} \NONMONO{\sigma} \PHI
%
\\
& \forall\vthinspace\TVars(\t).\; \exists \mathit{X} \mathbin: \t.\; \is\LAN\t\RAN(\mathit{X}^\sigma)
  && \!\begin{aligned}[t] 
        & \text{for~} \t \in \TypesOf_\PHI \text{~such that } \NONMONO{\t} \PHI \text{~ and $\t$ is not } \\[-\jot]
        & \text{an instance of the result type of some~} f \in \mathcal F 
        \end{aligned}
\end{aligned}\]
The \emph{polymorphic lightweight type guards} encoding \guards\query{}
is the composition $\encode{\phantom{i}}{\guards\query\tinycomma\args\tinycomma\erased}$.
It translates a polymorphic problem
over $\SIGMA$
into an untyped problem over $\SIGMA' = (\mathcal{F}',\Pp \mathrel{\uplus} \{\is^2\})$, where $\mathcal{F}', \Pp$ are as for
\args.
\end{defi}


The featherweight cousin is a straightforward generalisation of \guards\query{} %
along the lines of the generalisation of \mono\guards\query{} into \mono\guards\qquery.

\begin{defi}[Featherweight Guards \guards\qquery]\afterDot
The \emph{polymorphic featherweight type guards} encoding \guards\qquery{} is
identical to the lightweight encoding \guards\query{} except that the
condition ``if~$\MONO\t \PHI$'' 
in the $\forall$ case is
weakened to ``if~$\MONO\t \PHI$ or $\mathit{X}^\sigma \notin \NV(\phi)$''.
\end{defi}

\begin{exa}%
\label{ex:algebraic-lists-guards-qquery}%
The \guards\qquery{} encoding of Example~\ref{ex:algebraic-lists} follows:
\begin{quotex}
$\forall A, \mathit{Xs}.\;\, \is(A, \const{hd}(A, \mathit{Xs}))$ \\
$\forall A, \mathit{Xs}.\;\, \is(\const{list}(A), \mathit{Xs})$ \\[\betweenaxs]
$\forall A, \mathit{X}\!, \mathit{Xs}.\;\,
  \const{nil}(A) \not\eq \const{cons}(A, \mathit{X}\!, \mathit{Xs})$ \\
$\forall A, \mathit{Xs}.\;\,
  \mathit{Xs} \eq \const{nil}(A) \mathrel{\lor}
  (\exists \mathit{Y}\!, \mathit{Ys}.\;\,
  \is(A, \mathit{Y}) \mathrel\land \mathit{Xs} \eq \const{cons}(A, \mathit{Y}\!, \mathit{Ys})
  )$ \\
$\forall A, \mathit{X}\!, \mathit{Xs}.\;\,
  \is(A, \mathit{X}) \rightarrow \const{hd}(A, \const{cons}(A, \mathit{X}\!, \mathit{Xs})) \eq \mathit{X} \mathrel\land \const{tl}(A, \const{cons}(A, \mathit{X}\!, \mathit{Xs})) \eq \mathit{Xs}$ \\
\def\nocomma{\mskip-1mu} 
$\exists \mathit{X}\!,\nocomma \mathit{Y}\!,\nocomma \mathit{Xs},\nocomma \mathit{Ys}.\;\,
\is(\const{\www},\nocomma \mathit{X}) \mathrel\land \is(\const{\www},\nocomma \mathit{Y}) \mathrel\land
\const{cons}(\const{\www},\nocomma \mathit{X}\!,\nocomma \mathit{Xs}) \eq \const{cons}(\const{\www},\nocomma \mathit{Y}\!,\nocomma \mathit{Ys})
\mathrel\land (\mathit{X} \mathbin{\not\eq} \mathit{Y} \mathrel{\lor} \mathit{Xs} \mathbin{\not\eq} \mathit{Ys})$
\end{quotex}
\end{exa}

\begin{thm}[Correctness of \guards\query, \guards\qquery]\afterDot
\label{thm:correctness-of-guards-query}%
The polymorphic type guards encodings\/ \hbox{\rm\guards\query} and\/ \hbox{\rm\guards\qquery}
are correct.
\end{thm}
\begin{proof}
The argument is very similar to that for tags (Theorem
\ref{thm:correctness-of-tags-query}), with the following modifications:
\begin{itemize}[label=(A)]
\item[(A)] $\encode{\PHI}{\tags\query\tinycomma\cmono} =
  \{\phi^\#  \mid \phi \in \encode{\PHI}{\cmono\tinycomma\mono\tags\query}\} \,\cup\, \{\forall \mathit{X}^\tau.\; \is_\tau(\mathit{X}^\tau) \mid \tau \in
\GInst(V_\PHI)\} \,\cup\, \Ax$
\end{itemize}
where $\Ax$ consists of a set of formulae which is satisfied by each structure that interprets as everywhere-true
each $\is_\tau$ with $\tau \in \GInst(V_\PHI)$ and
$\phi^\#$ denotes the modification of $\phi$ obtained by adding,
for each $\tau \in T_\PHI - S_\PHI$, a guard $\is_\tau$ on several negative positions in $\phi$.

Apart from this, the proofs for \guards\query{} and \guards\qquery{} are identical to that for
\tags\qquery{}, except that instead of interpreting tags as identity we interpret guards as everywhere-true.
\QED
\end{proof}


\section{Implementation}
\label{sec:implementation}

Our research on polymorphic type encodings was driven by Sledgehammer, a
component of Isabelle\slash HOL that harnesses first-order automatic theorem provers
to discharge interactive proof obligations. The tool heuristically selects
hundreds of background facts, translates them to untyped or monomorphic
first-order logic, invokes the external provers in parallel, and reconstructs
machine-generated proofs in Isabelle.

All the encodings presented in this \Paper\ except for the complete monomorphisation
translation from Section~\ref{subsec:complete-mono} (which plays only an auxiliary theoretical role in our development)
enjoy the desirable property that if the source signature and problem are finite,
then the target signature and problem are also finite.
All of these encodings, including the traditional ones, are implemented in Sledgehammer and can be used to target
external first-order provers. The rest of this section considers
implementation issues in more detail.

\subsection{Heuristic Monomorphisation Algorithm}
\label{ssec:heuristic-monomorphisation-algorithm}

The monomorphisation algorithm implemented in Sledgehammer translates a
polymorphic problem into a monomorphic problem by heuristically
instantiating type variables. It involves three stages:
\begin{enumerate}[label=\arabic*.]
\item[1.] Separate the monomorphic and the polymorphic formulae, and collect all
symbols occurring in the monomorphic formulae (the \LQ mono-symbols\RQ).

\betweenitems

\item[2.] For each polymorphic axiom, stepwise refine a set of substitutions,
starting from the singleton set containing only the empty substitution, by
matching known mono-symbols against their polymorphic counterparts in the axiom.
So long as new mono-symbols emerge, collect them and repeat this stage.

\betweenitems

\item[3.] Apply the substitutions to the corresponding polymorphic
formulae. Only keep fully monomorphic formulae.
\end{enumerate}

\noindent To ensure termination, the iterations performed in stage~2
are limited to a configurable number $K$.
%
To curb the exponential growth, the algorithm also enforce an upper bound
$\mathrm{\Delta}$ on the number of new formulae. Sledgehammer operates with $K =
3$ and $\mathrm{\Delta} = 200$ by default, so that a problem with 500
formulae comprises at most 700~formulae after monomorphisation.
Experiments found these values suitable.
Given formulae about $\www{}$ and
$\mathit{list}(\a)$, the third iteration already generates
$\mathit{list}(\mathit{list}(\mathit{list}(\www)))$ instances;
adding yet another layer of $\mathit{list}$ is unlikely to help.
Increasing $\mathrm{\Delta}$ sometimes helps solve more problems, but its
potential for clutter is real.

\subsection{Extension to Higher-Order Logic}
\label{ssec:extension-to-higher-order-logic}

Isabelle/HOL's logic, polymorphic higher-order logic with axiomatic
type classes \cite{wenzel-1997}, is not the same as the polymorphic first-order
logic considered in this \Paper. Sledgehammer's translation is a three-step process, where
the first and last step may be omitted, depending on whether monomorphisation is
desired and whether the target prover supports monomorphic types:
\begin{enumerate}
\item[1.] \emph{Optionally monomorphise the problem.}
\betweenitems
\item[2.] \emph{Eliminate the higher-order constructs} \cite[\S2.1]{meng-paulson-2008-trans}.
$\lambda$-abstractions are rewritten to \const{SK} combinators
or to supercombinators ($\lambda$-lifting).
Functions are passed varying numbers of arguments via an
apply operator
$\const{hAPP} : \forall\a, \b.\;\mathit{fun}(\a,\b) \times \a
\to \b$ (where $\mathit{fun}$ is uninterpreted). Boolean 
terms are converted to formulae using a unary predicate
$\const{hBOOL} : \mathit{bool} \to \bool$ (where $\mathit{bool}$ is uninterpreted).
%
\betweenitems
\item[3.] \emph{Encode the type information.}
Polymorphic types are encoded using the techniques
described in this \Paper. Type classes are essentially sets of types;
they are encoded as polymorphic predicates $\forall\a.\>\bool$
(where $\a$ is a phantom type variable, Definition~\ref{def:phantom-type-argument}).
For example, a predicate $\const{linorder} : \forall\a.\>\bool$ could be used to
restrict the axioms specifying that $\const{less\_eq} : \forall\a.\;
\a\times\a\to\bool$ is a linear order to those types that satisfy the
\const{linorder} predicate (cf.\ Example~\ref{ex:linorder}).
The type class hierarchy is expressible as Horn
clauses \cite[\S2.3]{meng-paulson-2008-trans}.
\end{enumerate}\smallskip

\noindent The symbol \hAPP{} would hugely burden problems if it were introduced
systematically for all arguments to functions. To reduce clutter, Sledgehammer
computes the minimum arity~$n$ needed for each symbol and passes the first
$n$~arguments directly, falling back on \hAPP{} for additional arguments. In
general, more arguments can be passed directly if monomorphisation is performed
before \hAPP{} is introduced, because each monomorphic instance of a polymorphic
symbol is considered individually. Similar observations can be made for
\hBOOL{}.

\subsection{Infinite Types and Constructors}
\label{ssec:infinite-types-and-constructors}

The monotonicity calculus $\rhd$ is parameterised by a set $\Inf(\PHI)$ of
infinite types. One could employ an approach similar to that implemented in
Infinox \cite{claessen-lilliestrom-2011} to automatically infer finite
unsatisfiability of types. This tool relies on various proof principles to show
that a set of untyped first-order formulae only has models with infinite
domains. For example, it can infer that $\mathit{list\IUS}\www$ is infinite in
Example~\ref{ex:algebraic-lists-monomorphised} because
$\const{cons}_{\vvthinspace\www}$ is injective in its second argument but not
surjective.
However, in a proof assistant such as Isabelle, it is simpler to exploit
metainformation available through introspection. Isabelle's datatypes are
registered with their constructors; if some of them are recursive, or take an
argument of an infinite type, the datatype must be infinite and hence monotonic.

More specifically, the monotonicity inference is run on the entire problem and
maintains two finite sets of polymorphic types:\ the surely infinite types $J$
and the possibly nonmonotonic types $N$. Every type of a naked
variable in the problem is tested for infinity. If the test succeeds, the type
is inserted into $J$; otherwise, it is inserted into $N$. Simplifications are
performed:\ there no need to insert $\t$ to $J$ or $N$ if it is an instance
of a type already in the set; when inserting $\t$ to a set, it is safe to remove
any type in the set that is an instance of~$\t$. The monotonicity check then
becomes
\[\MONO{\t} \PHI \;\Longleftrightarrow\;
 (\exists\u \in J.\; \exists\rho.\; \t = \u\vthinspace\rho) \,\mathrel\lor\,
 (\forall\u \in N.\; \nexists\rho.\; \t\vthinspace\rho = \u\vthinspace\rho)\]
%

\subsection{Proof Reconstruction}
\label{ssec:proof-reconstruction}

To guard against bugs in the external provers, Sledgehammer reconstructs
machine-generated proofs in Isabelle. This is usually accomplished by the
\textit{metis} proof method \cite{paulson-susanto-2007},
supplying it with the short list of facts referenced in the proof found by the
prover. The proof method is based on the Metis prover \cite{hurd-2003}, a
complete resolution prover for untyped first-order logic.
The \textit{metis} call is all that remains from the Sledgehammer invocation in
the Isabelle theory, which can then be replayed without external provers. Given
only a handful of facts, \textit{metis} usually succeeds within milliseconds.

Prior to our work, a large share of the reconstruction failures were caused by
type-unsound proofs found by the external provers, due to the use of the unsound
encoding~\args{} \cite[\S4.1]{boehme-nipkow-2010}. We now replaced the internals
of Sledgehammer and \textit{metis} so that they use a translation module
supporting all the type encodings described in this \Paper.

Nonetheless, despite the typing information, individual inferences in Metis can be
ill-typed when types are reintroduced, causing the \textit{metis} proof method
to fail. There are two main failure scenarios.

First, the prover may in principle instantiate variables with ``ill-typed''
terms at any point in the proof. Fortunately, this hardly ever arises in
practice, because like other resolution provers Metis heavily restricts
paramodulation from and into variables \cite{
bachmair-et-al-1995}.

An issue that is more likely to plague users concerns
the infinite types $\Inf(\PHI)$.
In the theoretical part of the paper, we required infinity to be a consequence
of the problem $\PHI$. The implementation is less rigorous; it will happily
treat types that are known to be infinite in the Isabelle background theories
even if $\PHI$ itself does not imply infinity of the types.
For example, assuming \textit{nat} is known to be infinite, the implementation of
the monotonicity-based encodings
will not introduce any protectors around the naked variables $M$ and $N$ when
translating the problem
\[\const{on} \not\eq \const{off}\typ{state} \mathrel\land (\forall \mathit{X}\!, \mathit{Y} \mathbin: \mathit{nat}.\; \mathit{X} \eq \mathit{Y})\]
(where the second conjunct is presumably the negation of a conjecture
stating that there exist two distinct natural numbers).
That problem is satisfiable on its own but unsatisfiable with respect to the
background theory.
Untyped provers will instantiate $\mathit{X}$ and $\mathit{Y}$ with \const{on}
and \const{off} to derive a contradiction; and no ``type-sound'' proof is
possible unless we also provide characteristic theorems for \textit{nat}. In
general, we would need to provide infinity axioms for all types in $\Inf(\PHI)$
to make the encoding sound irrespective of the background theory; for example:
\begin{align*}
& \forall N \mathbin: \textit{nat}.\;\, \const{zero} \not\eq \const{suc}(N) &
& \forall M, N \mathbin: \textit{nat}.\;\, \const{suc}(M) \eq \const{suc}(N) \rightarrow M \eq N
\end{align*}
Although this now makes a sound proof possible (by instantiating $\mathit{X}$ and $\mathit{Y}$
with $\const{zero}$ and $\const{suc}(\const{zero})$), it does not prevent the
prover from discovering the spurious proof with \const{on} and \const{off},
which cannot be reconstructed by \textit{metis}.

\leftOut{
A similar issue affects the constructors $\Ctor(\PHI)$. Let
$\const{c}, \const{d} : \forall\a.\; k(\a)$
be among the constructors for $\mathit k$. The encodings based on \args$^\ctor$
will translate the satisfiable problem
\[\const{c}\LAN\mathit a\RAN \not\eq \const{d}\LAN\mathit a\RAN \mathrel\land \const{c}\LAN\mathit b\RAN \eq \const{d}\LAN\mathit b\RAN\]
into an unsatisfiable one:\ $\const{c} \not\eq \const{d} \mathrel\land \const{c} \eq \const{d}$.
In general, we would need to provide distinctness and injectivity axioms
for all constructors in $\Ctor(\PHI)$ to make the encoding
sound irrespective of the background theory; for example:
$\forall\a.\; \const{c}\LAN\a\RAN \not\eq \const{d}\LAN\a\RAN$.
}

\leftOut{
We stress that the issue above does not indicate unsoundness in
our encodings. Rather, we use metainformation from Isabelle 
to guide the translation;
the translation is then sound assuming that the metainformation
holds. If we do not encode that metainformation in the typed problem,
it might have a model where the metainformation
is \relax{false}, while the translation might exploit the
metainformation and be unsatisfiable, as in the
examples above, causing the proof reconstruction to fail.
}

Although the above scenarios rarely occur in practice, it would be more
satisfactory if proof reconstruction were always possible. A solution would be to
connect our formalised soundness proofs with a verified checker for untyped
first-order proofs. This remains for future work.

\section{Evaluation}
\label{sec:evaluation}

To evaluate the type encodings described in this \Paper, we put together
two sets of 1000 
polymorphic first-order problems originating from 10
existing Isabelle 
theories, translated with Sledgehammer's help (100 problems per theory).%
\footnote{The TPTP benchmark suite \cite{sutcliffe-tptp}, which is customarily used for
evaluating theorem provers, has just begun collecting polymorphic
(TFF1) problems.} 
Nine of the theories are the same as in a previous evaluation
\cite{blanchette-et-al-2013-smt}; the tenth one is an optimality proof for
Huffman's algorithm. Our test data are publicly available \cite{our-data}.

The problems in the first benchmark set include about 50
heuristically selected facts (before monomorphisation);
that number is increased to 500 for the second
set, to reveal how well the encodings scale with the problem size.

We evaluated each type encoding with five modern automatic theorem provers:\
the resolution provers E~1.8 \cite{schulz-2004},
SPASS 3.8ds \cite{blanchette-et-al-2012-spass},
and Vampire 3.0 (revision 1803) \cite{riazanov-voronkov-2002}
and the SMT solvers Alt-Ergo 0.95.2 \cite{bobot-et-al-2008}
and Z3 4.3.2 (revision 944dfee008) \cite{de-moura-bjoerner-2008}.
To make the evaluation more informative, we also
tested the provers' native support for types where it is available;
it is referred to as \kern.05ex\mono\native\kern.05ex{} (monomorphic)
and \native{} (polymorphic). Only Alt-Ergo supports polymorphic types natively.

Each prover was invoked with
the set of options we had previously determined worked best for Sledgehammer.%
\footnote{The setup for E was suggested by Stephan Schulz and includes the
little known \LQ symbol offset\RQ{} weight function. We ran Alt-Ergo with the
default setup, SPASS in Isabelle mode, Vampire in CASC mode, and Z3
through the \texttt{z3\_tptp} wrapper but otherwise with the default setup.}
The provers were granted 15 seconds of CPU time per problem
on one core of a 3.06~GHz Dual-Core Intel Xeon processor.
Most proofs were found within a few seconds; a higher time limit would have had
a restricted impact on the success rate \cite[\S4]{boehme-nipkow-2010}.
To avoid giving the unsound encodings (\erased{} and \args{}) 
an
unfair advantage, for these proof search was followed by a certification phase
that attempted to re-find the proof using a combination of sound encodings,
based on its referenced facts. This phase slightly penalises the unsound
encodings by rejecting a few sound proofs, but such is the price of
unsoundness in practice.

\newbox\hackBoxA
\newcommand\Z{\kern.05em}
\newcommand\Oh{\phantom{0}}
\newcommand\N{%
\def\negdot{\setbox\hackBoxA=\hbox{\hphantom{.}}\kern-\wd\hackBoxA}%
\phantom{0}\negdot}

\newcommand\headersolved{%
 &
 & \multicolumn{1}{@{}c}{\erased}
 & \multicolumn{1}{@{}c}{\args}
 & \multicolumn{1}{@{}c}{\tags}
 & \multicolumn{1}{@{}c}{\tags\query}
 & \multicolumn{1}{@{}c}{\tags\qquery}
 & \multicolumn{1}{@{}c}{\tags\at}
 & \multicolumn{1}{@{}c}{\guards}
 & \multicolumn{1}{@{}c}{\guards\query}
 & \multicolumn{1}{@{}c}{\guards\qquery}
 & \multicolumn{1}{@{}c}{\guards\at}
 & \multicolumn{1}{@{}c}{\native}
 \\}

\newcommand\MU[1]{#1 & -- & --}

\begin{figure}[t!]
\begin{minipage}{\textwidth}%
\centerline{\begin{tabular}\tabularStuff
\headersolved
\midrule
{Alt-Ergo}\zooh && 268 
& 275 & 243 & 218 & 296 & 272 & 233 & 291 & 295 & 259 & \relax 301 \\[-.15ex]
    & \monom & 293 & --\, 
& 245 & 293 & 301 & --\, & 296 & 303 & \win 304 & --\, & 302 \\
{E}\zooh && 319 & 330 
& 322 & 311 & 325 & 269 & 268 & 325 & \relax 336 & 288 & --\, \\[-.15ex]
    & \monom & 338 & --\, 
& 337 & \win 353 & 343 & --\, & 335 & 351 & 352 & --\, & --\, \\
{SPASS}\zooh && 289 & 316 
& 290 & 275 & \relax 319 & 196 & 223 & 302 & 314 & 244 & --\, \\[-.15ex]
    & \monom & 322 & --\, 
& 322 & 328 & 332 & --\, & 317 & 330 & 337 & --\, & \win 349 \\
{Vampire}\zooh && 328 & \relax 335 
& 291 & 304 & 320 & 243 & 281 & 320 & 328 & 301 & --\, \\[-.15ex]
    & \monom & 340 & --\,
& 331 & 355 & 351 & --\, & 329 & 356 & 354 & --\, & \win 370 \\
{Z3}\zooh && 295 & \relax 323 
& 299 & 245 & 317 & 311 & 289 & 301 & 317 & 295 & --\, \\[-.15ex]
    & \monom & 326 & --\,
& 307 & 320 & 341 & --\, & 342 & \win 344 & \win 344 & --\, & 343 \\
\end{tabular}}
\caption{\vthinspace Number of solved problems with 50~facts}
\label{fig:number-of-solved-problems-with-50-facts}
\end{minipage}

\betweenfigs

\begin{minipage}{\textwidth}%
\centerline{\begin{tabular}{\tabularStuff}
\headersolved
\midrule
{Alt-Ergo}\zooh && 268 & 218 & 183 & 223 & \relax 273 & 216 & 157 & 240 & \relax 273 & 190 & 260 \\[-.15ex]
    & \monom & 285 & --\, & 209 & 302 & 311 & --\, & 251 & 311 & 313 & --\, & \win 318 \\
{E}\zooh && 170 & \relax 355 & 261 & 271 & 353 & 193 & 205 & 345 & \relax 355 & 237 & --\, \\[-.15ex]
    & \monom & 376 & --\, & 331 & 388 & \win 393 & --\, & 316 & 388 & 390 & --\, & --\, \\
{SPASS}\zooh && 136 & \relax 321 & 289 & 259 & 299 & 154 & 163 & 278 & 296 & 194 & --\, \\[-.15ex]
    & \monom & 333 & --\, & 287 & 323 & 333 & --\, & 229 & 327 & 334 & --\, & \win 343 \\
{Vampire}\zooh && 318 & 396 & 164 & 286 & 365 & 230 & 219 & 312 & 349 & 259 & --\, \\[-.15ex]
    & \monom & 406 & --\, & 231 & 384 & 406 & --\, & 314 & 406 & 405 & --\, & \win 440 \\
{Z3}\zooh && 247 & 358 & 253 & 241 & 362 & 273 & 213 & 305 & \relax 363 & 270 & --\, \\[-.15ex]
    & \monom & 363 & --\, & 260 & 358 & 369 & --\, & 349 & \win 370 & 369 & --\, & \win 370 \\
\end{tabular}}
\caption{\vthinspace Number of solved problems with 500~facts}
\label{fig:number-of-solved-problems-with-500-facts}
\end{minipage}
\end{figure}

\begin{figure}
\begin{minipage}{\textwidth}%
\centerline{\begin{tabular}{@{\kern.7em}l@{\kern.7em}r@{\kern1.5em}r@{\kern.7em}r@{\kern1em}r@{\kern.7em}r@{\kern.7em}r@{\kern.7em}r@{\kern1em}r@{\kern.7em}r@{\kern.7em}r@{\kern.7em}r@{\kern.7em}}
&
  & \multicolumn{1}{c@{}}{\erased}
  & \multicolumn{1}{c@{}}{\args}
  & \multicolumn{1}{c@{}}{\tags}
  & \multicolumn{1}{c@{}}{\tags\query}
  & \multicolumn{1}{c@{}}{\tags\qquery}
  & \multicolumn{1}{c@{}}{\tags\at}
  & \multicolumn{1}{c@{}}{\guards}
  & \multicolumn{1}{c@{}}{\guards\query}
  & \multicolumn{1}{c@{}}{\guards\qquery}
  & \multicolumn{1}{c@{}}{\guards\at}
  \\
\midrule
{Clauses}\zooh &&
  89 & 99 & 100 & 108 & 140 & 167 & 166 & 139 & 139 & 166 \\[-.15ex]
& \monom &
  125 & --\, & 127 & 127 & 141 & --\, & 242 & 141 & 141 & --\, \\
{{Literals per clause~~}\zooh} &&
  2.3 & 2.4 & 2.4 & 2.3 & 2.2 & 2.5 & 4.3 & 3.2 & 2.6 & 3.8 \\[-.15ex]
& \monom &
  2.3 & --\, & 2.3 & 2.3 & 2.2 & --\, & 4.4 & 2.5 & 2.4 & --\, \\
{{Symbols per atom}\zooh} &&
  6.3 & 8.0 & 18.3 & 16.0 & 10.3 & 9.9 & 5.7 & 8.0 & 8.6 & 5.7 \\[-.15ex]
& \monom &
  6.2 & --\, & 10.6 & 7.0 & 6.2 & --\, & 4.3 & 5.5 & 5.7 & --\, \\
{{Symbols}\zooh} &&
  1276 & 1870 & 4339 & 3924 & 3235 & 4070 & 4051 & 3609 & 3103 & 3610 \\[-.15ex]
& \monom &
  1757 & --\, & 3060 & 2040 & 1935 & --\, & 4548 & 1951 & 1904 & --\,
\end{tabular}}
\caption{\vthinspace Average size of clausified problems with 50~facts}%
\label{fig:average-size-of-clausified-problems-with-50-facts}
\end{minipage}
\end{figure}

Figures \ref{fig:number-of-solved-problems-with-50-facts} and
\ref{fig:number-of-solved-problems-with-500-facts} give,
for each combination of
prover and encoding,
the number of solved problems from each problem set.
Rows marked with \monom{} concern the monomorphic encodings. The encodings
\mono\args{},
\mono\tags\at, and \mono\guards\at{} are omitted;
the first two coincide with \mono\erased{},
whereas \mono\tags\at{} and \mono\guards\at{}
are identical to degenerate versions of
\mono\tags\qquery{} and \mono\guards\qquery{} that treat all types as
possibly nonmonotonic.
We observe the following:%
\begin{itemize}
\item
Among the encodings to untyped first-order logic, the monomorphic
monotonicity-based encodings (especially \mono\guards\qquery{} but also
\mono\tags\qquery, \mono\guards\query, and \mono\tags\query) performed best
overall. Their performance is close to that of the provers'
native support for simple types~(\kern.05ex\mono\native\kern.05ex{}).
Polymorphic encodings lag behind; this is likely due in part to
the synergy between the monomorphiser and the translation of higher-order
constructs (cf.\ \shortsect~\ref{ssec:extension-to-higher-order-logic}).

\betweenitems

\item
Among the polymorphic encodings, our novel cover-based and monotonicity-based encodings
(\tags\at{}, \tags\query{}, \tags\qquery{}, \guards\at{}, \guards\query{}, and \guards\qquery{}),
with the exception of \tags\at{}, constitute a substantial improvement over the traditional
sound schemes (\tags{} and \guards{}).

\betweenitems

\item As suggested in the literature, there is no clear winner between tags and
guards. We expected monomorphic guards to be especially effective with SPASS,
since they are internally mapped to soft sorts (an optimised representation of
unary predicates \cite{weidenbach-2001}), but this is not corroborated by the data.

\betweenitems

\item
Despite the proof reconstruction phase, the unsound encoding \args{}
achieved similar results to the sound polymorphic encodings.
In contrast, its monomorphic cousin
\mono\erased{} is generally no match for the sound monomorphic schemes.

\betweenitems

\item
Oddly, the polymorphic prover Alt-Ergo performs significantly better on monomorphised problems than
on the corresponding polymorphic ones. This raises serious doubts about the quality of the prover's
heuristics for instantiating type variables.
\end{itemize}%

\noindent For the first benchmark set, Figure
\ref{fig:average-size-of-clausified-problems-with-50-facts} presents the
average number of clauses, literals per clause, symbols per atom, and symbols
for clausified problems (using E's clausifier), to give an idea of each
encoding's overhead. The striking point is the lightness of the monomorphic
encodings, as witnessed by the number of symbols.
Because monomorphisation generates several copies of the same formulae, we could have
expected it to lead to larger problems, but this underestimates the cost of
encoding types as terms in the polymorphic encodings.%
\footnote{The increase visible in the \erased{} column, from 89 to 125
clauses and from 1276 to 1757 symbols, is due to the multiple copies arising
from monomorphisation. Even though these become identical after type erasure,
they are counted as separate in the statistics, and it is up to the automatic
provers to notice that they are the same.}
The strong correlation between the success rates in
Figure~\ref{fig:number-of-solved-problems-with-500-facts} and the
average number of symbols in
Figure~\ref{fig:average-size-of-clausified-problems-with-50-facts}
confirms the expectation that clutter (whether type arguments, tags, or guards)
slows down automatic provers.

%
%
%

Independently of these empirical results, the new type encodings
made an impact at the 2012 edition of CASC, the annual
automatic prover competition \cite{sutcliffe-2012-casc-j6}.
Isabelle%
's automatic proof tools, including Sledgehammer, compete
against the automatic provers LEO-II, Satallax, and TPS
in the higher-order division. Largely thanks to the new schemes (but also to
improvements in the underlying first-order provers), Isabelle moved from the
third place it had occupied since 2009 to the first place.

\section{Related Work}
\label{sec:related-work}


The earliest descriptions of type tags and type guards we are aware of are due
to Enderton \cite[\S4.3]{enderton-1972} and Stickel \cite[p.~99]{stickel-1986}.
Wick and McCune \cite[\S4]{wick-mccune-1989}
compare type arguments, tags, and guards in a monomorphic setting. Type
arguments are 
described by Meng and Paulson
\cite{meng-paulson-2008-trans}, who also consider full type erasure and
polymorphic type tags and present a translation of axiomatic type classes.
%
As part of the MPTP project,
Urban \cite{urban-2006} extended the untyped TPTP FOF syntax with dependent
types to accommodate Mizar and designed translations to plain FOF\@.

The intermediate verification language and tool Boogie 2 \cite{leino-ruemmer-2010}
supports a restricted form of higher-rank polymorphism (with polymorphic maps), and
Why3 \cite{bobot-et-al-2011}
provides rank-1 polymorphism. Both define translations to a monomorphic logic and
rely on proxies to handle interpreted types
\cite{couchot-lescuyer-2007,leino-ruemmer-2010,bobot-paskevich-2011}.
One of the Boogie translations \cite[\S3.1]{leino-ruemmer-2010}
uses SMT triggers to prevent ill-typed instantiations
in conjunction with type arguments; however, this approach is risky in the
absence of a semantics for triggers.
Bouillaguet et al.~\cite[\S4]{bouillaguet-et-al-2007} showed that full type
erasure is sound if all types can be assumed to have the same cardinality and
exploit this in the verification system Jahob.

An alternative to encoding polymorphic types
or monomorphising them away is to
support them natively in the prover.
This is ubiquitous in interactive theorem provers, but perhaps
the only production-quality automatic prover that supports polymorphism is Alt-Ergo
\cite{bobot-et-al-2008}.

Blanchette and Krauss \cite{blanchette-krauss-2011} studied monotonicity inferences
for higher-order logic without polymorphism. Claessen et al.\ \cite{claessen-et-al-2011}
were first to apply them to type erasure.

\section{Conclusion}
\label{sec:conclusion}

This \Paper{} introduced a family of translations from polymorphic into untyped
first-order logic, with a focus on efficiency. Our monotonicity-based
encodings soundly erase all
types that are inferred monotonic, as well as most occurrences of the
remaining types. The best translations outperform the traditional
encoding schemes.

We implemented the new translations in the Sledgehammer tool for Isabelle\slash
HOL,
 thereby addressing a recurring user complaint. Although Isabelle certifies
external proofs, unsound proofs are annoying and often conceal sound proofs.
The same translation module forms the core of Isabelle's TPTP exporter
tool, which makes entire theorem libraries available to first-order reasoners.
Our refinements to the monomorphic case have made their way into
the Monotonox translator \cite{claessen-et-al-2011}. Applications such as Boogie
\cite{leino-ruemmer-2010}, LEO-II \cite{benzmueller-et-al-2008},
and Why3 \cite{bobot-et-al-2011} also stand to gain from lighter encodings.



The TPTP family recently welcomed the addition of TFF1
\cite{blanchette-paskevich-2013}, an extension of the monomorphic TFF0 logic
with rank-1 polymorphism. Equipped with a concrete
syntax and translation tools, we
can turn any popular automatic theorem prover into an efficient polymorphic
prover. Translating the untyped proof back into a typed proof is usually
straightforward, but there are important corner
cases that call for more research.
It should also be possible to extend our
approach to interpreted arithmetic.

From both a conceptual and an implementation point of view, the
encodings are all instances of a
general framework, in which mostly orthogonal features can be combined in
various ways. Defining such a large number of encodings makes it possible to select the
most appropriate scheme for each automatic prover, based on empirical evidence. In fact,
using strategy scheduling or parallelism, it is advantageous to have each prover employ a
combination of encodings with complementary strengths.




\begin{acknowledgements}
Koen Claessen and Tobias Nipkow made this collaboration possible.
Lukas Bulwahn, Peter Lammich, Rustan Leino, Tobias
Nipkow, Nir Piterman, Alexander Steen, Mark Summerfield, Tjark Weber, and several anonymous reviewers suggested
dozens of textual improvements; in particular, the three reviewers of this extended,
journal version accomplished a remarkable work, going well beyond the call of duty.
Blanchette was partly supported by the Deutsche
Forschungs\-gemein\-schaft (DFG) projects \relax{Quis Custodiet} and
\relax{Hardening the Hammer} (grants NI~491\slash 11-2 and NI~491\slash 14-1).
Popescu was partly supported by the DFG project
\relax{Security Type Systems and Deduction} (grant NI~491\slash 13-2) as part
of the program \relax{Reliably Secure Software Systems} (RS\textsuperscript{3},
priority program 1496). The authors are listed alphabetically.
\end{acknowledgements}

\bibliography{bib}{}


    \insert\copyins{\hsize.57\textwidth
\vbox to 0pt{\vskip12 pt%
      \fontsize{6}{7 pt}\normalfont\upshape
      \everypar{}%
      \noindent\fontencoding{T1}%
  \textsf{This work is licensed under the Creative Commons
  Attribution-NoDerivs License. To view a copy of this license, visit
  \texttt{http://creativecommons.org/licenses/by-nd/2.0/} or send a
  letter to Creative Commons, 171 Second St, Suite 300, San Francisco,
    CA 94105, USA, or Eisenacher Strasse 2, 10777 Berlin, Germany}\vss}
      \par
      }%

\end{document}